\documentclass[11pt]{article}
\usepackage{natbib}
\usepackage{amsmath}
\usepackage{amsthm}
\usepackage{amssymb}
\usepackage{algorithm}
\usepackage{graphicx}
\usepackage{sidecap}
\usepackage{subcaption}
\usepackage{caption}
\usepackage{comment}
\usepackage{float}
\usepackage{tikz}
\usepackage{pgf, pgfplots}
\usepackage{tikz-network}
\pgfplotsset{compat=1.17} 
\usepackage{xcolor}
\usetikzlibrary{arrows}
\usetikzlibrary{calc}
\usetikzlibrary{positioning}
\usetikzlibrary{automata}
\usetikzlibrary{shapes}
\usetikzlibrary{shapes.geometric}
\usetikzlibrary{decorations.markings}
\usepackage{setspace}
\usepackage{fullpage}
\usepackage{pdfpages}
\usepackage{import}
\usepackage{rotating}
\usepackage{hyperref}
\usepackage{cleveref} 
\usepackage{minitoc}
\usepackage{longtable, booktabs}
\usepackage{array}
\usepackage{algpseudocode} 
\usepackage{enumerate}

\newtheorem*{thm*}{Theorem}
\newtheorem{prop}{Proposition}
\newtheorem*{prop*}{Proposition} 
     
\newtheorem{lem}{Lemma}
\newtheorem*{lem*}{Lemma}

\newtheorem{rem}{Remark}
\newtheorem*{rem*}{Remark}
 
\newtheorem*{ex*}{Example} 
\newtheorem{ass}{Assumption}

\definecolor{green}{HTML}{51A351}
\definecolor{blue}{HTML}{2F96b4}
\definecolor{gray}{HTML}{BCBCBC}

\DeclareMathOperator*{\argmax}{arg\max}

\newcommand{\monthyeardate}{\ifcase \month \or January\or February\or March\or April\or May \or June\or July\or August\or September\or October\or November\or December\fi, \number \year}
\title{Mobile but Miserable: Social comparison networks and social mobility\footnote{We are grateful to Guglielmo Secchi for excellent research assistance and to audiences at BiNoMa, the Barcelona Summer Forum, the University of Bath, and the Network Science in Economics conference for helpful comments. Any remaining errors are the sole responsibility of the authors. }}

\author{Patrick Allmis\footnote{University of Cambridge (email: \texttt{pa509@cam.ac.uk})} \and Matthew Elliott\footnote{University of Cambridge (email: \texttt{matthew.l.elliott@gmail.com})} \and Christian Ghiglino\footnote{University of Essex (email: \texttt{cghig@essex.ac.uk})} \and Alastair Langtry\footnote{University of Bristol (email: \texttt{alastair.langtry@bristol.ac.uk})}}

\date{\today} 
\begin{document}
\maketitle
\begin{abstract}
Although social mobility is widely viewed as desirable, we document an `inverse-U' shaped relationship with well-being: mobility is positively associated with well-being when low, but the association diminishes as mobility increases and eventually becomes negative. To explain this, we build a model in which agents are embedded in a social comparison network and suffer from falling behind their peers. The model highlights two channels that can cause welfare to be decreasing in mobility at high levels. First, when mobility is high, further increases widen already significant gaps between agents from similar parental backgrounds but of different abilities. Second, when mobility is high, high-ability agents are already relatively well off. So further increases in social mobility increase inequality. We find evidence supportive of both channels in US data. The model also predicts, and data supports, that social integration -- where relationships form more along ability than background lines—allows societies to benefit from mobility for longer. We interpret this as social mobility being multidimensional with both economic and social dimensions that are complementary. 
\end{abstract}

\newpage
Social mobility allows people to transcend the circumstances of their birth and fulfill their potential. Across the political spectrum, politicians want more of it: \emph{``We need a far more socially mobile country''} (David Cameron, 2013); \emph{``the combined trends of increasing inequality and decreasing mobility pose a fundamental threat to the American Dream, our way of life''} (Barack Obama, 2013); \emph{``[Singapore] must not and will not let up on maintaining social mobility''} (Lee Hsien Loong, 2018).
And public debate typically views social mobility as an unalloyed good. Given this, a casual observer might be forgiven for assuming that greater social mobility consistently comes with greater overall well-being. But they would be wrong.

In the United States, there \emph{is} a positive association between social mobility and well-being in the least socially mobile third of counties (corr $= 0.27$). But this relationship goes into reverse in the most socially mobile third of counties (corr $= -0.08$). Among the most socially mobile counties, greater mobility is associated with \emph{reduced} well-being. \Cref{fig:1} plots the broader relationship for US counties.\footnote{In \Cref{sec:data} we explore this relationship more rigorously. It is robustly statistically significant. It is based on the widely used measure of social mobility from the \emph{Opportunity Atlas} \citep{chetty2018impacts, chetty2026opportunity}, and a county-level measure of well-being from the \emph{Human Flourishing Geographical Index} \citep{iacus2025human}.} 

\begin{figure}[h!]
    \centering
    \includegraphics[width=0.99\linewidth]{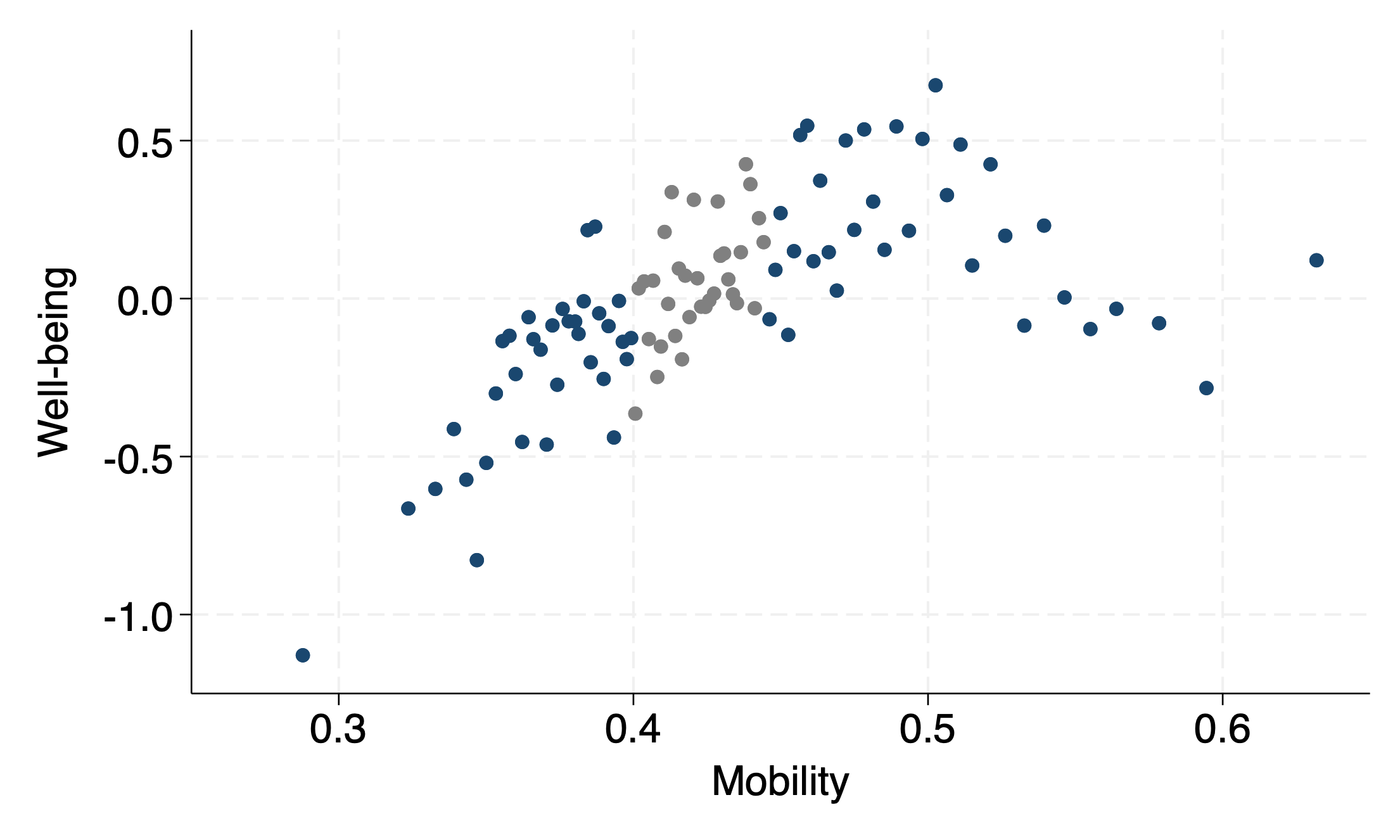}
    \caption{Average well-being and social mobility for US counties, binned scatter plot. Top and bottom terciles in blue. Middle tercile in gray.}
    \label{fig:1}
\end{figure}

To better understand this `inverse-U' shaped relationship between social mobility and well-being, we build a parsimonious theoretical model of networked social comparisons. We use this model to explore the welfare consequences of social mobility and how this interacts with the networked structure of society. Our theory matches this key feature -- well-being is increasing in mobility for low levels of mobility and decreasing in mobility for high levels of mobility (in an absue of terminology, we will henceforth call this an `inverse-U'). It makes additional predictions that we find support for in the data. While an impressive body of work on social mobility has focused on how to measure it \citep{chetty2014land}, what drives it \citep{chetty_nature1,chetty_nature2,annurev}, and on understanding what can be done to increase it \citep{chetty2026creating}, there is very little prior work in economics asking whether more social mobility is actually good for welfare.

An important feature of social mobility is that it applies to whole societies. And changing it has different effects on different people. This matters because there is significant evidence that people compare themselves to others: they evaluate their outcomes in comparison to those around them, not in absolute terms. Changing social mobility therefore affects social comparisons between people, ameliorating some negative comparisons but aggravating others. So the welfare implications of greater social mobility depend on, among other things, the net effect of these changes to the social comparisons.

Within this, exactly who an individual compares herself to -- her neighbors in a network of social comparisons -- will affect which comparisons carry weight. In turn, this influences whether, and how much, she gains or loses from changes in social mobility. At the aggregate level, the structure of the social network will then influence the overall impact of increasing social mobility.

Additionally, it is common for politicians and commentators to focus on how greater social mobility helps people \emph{rise up} to fulfill their potential, and how this can lower inequality. But this neglects the other side of the coin: greater social mobility must also make people with little potential (but who were previously buoyed by the circumstances of their birth) \emph{downwardly} mobile. And it can also increase inequality, as those with high ability -- strongly rewarded in a society with high social mobility -- pull away from those with low ability. This is another force that can create downsides to social mobility -- both on its own and as it interacts with the structure of communities.

Taking these key ideas, we build a model where people choose how much to consume, and when making this choice account for social comparisons with each of their friends in addition to the usual intrinsic costs and benefits. A social network captures who their friends are. A person's `earning potential' -- which captures how easily they can earn -- depends on their family background (the circumstances of their birth) and their ability (or potential). The level of social mobility determines the weight placed on each factor.

Our paper focuses on the welfare implications of social mobility. So we consider a general class of social welfare functions: ones that strictly value individual well-being, and weakly value reductions in inequality in well-being via transfers from better-off agents to worse-off agents (Pigou-Dalton transfers). To keep things simple, we assume there are only two levels of background and two levels of ability, and that people with the same background-ability pair have similar patterns of network connections. This creates clearly delineated ``social classes'', which is useful for aligning with discussions of class in public debate. People only differ in their background and ability (and hence earning potential), and network connections.

Our headline result is that for all social welfare functions in our broad class, there is an `inverse-U' shaped relationship between welfare and social mobility. When social mobility is initially low, increasing it unambiguously increases welfare. But when social mobility is initially high, increasing it unambiguously reduces welfare. This rationalizes the empirical relationship between well-being and social mobility that we document in \Cref{fig:1}. 

The core intuition for this result comes from the way social mobility affects social comparisons. Greater social mobility reduces the gap between groups of the same ability, narrowing those social comparisons. But it widens the gap between groups with the same parental background, aggravating those comparisons. When social mobility is initially low, the within-ability comparisons are large and the within-background comparisons are small. So the benefits of ameliorating the former more than outweigh the costs of aggravating the latter. Additionally, increasing social mobility from a low base brings different groups closer together in terms of their economic opportunities, reducing inequality. And when social mobility is initially high, these effects reverse -- both for the total impact of changing social comparisons and for inequality. So regardless of whether social mobility is high or low, both forces push in the same direction -- leading to our unambiguous welfare result.

Another key insight from our model is that when it comes to welfare effects, there is an important complementarity between the level of social mobility and the level of integration between those from more and less advantaged backgrounds. We find that when social mobility is high, more integration is unambiguously good for welfare. But when social mobility is low, more integration is unambiguously bad for welfare. And going in the other direction, when social mobility is at an intermediate level, whether more of it is good or bad for welfare depends on the level of integration. In sum, for both social mobility and integration, the welfare effects of increasing one of them can depend on the prevailing level of the other.

This has an important implication for policymakers: in isolation, pursuing policies that make people's economic outcomes more dependent on their ability relative to their background may be insufficient. In fact, pursued in isolation, they could be worse than doing nothing at all. To make them beneficial for society, government should also integrate communities, allowing people's friendships -- not just their earning potential -- to transcend the circumstances of their birth. That is, people need to be able to meet and connect with others of similar abilities, and must not remain stuck in the social circles they were born into. 

To be clear, this does not cast doubt on greater social mobility being a worthwhile goal for government and society. Rather, it highlights that a view of social mobility that considers only economic outcomes is too narrow. And such a blinkered view -- that neglects the structure of communities that people live in, where they actually experience those economic outcomes -- is what can create unintended harms.

\paragraph{Education policy in our model.} Education policy reforms -- from widening access to high-quality teaching, to providing better careers advice, work placements or soft skills training, and even early years childcare (through, for example, the US's \emph{Head Start} or the UK's \emph{Sure Start} programs) -- are typically viewed through the lens of the economic aspect of social mobility. They seek to reduce the influence that a child's parental background has on their opportunities to succeed both academically and in the workplace. 

A core and well-documented feature of social networks is homophily -- the tendency for people to form connections disproportionately with people from similar backgrounds \citep{mcpherson2001birds}. Education policies can also affect network structure, and hence homophily, by influencing who children spend time with and who they form connections with. For example, the structure of school catchment areas, and the treatment of private schooling, influences the within-school diversity of parental background. And in a given school, uniform policies and the use of paid-for extra-curricular activities can affect the extent to which students are able to form and maintain connections with those of different parental backgrounds \citep{putnam2016our}.\footnote{In \emph{Our Kids}, \citet{putnam2016our} argues forcefully about the harms of paid-for extra-curricular activities, including the segregation between children of different parental backgrounds it causes. Additionally, increasing homophily in terms of ability would provide similar benefits. However, we contend that it is difficult in practice to identify ability early on -- free from the effects of parental investments. Seemingly natural policies to sort children by `ability' may backfire, as they risk splitting up children along background lines \citep{atkinson2006result}.} 
Many of these friendships persist long into adulthood \citep{yougov2021}, so education policy can meaningfully affect the extent to which people form connections with those from different backgrounds.

Therefore our paper suggests that if education reforms are to improve \emph{welfare}, then they must consider both the economic and the social dimensions together. They must make schools `good' -- able to support students to achieve their potential. And they must integrate schools -- mixing students from different parental backgrounds. It is not enough to make all schools good schools if they remain segregated by parental background. Moreover, our results demonstrate that the more education reform succeeds on the economic dimension of making schools `good', the more important it becomes that it tackles the social dimension too.

\paragraph{A Roadmap.} The rest of the paper is organized in the usual way. \Cref{sec:literature} discusses related literature. \Cref{sec:model} sets out the model, and \Cref{sec:model_microfounded} provides a richer microfoundation. \Cref{sec:behavior} characterizes agents' behavior, and examines how it responds to changes in primitives of the model. \Cref{sec:welfare} provides the main results on the welfare effects of economic mobility and its complementarity with integration. \Cref{sec:data} details the empirical relationship between social mobility and well-being. \Cref{sec:abs_mobility} illustrates our mechanism through a discussion of high-mobility counties. \Cref{sec:discussion} provides concluding discussion, including how economic mobility can be reinterpreted as the extent of meritocracy.

\section{Related Literature}\label{sec:literature} 

This paper's core contribution is to examine the welfare effects of social mobility, and to show how the presence of social comparisons can create an `inverse-U' shaped relationship between mobility and welfare. Further, we distinguish between two dimensions of social mobility: ``economic mobility'' (fluidity across income levels) and ``social integration'' (fluidity across friendships). Our main thesis is that there is an important complementarity between the two. This paper therefore contributes to two distinct strands of literature, one concerned with social mobility and another concerned with social comparisons. 

\paragraph{Social mobility.} Within economics, existing literature focuses on the economic dimension of mobility, and in particular intergenerational mobility across socioeconomic groups. A rich body of work explores factors that shape economic mobility \citep{britishmobility, chetty2020race,chetty_nature1,chetty_nature2}. This work has gained significant traction among politicians and policymakers \citep{SMC2022, SMC2023, SMC2024, SMC2025}. Perhaps closest to us here is the work that examines the link between inequality and social mobility. A general finding from this literature is that higher inequality results in lower economic mobility in the future \citep{andrews2009more, corak2013income, durlauf2018understanding}. We do not contribute directly to this strand of literature. Rather, we examine the welfare effects of social mobility -- something not typically considered by work in economics. In doing so we seek to question whether social mobility is always good for welfare.

Unlike economics, significant work in other social sciences and medicine considers the implications of social mobility for well-being. Psychology has considered the impact of upheaval and strains associated with social mobility \citep{ellis1967social,jetten2008individual,daenekindt2017experience}. Work in sociology broadly finds benefits to upward mobility and costs to downward mobility \citep{chan2018social,gugushvili2019falling, prag2022intragenerational}.
There is significant work in medicine on the `social determinants of health' \citep{marmot1991health}. An important finding from this work is that there are meaningful health benefits from higher socioeconomic position \citep{marmot2013health} and from upward mobility \citep{vable2019possible}.\footnote{\cite{marmot2013health} is the canonical study. But this literature is both broad and deep. \cite{wilkinson2003social, marmot2005status} and \cite{ODPHP_SDOH} provide entry-points, but a review of this literature is beyond the scope of our paper.}

Work in these fields typically stops short of considering the overall impact of increased social mobility on well-being for a whole society, but \cite{fischer2009} is an exception.
They use data from the World Values Survey to study the overall impact on societies of both perceived and actual social mobility. They find that it is hard to disentangle the impact of social mobility on well-being that occurs through two channels: first the impact of better opportunities (which is beneficial) and second the impact of increased inequality (which is detrimental). However, in the case of high actual social mobility, the data indicates that the negative impact of the associated inequality dominates.  We contribute to this body of work by providing a theoretical model consistent with these findings (where upward mobility is good, but downward mobility is bad and where inequality plays a central role) to study welfare in the aggregate.

\paragraph{Social comparisons.} The idea that people compare themselves to others, and that their well-being is partly determined by these social comparisons, dates at least to \cite{veblen1899thetheory} and is now well-established in economics \citep{luttmer2005neighbors,senik2009direct}. Here, we take the view that these ``others'' are people's friends, and use a network to model the pattern of these friendships.\footnote{A large earlier literature assumed people compare themselves to a global average \citep{duesenberry1949income, abel1990asset, campbell1999force, ljungqvist2000tax}, or to their ordinal rank \citep{frank1985demand, hopkins2004running}.} Building on the network games literature -- which allows for richer patterns of social comparisons (see \cite{bramoulle2016games} and \cite{jackson2015games} for reviews of
that literature) -- recent work studies social comparisons from an economics perspective \citep{ghiglino2010keeping,immorlica2017social, ushchev2020social,lopez2021far}.
We differ from this prior work by asking a new question: about the implications of changing social mobility. Our focus is on implications for welfare and on understanding how the economic and social dimensions of social mobility operate on their own terms, and how they interact with one another.

In other social sciences, a central role for social comparisons is often the default position. See, for example, \cite{festinger1954theory, wood1989theory} and \cite{collins1996better} in psychology, and \cite{merton1950contributions, merton1968social} and \cite{runciman1966relative} in sociology.
More recent work has argued that the link between income inequality and adverse population-level outcomes may run in part through social comparison and the status anxiety it generates \citep{wilkinson2009income, layte2012association, layte2014who}. A common thread across these literatures is that comparisons to peers carry more weight than comparisons to a society-wide average, which motivates our use of a network to capture who compares to whom.

\section{Model}\label{sec:model} \paragraph{Agents and endowments.} There are $n$ agents such that $\frac{1}{4} n$ is an integer, with typical agents $i,j \in N$. Each agent is endowed with either high or low \emph{ability} $a_i \in \{H,L\}$, and a rich or poor (parental) \emph{background} $b_i \in \{R, P\}$. Agents with a given $(a,b)$ pair form a \emph{class}. We assume all classes are the same size. Let $q$ denote a typical class, and $Q = \{HR,LR,HP,LP\}$ denote the set of classes. 

We assume that: (i) $\frac{d y_{HP}}{d s} = -\frac{d y_{LR}}{d s} > 0$ and $\frac{d y_{HR}}{d s} = \frac{d y_{LP}}{d s} = 0$, and (ii)
$y(H,b_i,s)>y(L,b_i,s)$ for $s\in (0,1)$, with $\lim_{s\to 0} y(H,b_i,s)-y(L,b_i,s)=0$, and $y(a_i,R,s)>y(a_i,P,s)$ for $s \in(0,1)$, with $\lim_{s\to 1} y(a_i,R,s)-y(a_i,P,s)=0$. Intuitively, part (i) says that increasing economic mobility has equal and opposite effects on the $HP$ and $LR$ classes. The $HP$ class gain from higher mobility. The $LR$ class face offsetting losses. Part (ii) then says that in the case of `perfect mobility' \emph{only} ability matters, and in the case of `perfect immobility' only background matters for earning potential. 

\paragraph{Network.} Agents are embedded in a weighted and directed \emph{network}, represented by an $n \times n$ adjacency matrix $G$. The entry $G_{ij} \geq 0$ denotes the strength of a link from $i$ to $j$. An agent $j$ is $i$'s \emph{friend} if and only if $G_{ij} > 0$. By convention, we assume $G_{ii} = 0$ for all $i$. We assume that all agents have a common \emph{degree} $d>0$: that is, $\sum_{j}G_{ij}=d$ for all $i$.\footnote{This assumption takes a benchmark view that all agents have the same overall strength of social comparisons. Unless explicitly stated otherwise, all sums over $i$ are over all agents $i \in N$. In two abuses of terminology, we will refer to $d$ as `degree' (it is formally an `out-degree'), and to $G$ as `the network'.} 

We assume that agents form a fraction $\omega$ of their links (in a weighted sense) to those in the same class, a fraction $\rho \ (1-\omega)$ of their links to those of the other class who share their ability, and a fraction $(1-\rho) \ (1-\omega)$ of their links to those of the other class who share their background. The parameter $\omega \in (0,1)$ captures own-class connections and $\rho \in (0,1)$ captures the extent of \emph{integration} -- which is the level of homophily along the dimension of ability compared to the level of homophily along the dimension of background.\footnote{We contend that a natural understanding of a non-integrated (i.e., segregated) society is one where connections are determined by parental background. So in our model, higher integration corresponds to parental background playing less of a role in determining who your friends are. Homophily is the tendency to interact predominantly with similar others \citep{mcpherson2001birds,currarini2009economic}. We provide microfoundations for this structure of links in Online Appendix \ref{OA:network_microfoundations}. The headline idea there is that agents meet at random, and then form friendships based on similarity along ability and background lines.}

Formally, assume that $\sum_{j \in q'} G_{ij} = d \ g_{q q'}$ for all $i \in q$ and for all $q' \in Q$, where:

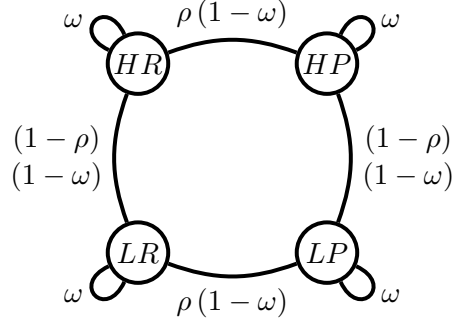
\begin{figure}[!htbp]
\centering
\begin{minipage}[t]{0.38\textwidth}
\hspace{-20mm}
\vspace{0pt} 
\begin{align*}
g_{qq'} =
\begin{cases}
    \omega \quad &\text{ if } q = q' \\  
\\
    \rho \ (1 - \omega) \quad &\text{ if } q \neq q', \, a_q = a_{q'} \\
\\
    (1 - \rho) \ (1 - \omega) \quad &\text{ if } q \neq q', \, b_q = b_{q'} \\
\\
    0 \quad &\text{ otherwise. }
\end{cases} \end{align*}
\end{minipage}
\hfill 
\begin{minipage}[t]{0.48\textwidth}
\centering
\vspace{0mm}
\begin{tikzpicture}[line width=1.5pt, node/.style={circle, draw, minimum size=8mm, inner sep=0pt, line width=1.5pt } ]
\node[node] (1) at (0,0) {$LR$};
\node[node] (2) at (2.5,0) {$LP$};
\node[node] (3) at (2.5,2.5) {$HP$};
\node[node] (4) at (0,2.5) {$HR$};
\draw (1) to[out=205, in=245, looseness=5]
  node[pos=0.5, left, align=center] {$\omega$} (1);

\draw (2) to[out=295, in=335, looseness=5]
  node[pos=0.5, right, align=center] {$\omega$} (2);

\draw (3) to[out=25, in=65, looseness=5]
  node[pos=0.5, right, align=center] {$\omega$} (3);

\draw (4) to[out=115, in=155, looseness=5]
  node[pos=0.5, left, align=center] {$\omega$} (4);

\draw[bend right=20]  (1) to node[midway, below] {$\rho \, (1-\omega)$} (2);
\draw[bend right=20]  (2) to node[midway, right] {\shortstack{$(1-\rho)$ \\ $ (1-\omega)$}} (3);
\draw[bend right=20]  (3) to node[midway, above] {$\rho \, (1-\omega)$} (4);
\draw[bend right=20]  (4) to node[midway, left]  {\shortstack{$(1-\rho)$ \\ $ (1-\omega)$}} (1);
\end{tikzpicture}
\end{minipage}
\caption{Proportion of links between different classes, and illustrative class-level network.}\label{fig:network_structure}
\end{figure}

\paragraph{Actions and Timing.} All agents simultaneously choose an amount of consumption, $x_{i} \in \mathbb{R}^+$.

\paragraph{Preferences.} Agent $i$ derives an \emph{intrinsic benefit} $u(x_i, y_i)$ from consuming, which depends on her own earning potential, $y_i$. She also makes social comparisons with each of her friends, which impose convex costs from consuming less than each friend. So $i$'s preferences are captured by:
\begin{align}\label{eq:preferences}
    U_i = \underbrace{u(x_i , y_i)}_{\text{intrinsic benefit}} - \underbrace{\sum_{j: x_i \leq x_{j}} G_{ij} \ F(x_j - x_i)}_{\text{social comparison costs}},
\end{align}
where $F(\cdot)$ is smooth, strictly increasing and strictly convex, with $F(0) = F'(0) = 0$. The function $F(\cdot)$ captures the costs of social comparisons, and the link weights $G_{ij}$ capture how much weight she places on each comparison. We assume that all agents have preferences that can be represented by the same utility function. 

We take $u(\cdot,\cdot)$ to be twice continuously differentiable, and write $x^a(y_i) := \argmax_{x \geq 0} u(x, y_i)$ for the level of consumption an agent with earning potential $y_i$ would choose in the absence of any social comparisons. Throughout, subscripts on $u$ denote partial derivatives. We make the following assumptions.

\begin{ass}\label{ass:u_fn}
For all $x_i \geq 0$ and all $y_i$:
\emph{(i)} $u_{xx}(x_i,y_i) < 0$, and $x^a(y_i)$ is interior;
\emph{(ii)} $u_y(x_i,y_i) > 0$;
\emph{(iii)} $u_{xy}(x_i,y_i) > 0$ for $x_i > 0$; and
\emph{(iv)} $V(y_i) := \max_{x \geq 0} u(x,y_i)$ is strictly concave.
\end{ass}

Intuitively, this assumes that the marginal benefit of consumption starts off positive, but consistently falls and eventually turns negative (part \emph{(i)}). And also that a higher earning potential makes an agent better off at any level of consumption and complements consumption (parts \emph{(ii)} and \emph{(iii)}). The complementarity assumption captures the idea that higher earning potential reduces the amount an agent must work to buy a given amount of consumption. If there are diminishing returns to leisure, then this reduction in work is more valuable when an agent consumes more -- i.e., when she has less leisure.\footnote{To see this, suppose agents can earn money at a constant wage $y$ and have a fixed endowment of time normalized to $1$, so their leisure time is $\ell = 1 - \tfrac{x}{y}$. Then consider an underlying utility function over consumption and leisure, $u(x,y)=\nu(x,1-\tfrac{x}{y})$. Given this, $u_{xy}(x,y) = \tfrac{x}{y^2} \nu_{x \ell}+ \tfrac{1}{y^2} \nu_\ell - \tfrac{x}{y^3} \nu_{\ell \ell}$. Assuming that utility is increasing and concave in leisure, $u_{xy}>0$ requires that consumption and leisure are not too strong substitutes for one another.}
Part \emph{(iv)} ensures there are diminishing returns to earning potential in terms of the direct benefits. 

\paragraph{Information, Solution Concept.} The network, preferences, and the structure of the game are all common knowledge. We look for Nash equilibria of the game.

\paragraph{A richer model.} This model where agents choose a level of consumption is in fact a reduced form of a richer model where agents make separate decisions for whether or not to consume each of a continuum of indivisible goods. In \Cref{sec:model_microfounded}, we set out this richer model and show how it collapses back to the one presented here.

\subsection{Discussion}\label{sec:model_discussion}
This is a model of social comparisons where the benefit of consumption depends on earning potential (itself determined by agents' ability and background, and economic mobility). While there is only one good in this model -- which attracts social comparisons -- we can view the declining marginal benefit of consuming as capturing the opportunity cost of both leisure and consuming (unmodeled) goods that do not attract social comparisons. We also assume that people only make `upwards' comparisons -- they experience losses from comparisons with friends who consume more, but not gains from comparisons with friends who consume less. This is motivated by significant empirical evidence that upward comparisons weigh much more heavily (e.g., \cite{ferrer2005income, vendrik2007happiness, leites2022effect, bertrand2016trickle}).

The key technical feature of our model is that agents make social comparisons \emph{friend-by-friend}. Lower actions by one friend are not a perfect substitute for higher actions by another. This is important because changing economic mobility will increase earning potential (and therefore, in equilibrium, consumption) for some agents and will reduce it for others. If agents make social comparisons with some hypothetical `average friend' (rather than friend-by-friend), then the countervailing effects on different friends could `wash out', leaving an agent's `average friend' unchanged. This would miss the effects of changing economic mobility.

It is also important to note that economic mobility and integration capture how earning potentials and friendships evolve from generation to generation. So the static model we present here should be viewed as a simple way of capturing this evolution from one generation to the next. There are two main reasons to focus on a single generation-to-generation transition, as we do here, rather than a fuller dynamic setting with many generations. First, a single transition takes decades, and so is relevant for policy in its own right -- even if it misses features of very long run dynamics. Second, available data only provides information on mobility and integration for a single generation-to-generation transition.

\section{Indivisible Goods: A Microfoundation}\label{sec:model_microfounded} Here, we provide a richer model where agents choose whether or not to consume each of a continuum of indivisible goods. We show that it has a unique class-symmetric equilibrium, which is identical to the unique equilibrium of the simpler model we set out in \Cref{sec:model}. Moreover, myopic best response dynamics starting from any initial profile of cut-off strategies (including those that are not class-symmetric) uniformly converge to that same equilibrium. As such, we micro-found the model from \Cref{sec:model} as the reduced form of this richer model. Skipping this section will therefore not affect understanding of any later sections. 

\subsection{Model} 
The agents, endowments, network, timing, and information are unchanged compared to \Cref{sec:model}. The key change is to the actions. Here, there is a continuum of goods, $m \in [0,A]$, for sufficiently large $A< \infty$. Each good $m$ is indivisible, and its cost is $c_m = m^\alpha$, with $\alpha>0$. All agents simultaneously choose whether or not to consume each good: $x_{im} \in \{0,1\}$ for each $m$. Write $X_i := \{ m \in [0,A] : x_{im} = 1 \}$ for the set of goods that $i$ consumes. We require $X_i$ to be measurable, so that the mass, expenditure, and social comparison terms below are all well defined.

Preferences are analogous to \Cref{sec:model}. But for convenience, we rewrite the intrinsic benefit as a difference between a benefit and a cost: $u(x,y) = v x - \kappa(x,y)$, with $v>0$ (this is without loss). Then we assume that the benefits accrue to the mass of goods consumed, while the costs depend on the expenditure on those goods. So the increasing cost of individual goods $m$ applies only to the `costs part' of this intrinsic benefit function. This makes cheaper goods more attractive, and creates a natural ordering of goods from cheaper to more expensive. Formally, we define:
\begin{align*}
    x_i \ := \ \int_0^A x_{im}\,dm
    \qquad \text{and} \qquad
    \hat{X}_i \ := \ \left[ (1+\alpha) \int_0^A m^{\alpha} x_{im}\,dm \right]^{\tfrac{1}{1+\alpha}} ,
\end{align*}
the mass of goods that $i$ consumes and her expenditure aggregate. So the analogy to \Cref{eq:preferences} is now:
\begin{align}\label{eq:preferences_microfounded}
U_i \ = \ \underbrace{v \, x_i \ - \ \kappa\big( \hat{X}_i , y_i \big)}_{\text{intrinsic benefit}}
\ - \ \underbrace{\sum_{j\neq i} G_{ij} \, F\left( \int_0^A \max\{(x_{jm}-x_{im}),0\}\,dm \right)}_{\text{social comparison costs}} ,
\end{align}
where $G_{ij}$ and $F(\cdot)$ are as in \Cref{sec:model}, the constant $v>0$ is the marginal benefit of an additional good, and $\kappa(\cdot,\cdot)$ is the cost of expenditure. The information structure is as in \Cref{sec:model}, and we again look for Nash equilibria of the game.

\subsection{Equilibrium}
As a preview of our first remark in the next section, the game in \Cref{sec:model} has a unique Nash equilibrium. Let $x_i^*$ denote $i$'s level of consumption in that equilibrium. This equilibrium is class-symmetric. We say that a profile is \emph{class-symmetric} whenever agents of the same class make the same choices.\footnote{In the simpler model, this is $x_i = x_{q(i)}$ for all $i \in q$ and for all $q$. In the model presented in this section, this is $X_i = X_{q(i)}$ for all $i \in q$ and for all $q$.} With this notation in place, we can state the main result of this section.

\begin{prop}
    \label{prop:class_symmetric}
    There is a unique class-symmetric Nash equilibrium: $X_i^* = \big[ 0 , x_i^* \big]$ for all $i$.
\end{prop}

When we focus on class-symmetric profiles, this result shows that the `richer' model presented here yields identical behavior to that of the simpler model in \Cref{sec:model}. Another way to see how the `richer' model maps back to the model from \Cref{sec:model} is to consider the myopic best response dynamics. Starting from any profile of cut-off strategies, those dynamics converge to the same equilibrium, $X_i^* = \big[ 0 , x_i^* \big]$ for all $i$. We now state this formally.

\begin{rem}\label{rem:cutoff_convergence}
    Start from any profile of cut-off strategies, $X^0_i = [0 , x^0_i]$, where $x^0_i \in [0,A]$ for all $i$ may differ across agents in the same class. And let all agents best respond myopically.\footnote{Agents need not do so simultaneously. The same conclusion holds if, at each date, an arbitrary non-empty set of agents best responds while the rest stand still, so long as every agent revises infinitely often. This is because the contraction established in the proof is in the supremum norm.} Then every $X^t$ is a profile of cut-off strategies, $X^t_i = [0 , x^t_i]$, and $x_i^t$ converges uniformly to $x^*_i$ for all $i$. 
\end{rem}

This happens because the best response to a cut-off strategy is itself a cut-off strategy, and the best response map is a contraction in the supremum norm.

\section{Behavior}\label{sec:behavior} 
Absent social comparisons, consumption is purely private: there are no strategic interactions and agents simply consume their ideal level, $x_i^* = x^a(y_i)$. Social comparisons then impose costs on agents when they consume less than their friends. This creates pressure to `catch up' with friends -- pushing everyone except the top-earning agents towards higher consumption. As these pressures are only to `catch up', they only `flow down' from higher consumption agents to their lower consumption friends. This feature ensures a unique equilibrium and that the top-earning agents do not see their consumption choices changed by the presence of social comparisons.

\begin{rem}\label{rem:unique_eqm}
    There exists a unique Nash equilibrium, with $x_{HR}^* = x^a(y_{HR})$, and $x_{q}^* \geq x^a(y_{q})$ for all $q \neq HR$, strictly so when $s \in (0,1)$.
\end{rem}

Those who do face pressures to `catch up' consume more when they are more closely connected to their richer friends in the network, and when these friends earn more. This is because both increase the pressures to `catch up' -- which agents try to alleviate with higher consumption. The latter effect is due to convex social comparison costs -- agents face stronger pressure at the margin to catch up when the gap with their friends is larger. This is in addition to the standard effect that higher own earning potential increases consumption directly. 

Here, we focus on the impact of changing the level of economic mobility and the level of integration. These are the key features of interest in our setting. But note that changing each of these affects primitives -- either earning potential or network structure -- for multiple classes simultaneously.

\paragraph{Increasing Mobility.} Higher economic mobility increases earning potential for the $HP$ class, but lowers it for the $LR$ class. So it unambiguously increases consumption for the $HP$ class and reduces consumption for the $LR$ class. These are just direct effects of changes to own earning potential. Changing mobility has no impact on consumption of the $HR$ class -- their own earning potential does not change, and as they face no pressures to catch up, they are unaffected by changes in others' behavior. 

In contrast, the effect on the $LP$ class is ambiguous in general. While their own earning potential does not change, they are affected by the changing consumption of the $HP$ and $LR$ classes due to social comparisons -- and these change in different directions. Which of the two forces wins out depends on the strength of the social pressures the $LP$ class faces from each -- which depends on both how closely connected they are in the network and how far away they are in terms of consumption. 

\begin{rem}\label{rem:indiv_comp_stat_x_s}
An increase in economic mobility, $s$:
    (i) does not change $x_{HR}^*$, 
    (ii) strictly increases $x_{HP}^*$, 
    (iii) strictly decreases $x_{LR}^*$, and 
    (iv) increases $x_{LP}^*$ if and only if 
\begin{align}
        (1-\rho)\, F''(x_{HP}^* - x_{LP}^*) \cdot\frac{dx^*_{HP}}{ds} + \rho\, F''(x_{LR}^* - x_{LP}^*)\cdot\frac{dx^*_{LR}}{ds} > 0
\end{align}
\end{rem}

\paragraph{Increasing Integration.} Greater integration swings everyone's social comparisons towards those with the same ability (but different background) and away from those with the same background (but different ability). This has no effect on the $HR$ class -- as they face no pressures to catch up from any of their social comparisons. It increases consumption for the $HP$ class because it strengthens connections to the $HR$ class, intensifying the pressures to `catch up'. While it also weakens the connections to the $LP$ class, this provides no respite because in equilibrium the $HP$ class consume more than the $LP$ class, and so do not face any pressure to catch up to them. For the $LR$ class, higher integration has a mirror image impact, and so reduces their consumption.  

For the $LP$ class, the impact is ambiguous in general. Greater integration leads to stronger connections to the $LR$ class -- which increases the pressure to catch up -- while also weakening connections to the $HP$ class -- which reduces the pressure to catch up. There are also indirect effects coming from the fact that the consumption choices of the $LR$ and $HP$ groups are also changing as integration increases. 

\begin{rem}\label{rem:indiv_comp_stat_x_rho}
An increase in integration, $\rho$:
    (i) does not change $x_{HR}^*$, 
    (ii) strictly increases $x_{HP}^*$, 
    (iii) strictly decreases $x_{LR}^*$, and 
    (iv) increases $x_{LP}^*$ if and only if 
\begin{align}
        \underbrace{- F'(x^*_{HP} - x^*_{LP}) + F'(x^*_{LR} - x^*_{LP}) \phantom{\int}}_{\text{direct effect}} + \underbrace{(1-\rho)\, F''_{LP,HP}\cdot\frac{dx^*_{HP}}{d\rho} + \rho\, F''_{LP,LR}\cdot\frac{dx^*_{LR}}{d\rho}}_{\text{indirect effect}} > 0
\end{align}
\end{rem}

Having examined how each class's consumption in our model depends on economic mobility and integration, we now turn to our main focus: welfare.

\section{Welfare}\label{sec:welfare} Throughout this section we consider the class of continuously differentiable social welfare functions, $W$, that: \emph{(i)} have $\frac{\partial W}{\partial U_i} > 0$ for all $i$, and \emph{(ii)} exhibit (weak) inequality aversion insofar as they satisfy the Pigou-Dalton Principle (i.e. weakly increasing in a transfer of utility from an agent with higher utility to one with lower utility that does not reverse the ordering of utilities; see \citet[p.~67]{moulin2004fair}). For simplicity, we assume that the partial derivatives are bounded above and bounded away from zero. 

Note that this covers the utilitarian welfare function as a special case where the Pigou-Dalton Principle is only weakly satisfied (i.e. these transfers have no impact on the welfare function), and individual utilities are aggregated linearly. The utilitarian case is of special interest because it is implicitly the welfare function used in our data analysis.

Our main result shows that when economic mobility is initially low, greater mobility unambiguously increases welfare for this broad class of social welfare functions. And the reverse is true when economic mobility is initially high: greater economic mobility unambiguously decreases welfare. In an abuse of terminology, we call this an `inverse-U' relationship.

\begin{prop}\label{prop:convex_comps_result}
For any strictly increasing social welfare function, $W$ that (weakly) satisfies the Pigou-Dalton Principle, social welfare is:

\noindent (i)  strictly increasing in economic mobility if economic mobility is low (i.e. sufficiently close to $0$), and

\noindent (ii) strictly decreasing in economic mobility if economic mobility is high (i.e. sufficiently close to $1$).
\end{prop}

This result matches our headline stylized fact, shown in \Cref{fig:1}, that among US counties there is an `inverse-U' shaped relationship between economic mobility and aggregate well-being. Intuitively, this result arises because of the way higher social mobility affects the various social comparisons, ameliorating some while aggravating others. When mobility is initially low, the comparisons that higher mobility ameliorates are `large', and so there are large welfare gains from shrinking them. In contrast, the comparisons that greater mobility aggravates are `small' -- so there are only small welfare costs from worsening them. Everything works in reverse when mobility is initially high: increasing it ameliorates the `small' comparisons and aggravates the `large' ones -- which reduces welfare overall.

To see how this works, note that in equilibrium four comparisons create pressure to `catch up'. Letting $q \rightsquigarrow q'$ denote that class $q$ makes an unfavorable social comparison with class $q'$, we have: (1) $HP \rightsquigarrow HR$, (2) $LR \rightsquigarrow HR$, (3) $LP \rightsquigarrow HP$, and (4) $LP \rightsquigarrow LR$. Higher economic mobility moves the $HP$ group `up' while moving the $LR$ group `down'. So it ameliorates comparisons (1) and (4) -- by bringing the $HP$ group up closer to the $HR$ group, and the $LR$ group down closer to the $LP$ group. And by identical reasoning it aggravates comparisons (2) and (3).


When economic mobility is initially low, $x_{HP}^*$ is relatively low and $x_{LR}^*$ is relatively high (and regardless of the level of economic mobility, $x_{LP}^*$ is lowest and $x_{HR}^*$ is highest).
In turn, this means the comparisons (1) $HP \rightsquigarrow HR$ and (4) $LP \rightsquigarrow LR$ are `large', in the sense that the gap in consumption is large. So these comparisons are very costly, both in absolute terms and at the margin. In contrast, the comparisons (2) $LR \rightsquigarrow HR$ and (3) $LP \rightsquigarrow HP$ are `small'. So these comparisons are not very costly.
When starting from an initially low level, increasing economic mobility therefore ameliorates the very costly `large' comparisons and aggravates the not so costly `small' comparisons. This pushes overall welfare up.

Additionally, these changes unambiguously reduce the \emph{inequality} of utilities. When economic mobility is initially low, the groups who gain ($HP$ and $LP$) had low utility to begin with. And the group who loses ($LR$) had high utility to begin with. That overall utility and inequality of utility both change in the same direction combine to yield our general welfare result.

Of course, the logic of the previous two paragraphs works in reverse when economic mobility is initially high. The two comparisons that are ameliorated by raising economic mobility are `small' — so there are small welfare gains from this. In contrast, the two comparisons that are aggravated by raising economic mobility are `large' — leading to large welfare losses. Further, as economic mobility is already high to begin with, the winners from further increases to economic mobility are already better off than the losers. So inequality also rises.

The `U-shaped' relationship between inequality of utilities and economic mobility is also shaped by the `U-shaped' relationship between \emph{income inequality} (i.e. inequality of earning potential) and economic mobility. Greater economic mobility means earning potential depends more heavily on ability -- rewarding those with high ability. Past a certain point, those with high ability (but with poor parents) are already out-earning those with rich parents but low ability. So greater mobility then increases income inequality.

\begin{rem}\label{rem:income_inequality}
    Inequality of earning potential is first decreasing and then increasing in economic mobility.
\end{rem}

While this follows fairly directly from our modeling assumptions, we contend that this relationship is an inherent feature of social mobility, not specific to our modeling approach. We discuss this in \Cref{sec:discussion}.

\paragraph{The role of integration.} Increasing integration also affects the pattern of social comparisons. It strengthens the (1) $HP \rightsquigarrow HR$, and (4) $LP \rightsquigarrow LR$ comparisons. And it weakens the (2) $LR \rightsquigarrow HR$, (3) $LP \rightsquigarrow HP$ comparisons. As explained above, when mobility is initially low, the comparisons (1) $HP \rightsquigarrow HR$, and (4) $LP \rightsquigarrow LR$ are `large'. So strengthening them leads to large welfare losses. And at the same time, the comparisons (2) $LR \rightsquigarrow HR$, and (3) $LP \rightsquigarrow HP$ are `small'. So weakening them leads to small welfare gains. Overall, higher integration therefore \emph{reduces} welfare when economic mobility is initially low. By identical logic, the reverse is of course also true when economic mobility is initially high: welfare is increasing in integration.

\begin{prop}\label{prop:welfare_rho}
For any strictly increasing social welfare function, $W$, that (weakly) satisfies the Pigou-Dalton Principle, social welfare is:

\noindent (i) strictly decreasing in integration if economic mobility is low (i.e. sufficiently close to $0$), and

\noindent (ii) strictly increasing in integration if economic mobility is high (i.e. sufficiently close to $1$).
\end{prop}

The key insight here is that whether more integration is good or bad for welfare depends critically on the prevailing level of economic mobility. In highly immobile societies, where earning potential is largely determined by who your parents are, integrating societies to create more connections between people from different backgrounds harms welfare. This is because these new connections create more harmful social comparisons, and impose strong pressure on poorer groups to over-consume. A failure to address economic immobility could therefore lead to unintended harms of policies aimed at integrating society.

Similarly, the welfare effects of greater mobility can depend on the prevailing level of integration. This is because, when society is well integrated (high $\rho$), most of the social comparisons are between same-ability people. These are precisely the ones that are ameliorated by greater economic mobility. The same-background comparisons, which are aggravated by economic mobility, are not prominent in well-integrated societies.

\begin{prop}\label{prop:welfare_intermediate_mobility_general}
Let $s^y \in (0,1)$ be the unique level of economic mobility at which $y_{HP} = y_{LR}$. For any strictly increasing social welfare function, $W$, that (weakly) satisfies the Pigou-Dalton Principle, there exist $\underline{s}, \bar{s}$ with $\underline{s} < s^y < \bar{s}$ such that:
\begin{enumerate}
    \item[(i)] for all $s < \bar{s}$, social welfare is strictly increasing in economic mobility if integration $\rho$ is sufficiently close to 1, and
    \item[(ii)] for all $s > \underline{s}$, social welfare is strictly decreasing in economic mobility if integration $\rho$ is sufficiently close to 0.
\end{enumerate}
\end{prop}

This shows that when mobility is initially low or moderate, high integration is sufficient to ensure welfare gains from greater mobility. In contrast, integration is not sufficient for welfare gains when mobility is initially high -- greater mobility will increase inequality as the already strongly rewarded high-ability agents pull away from the low-ability agents. But integration is necessary. In a very poorly integrated society that starts with moderate or high mobility, welfare losses from further increases in mobility are guaranteed in our model.

Together, these results carry an important policy implication: to get the most from either increasing social mobility or integration, they need to go hand in hand.

\section{The Empirical Relationship}\label{sec:data} In this section we first show that the `inverse-U' relationship between well-being and economic mobility that we presented in the introduction is robust. We then show that the data is consistent with two additional predictions of our theory -- on the role integration plays in the relationship between economic mobility and well-being, and the relationship with income inequality.

\subsection{Data}
\paragraph{Economic Mobility.} We use the now-common measure of economic mobility from the \emph{Opportunity Atlas} \citep{chetty2018impacts}: the predicted income rank (percentile) in adulthood for children born in 1978--1983 to parents who were at the 25th percentile of the national income distribution. 

\paragraph{Well-being.} As a proxy for well-being, we use data from the \emph{Human Flourishing Geographical Index} \citep{iacus2025human}, which measures various aspects of well-being by analyzing the sentiments of over 2.6 billion geotagged tweets. We use the happiness metric as a proxy for utility in our model. It is the most commonly found sentiment (whether positive or negative) in the data.\footnote{In the Online Appendix, we also report results using measures of optimism and life satisfaction as alternative proxies. This does not affect the qualitative results.} 
We average across all years -- giving equal weight to each tweet, and using all available data.

\paragraph{Integration.} 
As a proxy for the level of integration, we use the ``childhood economic connectedness'' from \cite{chetty_nature1,chetty_nature2}. For children whose parents have low socio-economic status, this measures the share of their friends whose parents have high socio-economic status. Recall that within our model, more friendships between people with different parental backgrounds is equivalent to more friendships between people with similar ability.




\paragraph{Census Data.} As controls, we use data on median income ($'000$s dollars per year), population density ($'000$s per square kilometer), urban population share, migration (as a fraction of county population, both net and gross), and religiosity from US Census data. 
We also get income inequality data (Gini coefficient) from the same source

\paragraph{Social Capital.} We use measures of cohesiveness and civic engagement from \cite{chetty_nature1, chetty_nature2} as further controls. These are both forms of social capital \citep{coleman1988social,putnam1995tuning}.\footnote{\citet{chetty_nature1,chetty_nature2} also contend that `connectedness' is a form of social capital. We use one measure of connectedness as a proxy for integration, so it is not an appropriate control here. Regardless, our results are not affected by adding in their preferred measure of economic connectedness as an additional control.}

\subsection{The main result: an `inverse-U' relationship}
In the introduction, we presented the unconditional relationship between economic mobility and well-being, finding that the relationship is positive for the bottom third of US counties (by their economic mobility), but \emph{negative} for the top third. Here, we explore this relationship in more detail.

\Cref{tab:1} reports our main regression results. Column (I) provides an analogy to \Cref{fig:1} -- it includes only the mobility term and its square. Column (II) adds the integration control. Column (III) adds state fixed effects and a control for median income -- key potential confounders. In column (IV), we add urbanization, migration, religion, and social capital controls.\footnote{There are of course a variety of possible channels through which income, urbanization, migration, religion, and social capital could affect well-being. We are not directly interested in these mechanisms. Rather, we control for these factors to reduce the risk that the relationship between economic mobility and well-being is spurious. Also note that we do not control for income inequality here. It is potentially a bad control. Our theory predicts that economic mobility influences inequality, and hence changes in inequality may be part of the mechanism we study. Nevertheless, we add income inequality controls in \Cref{OA:empirical_relationship} and show doing so has no impact on the qualitative results.}
Finally, column (V) adds an interaction term between mobility and integration. This is informed by our theoretical model, which finds that the impact of mobility depends on the level of integration.

\begin{table}[ht]\centering
\def\sym#1{\ifmmode^{#1}\else\(^{#1}\)\fi}
\caption{WELL-BEING AND ECONOMIC MOBILITY -- MAIN REGRESSION RESULTS}\label{tab:1}
\begin{tabular}{l*{5}{c}}
\toprule
Dep. var.: well-being
          & (I) & (II) & (III) & (IV) & (V) \\
\hline 
Mobility  &   31.457\sym{***}&   22.802\sym{***}&   24.451\sym{***}&   19.563\sym{***}&   20.111\sym{***}\\
          &  (2.629)         &  (3.311)         &  (3.321)         &  (3.200)         &  (3.202)         \\
Mobility$^{2}$&  -31.693\sym{***}&  -21.799\sym{***}&  -25.840\sym{***}&  -21.432\sym{***}&  -25.675\sym{***}\\
          &  (2.894)         &  (3.707)         &  (3.696)         &  (3.585)         &  (3.907)         \\
Integration&                  &   -0.034         &   -0.522\sym{***}&   -0.683\sym{***}&   -2.299\sym{***}\\
          &                  &  (0.081)         &  (0.128)         &  (0.130)         &  (0.610)         \\
Mobility $\times$ Integration&                  &                  &                  &                  &    3.814\sym{***}\\
          &                  &                  &                  &                  &  (1.405)         \\
Median Income&                  &                  &    0.157\sym{***}&    0.140\sym{***}&    0.131\sym{***}\\
          &                  &                  &  (0.026)         &  (0.028)         &  (0.028)         \\
\hline
State FE & NO & NO & YES & YES & YES \\
Pop. Contr. & NO & NO & NO & YES & YES \\
Religion Contr. & NO & NO & NO & YES & YES \\
Social Capital Contr. & NO & NO & NO & YES & YES \\
R-squared & 0.068 & 0.059 & 0.417 & 0.492 & 0.493\\
Observations & 3136  & 2728 & 2728  & 2728   & 2728 \\
\bottomrule
\end{tabular}
\begin{center}{\parbox[b]{16cm}{\footnotesize \emph{Notes:} Standard errors in parentheses. \sym{*} \(p<0.10\), \sym{**} \(p<0.05\), \sym{***} \(p<0.01\). Pop. controls = population density, population density squared, fraction of county that is urban, net migration, gross migration (total inward plus outward migration). Net and gross migration are calculated as a fraction of the county population. Religion controls = fraction of people reporting that they belong to a religious group (adherence), and a Herfindahl-Hirschman Index of adherence using four groups (``Catholic'', ``Evangelical'', ``Other Christian'', and ``Non-Christian Religious''). Social capital controls = clustering (fraction of a person's friend pairs who are also friends), `support ratio' (proportion of within-ZIP code friendships that have a mutual friend in the same ZIP code), proxy for volunteering rate, and proxy for density of civic organizations. }} 
\end{center}
\end{table}

The key insight is that the `inverse-U' relationship between economic mobility and well-being is robust and highly economically significant. In \Cref{OA:empirical_relationship}, we consider a range of alternative specifications: (a) running separate regressions on segments of the data, (b) trimming the tails of the data, and (c) using alternative measures of well-being. In all cases, the key result -- that well-being is increasing in mobility for low levels of mobility and decreasing in mobility for high levels of mobility -- remains highly significant.

It is important to note that while our theory is for a general welfare function, our empirical analysis is implicitly using a \emph{utilitarian} aggregation. This is because our well-being data captures average well-being, with no consideration given to dispersion. In this case our theory model is silent as to the relationship between well-being and inequality after controlling for the level of economic mobility.

\subsection{Other Predictions} 
\paragraph{Well-being and Integration.} In \Cref{prop:welfare_rho} we show that 
more integration reduces welfare when economic mobility is low, and increases welfare when economic mobility is high. 
This motivates the addition of an interaction term between economic mobility and our proxy for integration in our regression. We include this in our main specifications in column (V) of \Cref{tab:1}.

The positive and statistically significant coefficient on the interaction between mobility and integration is consistent with the prediction: an increase in economic mobility results in a larger increase in utilitarian welfare when integration is higher. In addition, \Cref{prop:welfare_rho} says that when mobility is low, there is a negative relationship between well-being and integration, while when mobility is high, there is a positive relationship between well-being and integration. The empirical findings are also in line with this prediction: there is a negative coefficient on the stand-alone integration term, and a positive and statistically significant coefficient on the interaction between economic mobility and integration. More broadly, this finding is suggestive that the network structure is playing a role in mediating the relationship between economic mobility and well-being -- also congruent with our theory.

\paragraph{Mobility and Inequality.} Recall that one part of the mechanism in our model is that increasing economic mobility has a `U-shaped' relationship with income inequality (\Cref{rem:income_inequality}). This reinforces the non-monotonic effect that economic mobility has on utilitarian welfare. 
We find correlational evidence supporting this in the data. Rising economic mobility is associated with first decreasing and then increasing income inequality. This is shown graphically in \Cref{fig:2}, and \Cref{tab:2} shows the relationship is robust to including a range of controls (the same set that we used in \Cref{tab:1}). 
 
\begin{table}[h!]\centering
\def\sym#1{\ifmmode^{#1}\else\(^{#1}\)\fi}
\caption{INCOME INEQUALITY AND ECONOMIC MOBILITY -- REGRESSION RESULTS}\label{tab:2}
\begin{tabular}{l*{5}{c}}
\toprule
Dep. var.: Gini & (I) & (II) & (III) & (IV) & (V) \\
\hline
Mobility  &   -0.161\sym{***}&   -1.039\sym{***}&   -1.043\sym{***}&   -1.192\sym{***}&   -0.891\sym{***}\\
          &  (0.010)         &  (0.091)         &  (0.104)         &  (0.133)         &  (0.125)         \\
Mobility$^{2}$&                  &    0.972\sym{***}&    0.965\sym{***}&    1.185\sym{***}&    0.895\sym{***}\\
          &                  &  (0.100)         &  (0.114)         &  (0.148)         &  (0.140)         \\
Median Income&                  &                  &                  &   -0.015\sym{***}&   -0.016\sym{***}\\
          &                  &                  &                  &  (0.001)         &  (0.001)         \\
Integration&                  &                  &                  &    0.036\sym{***}&    0.035\sym{***}\\
          &                  &                  &                  &  (0.005)         &  (0.005)         \\
\hline
State FE & NO & NO & YES & YES & YES \\
Pop. Contr. & NO & NO & NO & NO & YES \\
Religion Contr. & NO & NO & NO & NO & YES \\
Social Capital Contr. & NO & NO & NO & NO & YES \\
\hline
R-squared & 0.082 & 0.109 & 0.280 & 0.360 & 0.474   \\
Observations& 3136 & 3136 & 3134  & 2728 & 2728 \\
\bottomrule
\end{tabular}
\begin{center}{\parbox[b]{12cm}{\footnotesize \emph{Notes:} Standard errors in parentheses. \sym{*} \(p<0.10\), \sym{**} \(p<0.05\), \sym{***} \(p<0.01\). Population, Religion, and Social capital controls as in \Cref{tab:1}.}} 
\end{center}
\end{table} 

\begin{figure}[h!]
    \centering
    \includegraphics[width=0.99\linewidth]{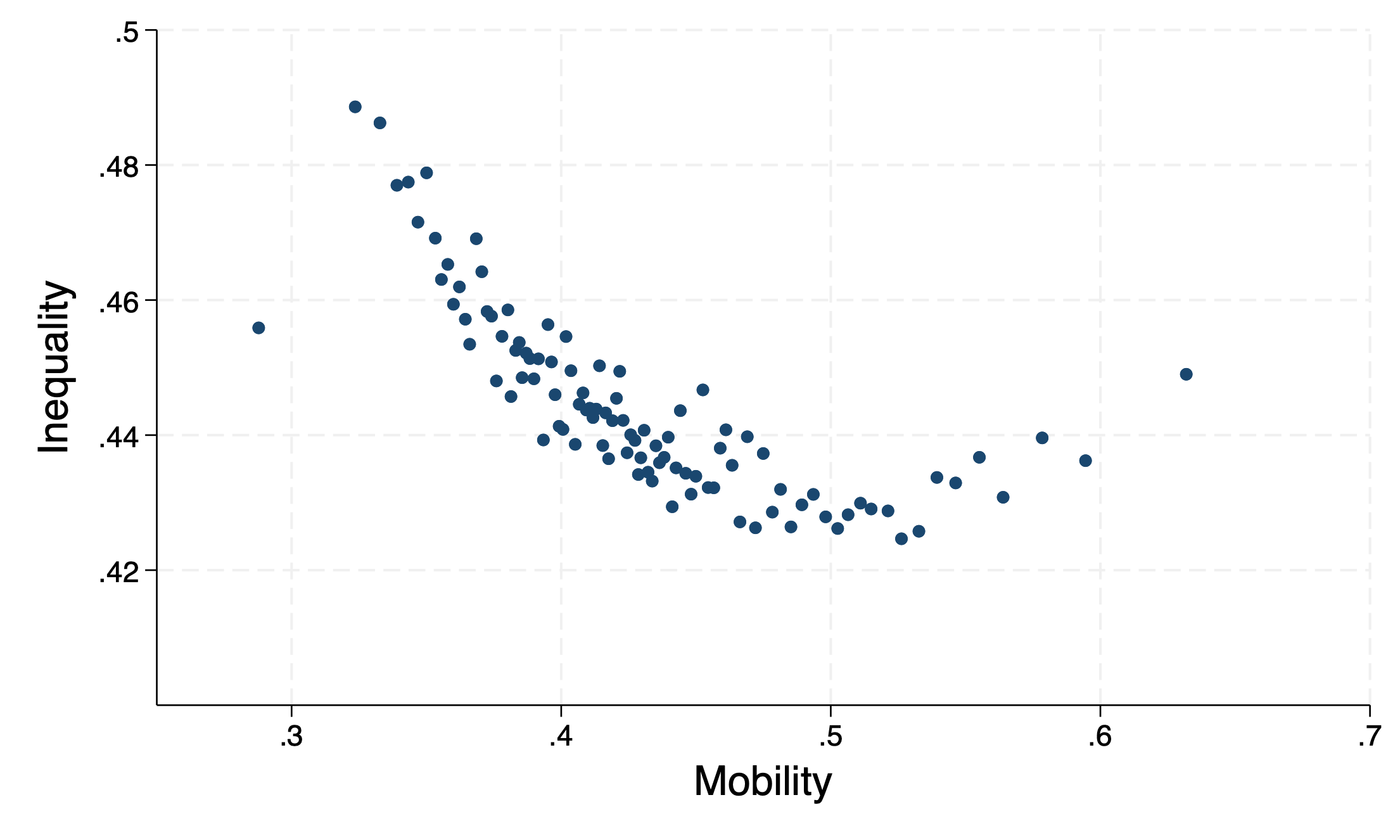}
    \caption{Inequality and mobility for US counties, binned scatter plot. }
    \label{fig:2}
\end{figure}

\section{Understanding High-mobility Counties}\label{sec:abs_mobility}
Following \citet{chetty2014land}, we have used `absolute upward mobility' as our preferred measure of social mobility. This is the predicted rank in the national income distribution of children whose parents were at the 25th percentile of the national income distribution. Notice that in some counties this measure exceeds $0.5$ -- meaning that children from low-income families in the county, on average, end up above the national median in the national income distribution as adults. On the face of it, this may seem puzzling. At a national level, if parental income had no predictive power at all, children of parents at the 25th percentile would be expected to reach the 50th percentile on average.

However, this can happen at the county level due to local growth shocks that move a county relative to the rest of the nation. Like \citet{chetty2014land}, we view these shocks as real sources of upward mobility, and contend that they would have welfare effects operating through our mechanism. As these shocks will have benefited some more than others, they may create even more painful social comparisons for those left behind within the county. 

In \Cref{subsec:bakken_case} we discuss the case of counties near the Bakken formation in North Dakota. Many of these counties have high social mobility (above $0.5$). There, an oil boom in the late 2000s created significant upward mobility, but left behind a proportion of the population. Then for completeness, in \Cref{subsec:rel_mobility} we consider a measure of \emph{relative} mobility (again, from \citet{chetty2014land}) that looks only at how predictive parental income rank is for child income rank within a county. Formally, this is $1$ minus the coefficient from a regression of a child's rank in the income distribution on her parents' rank in the income distribution.\footnote{Of course the relationship between child rank and parental rank need not be linear. But \citet[cf. footnote 10]{chetty2014land} show that empirically this rank-rank relationship is approximately linear, and so that the regression coefficient is a good measure of this rank-rank relationship.} 
We find that using this alternative measure of mobility does not change the main qualitative result -- that high mobility has downsides for well-being.

\subsection{A case study: the Bakken oil boom}\label{subsec:bakken_case}
The Bakken formation straddles western North Dakota and eastern Montana. In the late 2000s, higher oil prices and advances in horizontal drilling and hydraulic fracturing made large-scale extraction economically viable. Oil production in the region then expanded rapidly, transforming the local economy.
The boom created high-wage opportunities in a range of occupations that did not require the educational credentials normally associated with entry into the upper part of the national income distribution.\footnote{Oil production increased from 150,000 barrels per day in 2007 to more than one million in 2014 \citep{kulbeth_coleman_2014_bakken}. A majority of all jobs added between 2007 and 2011 were in resource extraction and related sectors \citep{ferree_smith_2013_bakken}. Average economic mobility is also significantly higher in the Bakken region ($\approx 0.59$) than in the rest of the US ($\approx 0.43$). }

But the rising tide did not lift all boats. The largest gains accrued to workers in oil and oil-adjacent industries, as well as land and property owners. In contrast, renters, those on fixed incomes, public-sector workers, and employees in non-oil sectors often faced substantial costs \citep{fernando2016oil}. Housing shortages and rent increases were central \citep{wascalus_2012_bakken_seniors,kim2020shale, gershenson2024fracking}. So while it pushed absolute mobility up overall, the boom created both winners and losers.\footnote{\cite{fernando2016attitudes} emphasize the importance of addressing the concerns of residents who perceived themselves as adversely affected by shale oil development.} 
Qualitative work on western North Dakota stresses precisely this heterogeneity: different groups within the same community perceived the boom's effects on quality of life very differently \citep{fernando2016oil}. These findings are consistent with the idea that the same process that raised economic opportunity could also reduce average well-being through its negative impacts on those left behind.

A notable feature here is that the division between winners and losers did not map neatly onto pre-existing class lines. Some winners were blue-collar workers in oilfield services or trucking.\footnote{These could be pre-existing residents, or new arrivals. More generally, migration could affect overall well-being in a county, through a number of possible mechanisms. Our regressions therefore control for both net migration and gross migration (sum of arriving and leaving).} 
Others were landowners or mineral-rights owners whose assets became unexpectedly valuable. Some losers were long-term residents, public employees or elderly renters -- whose incomes did not keep pace with boom-driven costs \citep{fernando2016oil, fernando2016socioeconomic}. This matters for the social-comparison mechanism. 
It means the sudden resource boom could split people who remained tightly connected in social networks: neighbors, schoolmates, extended families, churches, and long-standing small-town communities. The same shock that produced high absolute upward mobility could therefore also intensify comparison pressures within communities.\footnote{As in the Bakken region, high earners in the broader United States have also seen their incomes pull away from the rest \citep{reeves2015dangerous}. Unlike the Bakken region, however, there is evidence that the highly educated and affluent have increasingly separated from the rest of society through residential sorting, schooling, family structure, and associational life \citep{bischoff2014residential,owens2016income,lundberg2016family,snellman2015engagement,putnam2016our}. So we would expect the winners and relative losers from this increasing inequality in the broader United States to make weaker social comparisons with each other compared to in the Bakken region.}

Viewed through the lens of our model, this is exactly the configuration in which high absolute mobility would fail to translate into high well-being. The boom raised the economic outcomes of some people from low-income backgrounds, thereby increasing measured absolute upward mobility. But it also widened locally salient gaps, increased inequality, and put pressure on housing. The Bakken case therefore illustrates how counties with absolute upward mobility above $0.5$ can be places where genuine economic opportunity arrived through a disruptive local shock whose overall welfare effects were ambiguous or even negative.\footnote{\Cref{OA:resources} provides details on areas that experienced such economic shocks.}


\subsection{Controlling for relative mobility}\label{subsec:rel_mobility}
We now show that our headline result is robust to using an alternative measure of economic mobility that measures only how people move relative to others within the same county. Specifically, we use $1$ minus the `rank-rank slope' from \citet{chetty2014land} (then standardized to put it on the unit interval). The `rank-rank slope' is the coefficient from a regression of children's income percentile rank in adulthood on their parents' income percentile rank.

Economically, it captures the degree of intergenerational persistence in income ranks. And it allows for overall shifts in the level of the income rank within a county -- this is absorbed by the regression intercept. This strips out county-level shocks -- such as the oil boom in the Bakken region that we discussed above -- that push a county up or down relative to the rest of the nation. A larger regression coefficient indicates greater intergenerational persistence, and hence lower relative mobility.\footnote{\citet{chetty2014land} show that the relationship between parent rank and child rank within a county is approximately linear (although with some non-linearities at the tails), and so argue that the coefficient on the regression -- the `rank-rank slope' -- is a good measure of relative mobility overall.}

\begin{table}[htbp]\centering
\def\sym#1{\ifmmode^{#1}\else\(^{#1}\)\fi}
\caption{WELL-BEING AND ECONOMIC MOBILITY -- AN ALTERNATIVE MEASURE OF MOBILITY}\label{tab:3}
\begin{tabular}{l*{6}{c}}
\toprule
Dep. var.: well-being
          & (I) & (II) & (III) & (IV) & (V) & (VI) \\
\hline 
Relative Mobility &    4.545\sym{***}&    4.751\sym{***}&    2.728\sym{***}&    2.034\sym{**} &    1.859\sym{**} &    1.249         \\
          &  (0.807)         &  (0.920)         &  (0.814)         &  (0.791)         &  (0.803)         &  (0.851)         \\
Relative Mobility$^{2}$&   -1.730\sym{**} &   -1.874\sym{**} &   -1.506\sym{**} &   -1.298\sym{*}  &   -1.648\sym{**} &   -0.535         \\
          &  (0.711)         &  (0.797)         &  (0.703)         &  (0.686)         &  (0.739)         &  (0.759)         \\
Integration&                  &   -0.147\sym{*}  &   -0.436\sym{***}&   -0.639\sym{***}&   -1.049\sym{***}&   -2.452\sym{***}\\
          &                  &  (0.075)         &  (0.127)         &  (0.129)         &  (0.349)         &  (0.616)         \\
Relative Mobility $\times$ Integration&                  &                  &                  &                  &    0.724         &   -0.118         \\
          &                  &                  &                  &                  &  (0.572)         &  (0.719)         \\
Median Income&                  &                  &    0.141\sym{***}&    0.137\sym{***}&    0.133\sym{***}&    0.127\sym{***}\\
          &                  &                  &  (0.026)         &  (0.027)         &  (0.027)         &  (0.028)         \\
(Absolute) Mobility &                  &                  &                  &                  &                  &   18.346\sym{***}\\
          &                  &                  &                  &                  &                  &  (3.448)         \\
(Absolute) Mobility$^{2}$&                  &                  &                  &                  &                  &  -25.026\sym{***}\\
          &                  &                  &                  &                  &                  &  (4.202)         \\
(Absolute) Mobility $\times$ Integration&                  &                  &                  &                  &                  &    4.433\sym{**} \\
          &                  &                  &                  &                  &                  &  (1.784)         \\
\hline
State FE & NO & NO & YES & YES & YES & YES \\
Pop. Contr. & NO & NO & NO & YES & YES & YES  \\
Religion Contr. & NO & NO & NO & YES & YES & YES \\
Social Capital Contr. & NO & NO & NO & YES & YES & YES \\
\hline
R-squared & 0.133 & 0.124 & 0.427 & 0.501 & 0.502 & 0.508  \\
Observations& 2764  & 2649 & 2649 & 2649  & 2649  & 2649  \\
\bottomrule
\end{tabular}
\begin{center}{\parbox[b]{16cm}{\footnotesize \emph{Notes:} Standard errors in parentheses. \sym{*} \(p<0.10\), \sym{**} \(p<0.05\), \sym{***} \(p<0.01\). Population, Religion, and Social capital controls as in \Cref{tab:1}.}} 
\end{center}
\end{table}

Columns (I)--(V) of \Cref{tab:3} report the same regressions as the respective columns in \Cref{tab:1}, but using `relative mobility' instead of our preferred measure. Column (VI) then includes both our preferred measure and the alternative measure (`relative mobility') at the same time. The coefficient on relative mobility (and on its squared term) is no longer statistically significant in column (VI), but keeps the same sign. Additionally, the interaction term between relative mobility and integration flips sign between (V)and (VI), while the absolute mobility interaction in (VI) is significant. Other results are unchanged. So when including both, our preferred measure appears to have a more robust relationship with well-being.

The key point is that our headline finding -- that beyond a certain point there can be welfare costs to more economic mobility -- persists if we use this alternative measure of mobility. This is despite the fact that it strips out the kinds of local shocks that affected the Bakken region and which we contend have real welfare implications that should be accounted for.

\section{Concluding Discussion}\label{sec:discussion}   We finish by briefly discussing an alternative interpretation of our theory, and then make a couple of broader comments on the contribution of our paper to public debate.

\subsection{Mobility \emph{as} Meritocracy: an alternative interpretation}\label{sec:meritocracy}
Our key parameter of interest -- economic mobility -- measures the extent to which a person's economic outcomes are determined by their own ability relative to their parental background. Economic mobility therefore captures the extent to which a society rewards talent. So our notion of economic mobility can be interpreted as a view of \emph{meritocracy} that aligns ``merit'' with talent \citep{loury1981intergenerational}. This is the view of meritocracy set out in Michael Sandel's influential \citeyear{sandel2020tyranny} book \emph{The Tyranny of Merit}, as well as \cite{markovits2019meritocracy} and \cite{young1958rise} (who is credited with coining the term). Adopting this view of meritocracy, we could interpret changes in our key parameter, $s$, as changes in the extent to which a society is meritocratic. All of our results would clearly remain exactly the same under this interpretation.

Both \cite{sandel2020tyranny} and \cite{markovits2019meritocracy} argue that meritocracy can be corrosive to society, in part by creating (and justifying) large inequalities between groups. Additionally, in this view, meritocracy can provoke humiliation and resentment in those who end up at the bottom of society by pressing a moral argument that those people \emph{deserve} to be there. In doing so it compounds their lack of economic success by reducing their social prestige or esteem. Further, it can provoke hubris in those who end up at the top of society -- because it claims they too \emph{deserve} their position. 

We contend that this idea of provoking humiliation and hubris maps closely onto increases in the overall strength of social comparison, $d$, in our model. Intuitively, more humiliation from being at the bottom of the economic pile is one way of making unfavorable comparisons (with those who are richer/more successful) more harmful. Greater hubris from being at the top of the economic pile also seems likely to compound this.  
And note that increasing $d$ pushes down welfare.\footnote{It is clear that in our model increasing the strength of social comparison, $d$, increases utility losses from unfavorable comparisons, but does not provide any offsetting benefits. So increasing $d$ will, all else equal, reduce utilitarian welfare.}

So suppose that an increase in the extent of meritocracy were to increase both economic mobility, $s$, and the strength of social comparisons, $d$, simultaneously. Then welfare will be decreasing in the extent of meritocracy whenever mobility is initially high -- exactly as in \Cref{prop:convex_comps_result}. However, we cannot guarantee that welfare is increasing in the extent of meritocracy in this setup. When mobility is initially low, the negative welfare effects of increasing $d$ could outweigh the gains from raising $s$. 

\paragraph{Ability vs. Background.} So far, we have worked on the basis that more weight on ability corresponds to greater social mobility/meritocracy. In theory, this is straightforward. But it abstracts from the thorny question of how to separate ability from parental background. In reality, many aspects of what are commonly considered ability -- including intelligence, work ethic, and some aspects of health -- are at least partly determined by one's parents \citep{benabou2000merit}. Whether these are inherited (nature) or due to the use of parental resources in childhood (nurture) is irrelevant for our purposes. So it is best to view the ability parameter in our model as the part of an agent's ability that is not attributable to their parents. 

\subsection{Policy Implications}
Whether interpreted as mobility or as meritocracy, our model shows how putting more weight on people's ability can generate negative welfare effects. Those with low ability lose out -- and are pushed into ever more costly social comparisons. When economic mobility is high enough, these losses outweigh the gains to the winners. But our results also show a critical interaction with the level of integration. In well-integrated societies, these costs to the losers are mitigated, so there are greater benefits to high mobility.

Policymakers who want society to benefit from greater mobility should also integrate society, allowing people's relationships to transcend the circumstances of their birth and to become organized according to ability rather than parental background. For example, organizing education by a merit-based grammar school system, rather than a private school system, could help a society that is already fairly mobile and wants to ensure that it continues to experience gains from high mobility. This is because it aims to sort children more along the lines of ability than along the lines of parental background, helping them to form friendships in this way.

\subsection{Mobility and Inequality}\label{sec:mobility_inequality}
The other mechanism that drives our general welfare results is the `U-shaped' relationship between economic mobility and earning potential (income) inequality (see \Cref{rem:income_inequality}). This relationship does not depend on social comparisons in our model. It simply recognizes that past a certain point, those who gain from greater economic mobility (high-ability agents with poor parents) are already out-earning those who lose out from greater mobility (low-ability agents with rich parents). So past this point, greater mobility exacerbates the differences in earning potential between groups -- increasing inequality. 
This creates an additional drag on the benefits of increasing mobility when welfare is inequality averse. 

This is in many ways obvious in theory, but is largely ignored in public debate. Instead, the focus is largely on the downwards-sloping part of the relationship: that higher inequality goes together with lower economic mobility. This is often called the ``Great Gatsby Curve'' \citep{corak2013income}. One possible reason for the focus on this part of the relationship is that it is currently the more pervasive one empirically. 

While \Cref{fig:2} documents a `U-shaped' relationship between income inequality and social mobility in the United States, the turning point is around the 85th percentile of the data. So the majority of US counties sit in the downwards-sloping region. For them, more mobility is associated with less inequality. And ignoring non-linearities, the overall relationship is a downwards-sloping one -- see column (I) in \Cref{tab:2}. 
So our finding complements existing work on the Great Gatsby Curve. We find a non-linearity in the relationship, and that for high levels of economic mobility, the relationship goes into reverse -- more mobility is associated with \emph{more} inequality. 

Moreover, whenever social welfare is strictly inequality-averse, this `U-shaped' relationship between social mobility and inequality will push further towards an `inverse-U' shaped relationship between social mobility and social welfare. But note that this feature does \emph{not} feed into our empirical finding. We, and others, implicitly use a utilitarian aggregation in our empirical estimation, and so do not permit an inequality-averse social planner. Including inequality-aversion at the aggregate level would only further reinforce the `inverse-U' shaped relationship we observe, and strengthen both of our results and the main economic messages from the paper.

\subsection{Mobility and Efficiency}\label{sec:mobility_efficiency}
One mechanism we have abstracted away from -- and a common argument in favor of putting more weight on people's ability -- is that higher mobility can improve overall efficiency.
Allocating jobs to people based on ability leads to better matching between people and jobs. This can have spillover benefits for everyone in society, not just the high ability people who `move up' as a result of higher mobility. 

Efficiency gains would therefore push towards higher welfare, no matter the starting level of economic mobility. If the efficiency gains are weak enough, they would not alter our welfare results in Propositions \ref{prop:convex_comps_result} and \ref{prop:welfare_rho}. Of course, sufficiently strong efficiency gains would then leave an everywhere increasing relationship between economic mobility and welfare. As we cannot shut down an efficiency gains channel in our data analysis, whether the `inverse-U' shaped relationship predicted by our theory is then preserved in practice becomes an empirical question. The evidence we presented suggests that it \emph{is} preserved, and this in part motivates our abstraction from efficiency benefits in the theory.

\subsection{An Optimal Response by Local Government}
In our analysis, we noted that integration and economic mobility are complements in terms of their welfare effects. The marginal gain from increasing one increases in the prevailing level of the other. Given this, if local governments can impact the level of integration, we might expect them to do more to improve integration if they are endowed with a more economically mobile community. Similarly, if they can impact economic mobility, we might expect them to do more to improve mobility if they are endowed with a more integrated community.

Regardless of the direction of causality, such an endogenous response by local governments could create a positive relationship between economic mobility and integration. It turns out that we observe exactly this relationship in our data. More economically mobile counties are, on average, more integrated (corr$ = 0.41 , p<0.001$). \Cref{fig:3} plots the overall relationship.

\begin{figure}
    \centering
    \includegraphics[width=0.8\textwidth]{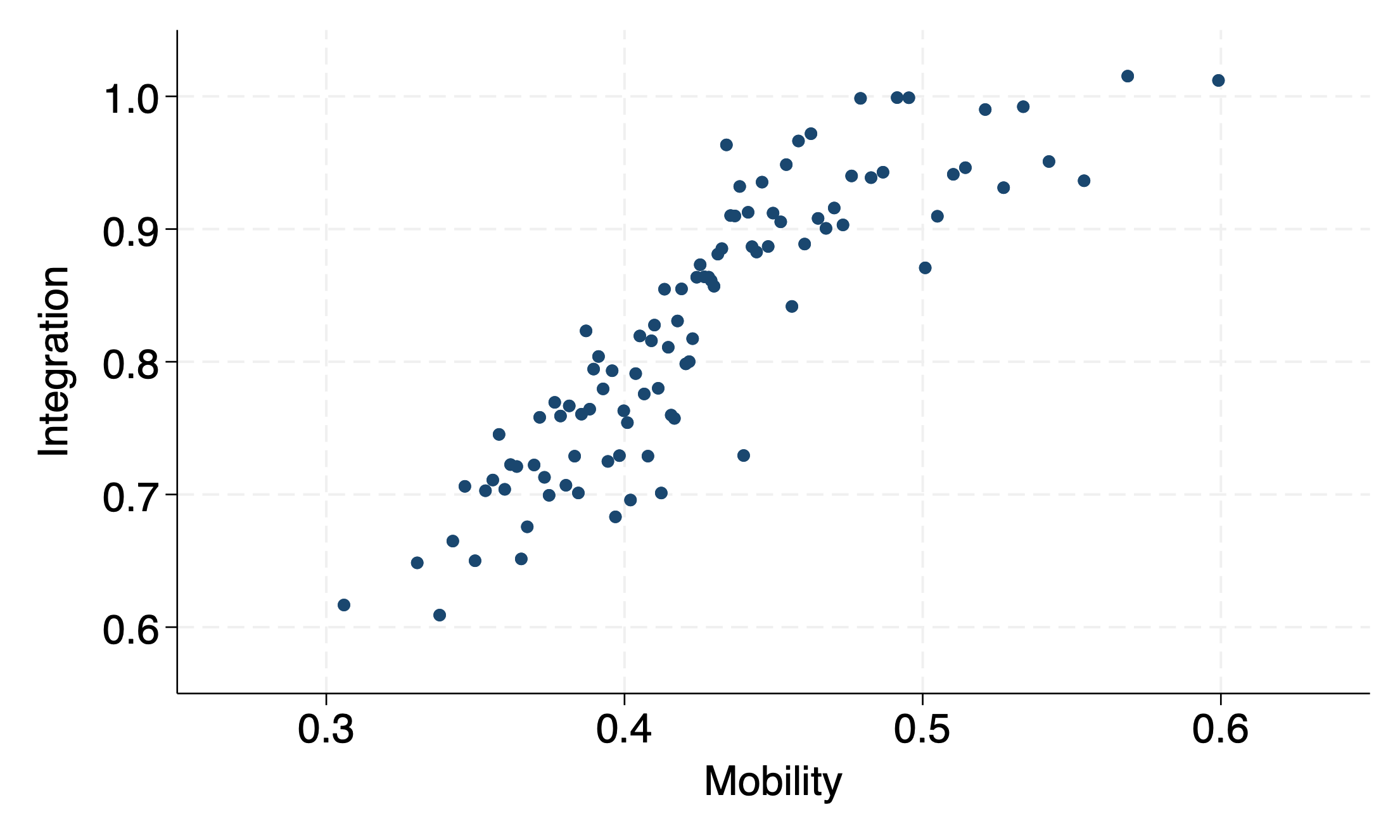}
    \caption{Integration, measured by ``childhood economic connectedness'' \citep{chetty_nature1,chetty_nature2}, and social mobility for US counties, binned scatter plot.}
    \label{fig:3}
\end{figure}

\subsection{The Range of Social Mobility} 
Our theory finds that for `low' levels of economic mobility, more mobility is good for welfare, and for `high' levels of economic mobility, more mobility is bad for welfare. But our theory is silent as to what will count as `low' and `high' in practice. So it is entirely possible that the `high' level of economic mobility falls outside of the empirically observed range. In that case, there could be a monotone relationship between well-being and mobility in the data. While this would not be inconsistent with our theory, it would leave it of more limited empirical relevance.

But this is not the case. The threshold \emph{does} lie somewhere within the range of social mobility levels that are observed in the United States. And so we are able to identify the `inverse-U' shaped relationship. Again, while stressing the back-of-the-envelope nature of our empirical exercise, it is instructive to note that the turning point is around the 80th -- 90th percentile of economic mobility in our data. So this negative relationship may be more of a fringe concern for policymakers today -- as it only appears to be `biting' for a minority of counties. But if society succeeds in improving economic mobility over time, then the dangers that we highlight in this paper -- and the network-related policy interventions that we show can mitigate them -- will also become more important over time.

\singlespacing
\bibliographystyle{abbrvnat}
\addcontentsline{toc}{section}{References}
\bibliography{bib}
\cleardoublepage
\onehalfspacing
\appendix
\numberwithin{equation}{section}
\numberwithin{thm}{section}
\numberwithin{prop}{section}
\numberwithin{defn}{section}
\numberwithin{rem}{section}
\numberwithin{cnj}{section}
\numberwithin{lem}{section}
\numberwithin{cor}{section}
\newpage
\section{Proofs}\label{sec:proofs} 
\paragraph{Preliminaries.} For convenience we let $\theta := d \, (1-\omega)$ and work with this throughout. 
It is also convenient to let
    \begin{align}\label{eq:marginal_payoff}
        \Psi_i(x_i , \mathbf{x}_{-i}) := u_x(x_i, y_i) + \sum_{j \neq i} G_{ij} \ F'\big( \max\{ x_{j} - x_i , \, 0 \} \big),
    \end{align}
which is agent $i$'s marginal payoff when she consumes $x_i$ and the other agents consume $\mathbf{x}_{-i}$.\footnote{Because $F(0) = F'(0) = 0$ by assumption, every term in \Cref{eq:preferences} is continuously differentiable in $x_i$.} Three properties follow immediately:
\begin{itemize}
    \item[(P1)] \emph{Best responses are unique and strictly positive.} This is because $u(x,y)$ has a unique interior maximum (by assumption), $u_x(\cdot, y)$ is decreasing and $F'(\cdot)$ is increasing with $F'(0) = 0$.
    
    \item[(P2)] \emph{Social comparisons (weakly) increase consumption.} This follows from the fact that $F'(\max\{x_j - x_i,0\})\geq0$ for all $x_i$ is decreasing in $x_i$ (strictly so when $x_j > x_i$).   
    \item[(P3)] \emph{Higher earning potential raises both the ideal level and the marginal benefit of consumption.} If $y_i > y_j$ then $x^a(y_i) > x^a(y_j)$ and $u_x(x,y_i) > u_x(x,y_j)$ for all $x>0$. Both follow from \Cref{ass:u_fn}(iii). 
\end{itemize}

\noindent It is also helpful to prove a series of lemmas that will be used in the proofs of the main results.



\begin{lem}\label{lem:A0_uniqueness}
    The game has a unique Nash equilibrium, and it is class-symmetric (i.e. $x_i^* = x_j^* = x_q^*$ for all $i,j \in q$).
\end{lem}
\begin{proof}
    \textbf{Step 1: uniqueness.} Suppose $\mathbf{x}$ and $\mathbf{z}$ are both Nash equilibria, with $\mathbf{x} \neq \mathbf{z}$. By (P1) we must have $\Psi_i(\mathbf{x}) = \Psi_i(\mathbf{z}) = 0$ for all $i$. Then let
    $\delta \ := \ \max_{i \in N} \ (x_i - z_i) \ > \ 0$,
    and let $i$ be an agent attaining this maximum. For every other agent $j$ we have $x_j - z_j \leq \delta = x_i - z_i$, and therefore
    \begin{align*}
        x_j - x_i \ \leq \ z_j - z_i \qquad \text{and hence} \qquad \max\{ x_j - x_i , \, 0 \} \ \leq \ \max\{ z_j - z_i , \, 0 \} .
    \end{align*}
    As $F'(\cdot)$ is increasing and $G_{ij}\geq 0$, every comparison term in $\Psi_i$ is weakly smaller at $\mathbf{x}$ than at $\mathbf{z}$. And $x_i > z_i$, so $u_x(x_i , y_i) < u_x(z_i , y_i)$ by \Cref{ass:u_fn}(i). Adding the two yields $\Psi_i(\mathbf{x}) < \Psi_i(\mathbf{z})$. Contradiction.


    \textbf{Step 2: existence.} Let $S := [0 , x^a(y_{HR})]$, and consider profiles in $S^n$. First, best responses map $S^n$ into $S$: they are strictly positive by (P1), and at $x_i = x^a(y_{HR})$ no other agent consumes strictly more, so every comparison term vanishes and $\Psi_i = u_x\big( x^a(y_{HR}) , y_i \big) \leq u_x\big( x^a(y_{HR}) , y_{HR} \big) = 0$ by (P3). So $x^{BR}_i \leq x^a(y_{HR})$. Second, $i$'s best response is non-decreasing in $x_j$ for all $j\neq i$ (as $F'\big( \max\{x_j - x_i , 0\} \big)$ is increasing in $x_j$). 

    Now restrict attention to class-symmetric profiles. By the paragraph above, the best response of every member of a class is the same at such a profile, so the best-response map takes class-symmetric profiles to class-symmetric profiles. Identifying class-symmetric profiles in $S^n$ with vectors in $S^4$, which is a complete lattice, the best-response map is therefore a monotone self-map of $S^4$. By Tarski's fixed point theorem it has a fixed point. 

    \textbf{Putting it together.} Step 2 shows there is a class-symmetric equilibrium and Step 1 says there is at most one equilibrium. So the equilibrium is unique, and it is class-symmetric.
\end{proof}

In light of \Cref{lem:A0_uniqueness}, we write $x^*_q$ for the common equilibrium consumption of class $q$, and work with the class-level marginal payoffs $\Psi_q$ from here on.

\begin{lem}\label{lem:A1_ordering}
    For all $s \in (0,1)$, we have $x_{HR}^* > x_{HP}^* , x_{LR}^* > x_{LP}^*$.
\end{lem}
\begin{proof}
    \textbf{Claim 1: $HR$ class have the highest consumption.} Suppose not. Then some other class $\hat q \neq HR$ has the highest consumption. So every comparison term in $\Psi_{\hat q}$ vanishes at $x^*_{\hat q}$. So $\Psi_{\hat q}(x^*_{\hat q}) = u_x(x^*_{\hat q} , y_{\hat q}) = 0$, which implies that $x^*_{\hat q} = x^a(y_{\hat q})$. But $y_{\hat q} < y_{HR}$, so applying (P3) and then (P2),
    $x^*_{\hat q} = x^a(y_{\hat q}) < x^a(y_{HR}) \leq x^*_{HR}$, which contradicts $x^*_{\hat{q}} \geq x^*_{HR}$.

    \textbf{Claim 2: $LP$ class have lowest consumption.} Given Claim 1, it is enough to show that $x^*_{LP} < x^*_{HP}$ and $x^*_{LP} < x^*_{LR}$. Suppose not. Then exactly one of the three cases below holds, and we derive a contradiction in each. 
    \textbf{Case (i).} Suppose $x^*_{LP} \geq x^*_{HP},x^*_{LR}$. Then every comparison term in $\Psi_{LP}$ vanishes at $x^*_{LP}$ and hence $x^*_{LP} = x^a(y_{LP})$.\footnote{Recall that the $LP$ class do not make comparisons with the $HR$ class.} But $y_{LP} < y_{HP}$, so (P3) and then (P2) give $x^*_{LP} = x^a(y_{LP}) < x^a(y_{HP}) \leq x^*_{HP}$. A contradiction.

    \textbf{Case (ii).} Suppose $x^*_{LR} > x^*_{LP} \geq x^*_{HP}$. Here we evaluate the marginal payoff of the $LP$ class \emph{at} the consumption level of the $HP$ class. At $x = x^*_{HP}$, the $LP$ class only makes an unfavorable comparison with the $LR$ class. 
    So:
    \begin{align}
        \Psi_{LP}(x^*_{HP}) = u_x(x^*_{HP} , y_{LP}) + \rho \theta F'(x^*_{LR} - x^*_{HP}). \label{eq:E.2}
    \end{align}
    And the first order condition of the $HP$ class, $\Psi_{HP}(x^*_{HP}) = 0$, reads:
    \begin{align}
        u_x(x^*_{HP} , y_{HP}) + \rho \theta F'(x^*_{HR} - x^*_{HP}) + (1-\rho)\theta F'(x^*_{LP} - x^*_{HP}) = 0. \label{eq:E.3}
    \end{align}
    Subtracting \Cref{eq:E.3} from \Cref{eq:E.2} yields:
    \begin{align*}
        \Psi_{LP}(x^*_{HP}) = \underbrace{\Big[ u_x(x^*_{HP} , y_{LP}) - u_x(x^*_{HP} , y_{HP}) \Big]}_{<\, 0}
        &+ \underbrace{\rho \theta \Big[ F'(x^*_{LR} - x^*_{HP}) - F'(x^*_{HR} - x^*_{HP}) \Big]}_{<\, 0} \\
        &- \underbrace{(1-\rho)\theta F'(x^*_{LP} - x^*_{HP})}_{\geq \, 0}.
    \end{align*}
    The first bracket is negative by (P3), as $y_{LP} < y_{HP}$ and $x^*_{HP} > 0$ (by (P1)). The second is negative because $x^*_{LR} < x^*_{HR}$ (by Claim 1) and $F'(\cdot)$ is increasing. And the third term is non-negative, and is subtracted. So $\Psi_{LP}(x^*_{HP}) < 0$, which by (P1) means $x^*_{LP} < x^*_{HP}$. Contradiction.

    \textbf{Case (iii).} Suppose $x^*_{HP} > x^*_{LP} \geq x^*_{LR}$. This is the mirror image of Case (ii), with the $HP$ and $LR$ classes interchanging roles and $\rho$ replaced by $1-\rho$ throughout.\footnote{Both facts that the argument in case (ii) relies on survive the swap: $y_{LP} < y_{LR}$ (now because the $LP$ and $LR$ classes differ in background rather than in ability), and $x^*_{HP} < x^*_{HR}$ (by Claim 1).} 
%
%
Therefore $x_{LP}^*$ must be lower than the consumption of any other class.
\end{proof}


\begin{lem}\label{lem:A2_unique_eqm}
    The unique Nash equilibrium is the solution to the following system of equations: 
    \begin{align*}
        u_x(x^*_{HR}, y_{HR}) &= 0 \\
        u_x(x^*_{HP}, y_{HP}) + \rho \theta F'(x^*_{HR} - x^*_{HP}) &= 0 \\
        u_x(x^*_{LR}, y_{LR}) + (1-\rho) \theta F'(x^*_{HR} - x^*_{LR}) &= 0 \\
        u_x(x^*_{LP}, y_{LP}) + (1-\rho)\theta F'(x^*_{HP} - x^*_{LP}) + \rho \theta F'(x^*_{LR} - x^*_{LP}) &= 0
    \end{align*}
\end{lem}
\begin{proof}
    \Cref{lem:A1_ordering} pins down which classes make (unfavorable) social comparisons with which: the $HR$ class make none; the $HP$ and $LR$ classes compare themselves only with the $HR$ class (they are not linked to each other, and they consume more than the $LP$ class); and the $LP$ class compare themselves with both.\footnote{This gives a second, constructive route to the uniqueness already established in \Cref{lem:A0_uniqueness}.}
    %
\end{proof}

\begin{lem}\label{lem:A3_limiting_cons}
    Uniformly in $\rho$: (i) as $s\to 1$, $x_{HP}^* \to x_{HR}^*$ and $x_{LR}^* \to x_{LP}^*$, and (ii) as $s\to 0$, $x_{LR}^* \to x_{HR}^*$ and $x_{HP}^* \to x_{LP}^*$.
\end{lem}
\begin{proof}
First, note that \Cref{lem:A0_uniqueness} holds for all $(s,\rho)\in[0,1]^2$, and hence so do the first-order condition ins \Cref{lem:A2_unique_eqm}. Moreover, these first-order conditions are continuous in choices and parameters, so $x_q^* \in [0, x^*_{HR}]$ and is uniformly continuous in $(s,\rho)$ for all $q\in Q, (s,\rho) \in [0,1]^2$.\footnote{More explicitly, for any sequence $(s_n,\rho_n)\to(s,\rho)$, compactness gives a convergent subsequence of equilibrium profiles. Continuity of the first-order conditions implies that its limit is an equilibrium at $(s,\rho)$; uniqueness then implies that the limit is $\mathbf{x}^*(s,\rho)$.}

    \textbf{Part (i).} At $s=1$ we have $y_{HP}=y_{HR}$. So we have: $u_x(x_{HP},y_{HR}) +\rho\theta F'(x^*_{HR}-x_{HP})=0$. Since $u_x(x^*_{HR},y_{HR})=0$ and $F'(0)=0$, $x^*_{HR}$ solves this condition. Hence, $x^*_{HP}(1,\rho)=x^*_{HR}$ for every $\rho\in[0,1]$. Additionally, at $s=1$, we have $y_{LR}=y_{LP}$, and so by an identical argument we have $x^*_{LP}(1,\rho)=x^*_{LR}(1,\rho)$ for every $\rho\in[0,1]$. Uniform continuity of the equilibrium profile then gives both convergence statements, uniformly in $\rho$.


    \textbf{Part (ii).} Follows an identical argument, with $y_{LR} \to y_{HR}$ and $y_{HP} \to y_{LP}$ as $s\to0$, and interchanging $\rho$ with $1-\rho$.
\end{proof}

\begin{lem}\label{lem:A4_marginal_value}
    Let $\mu_q := u_y(x^*_q , y_q)$ denote the marginal value of earning potential to class $q$.

    (i) For $s$ sufficiently close to $1$, $\mu_{HP} < \mu_{LR}$,
    (ii) For $s$ sufficiently close to $0$, $\mu_{HP} > \mu_{LR}$.
\end{lem}
\begin{proof}
    \textbf{Part (i).} 
    By \Cref{lem:A3_limiting_cons}, as $s\to1$ we have $x^*_{HP} \to x^*_{HR} = x^a(y_{HR})$, while $y_{HP} \to y_{HR}$. So the $HP$ class end up consuming exactly their ideal level, and $\mu_{HP} \to u_y\big(x^a(y_{HR}) , y_{HR}\big) = V'(y_{HR})$,
    where the equality is the Envelope Theorem applied to $V(y) := \max_{x\geq0} u(x,y)$.

    Also as $s\to1$ we have $y_{LR} \to y_{LP}$, and $x^*_{LR} \to \bar{x}$, where $\bar{x}$ is the unique root of $u_x(x , y_{LP}) + (1-\rho)\theta F'(x^*_{HR} - x)$ (see \Cref{lem:A3_limiting_cons}). As $\bar{x} < x^*_{HR}$, the social comparison term is strictly positive at $\bar{x}$, so $u_x(\bar{x} , y_{LP}) < 0$ and hence $\bar{x} > x^a(y_{LP})$. Therefore $\mu_{LR} \to u_y(\bar{x} , y_{LP})$, and:
    \begin{align*}
        \lim_{s\to1} \mu_{HP} = V'(y_{HR}) \ < \ V'(y_{LP}) = u_y\big(x^a(y_{LP}) , y_{LP}\big) \ < \ u_y(\bar{x} , y_{LP}) = \lim_{s\to1}\mu_{LR}.
    \end{align*}
    The first inequality holds because $V(\cdot)$ is strictly concave (\Cref{ass:u_fn}(iv)) and $y_{HR} > y_{LP}$. The second holds because $u_y(\cdot , y_{LP})$ is strictly increasing (\Cref{ass:u_fn}(iii)) and $\bar{x} > x^a(y_{LP})$. So both the difference in earning potentials and the difference in consumption push in the same direction. As the inequalities are strict, $\mu_{HP} < \mu_{LR}$ for $s$ sufficiently close to $1$.

    \textbf{Part (ii).} Follows a mirrored argument. As $s\to0$, we have $x^*_{LR} \to x^*_{HR} = x^a(y_{HR})$ and $x^*_{HP} \to \bar{x}' > x^a(y_{LP})$. Interchanging the roles of the $HP$ and $LR$ classes throughout (and replacing $\rho$ with $1-\rho$) then gives $\lim_{s\to0}\mu_{LR} = V'(y_{HR}) < V'(y_{LP}) < u_y(\bar{x}' , y_{LP}) = \lim_{s\to0}\mu_{HP}$. 
\end{proof}

\begin{lem}\label{lem:A5_bounds}
    There exist $\underline{K}>0$ and $\overline{K}< \infty$ such that, for all $s \in (0,1)$ and all $\rho \in (0,1)$: $x^*_{HP}$ is continuously differentiable in $y_{HP}$ with $\frac{d x_{HP}^*}{d y_{HP}} \in [\underline{K}, \overline{K}]$, and $x^*_{LR}$ is continuously differentiable in $y_{LR}$ with $\frac{d x_{LR}^*}{d y_{LR}} \in [\underline{K}, \overline{K}]$.
\end{lem}
\begin{proof}
We prove this in \Cref{sec:extra_proofs}.
\end{proof}

\begin{lem}\label{lem:A6_endpoint_separation}
There exist $\underline{s}, \bar{s} \in (0,1)$ with $\underline{s} < s^y < \bar{s}$ such that:

(i) for all $s < \bar{s}$: \ $\lim_{\rho \to 1} (U^*_{HP} - U^*_{LR}) < 0$ \ and \ $\lim_{\rho \to 1} (\mu_{HP} - \mu_{LR}) > 0$;

(ii) for all $s > \underline{s}$: \ $\lim_{\rho \to 0} (U^*_{HP} - U^*_{LR}) > 0$ \ and \ $\lim_{\rho \to 0} (\mu_{HP} - \mu_{LR}) < 0$.
\end{lem}


\begin{proof}
First, recall that $s^y$ is defined such that $y_{HP}(s^y) = y_{LR}(s^y)$, and $y_{HP}(s) \leq y_{LR}(s) \iff s \leq s^y$. Second, that $V(y) := \max_{x \geq 0} u(x,y)$ is strictly increasing and concave (by \Cref{ass:u_fn}(iv)). We prove part~(i); part~(ii) follows by the mirror argument described at the end. 

By \Cref{lem:A1_ordering}, each of the $HP$ and $LR$ classes makes an unfavorable social comparison with the $HR$ class and with no other class. So equilibrium utilities are
\begin{align*}
    U^*_{HP} &= u(x^*_{HP} , y_{HP}) - \rho\theta F(x^*_{HR} - x^*_{HP}), \quad \text{ and } \quad 
    U^*_{LR} &= u(x^*_{LR} , y_{LR}) - (1-\rho)\theta F(x^*_{HR} - x^*_{LR}).
\end{align*}

\paragraph{Step 1. The utility comparison.} As $\rho \to 1$, $(1-\rho)\theta \to 0$, while $F(x^*_{HR} - x^*_{LR})$ remains bounded. So $x^*_{LR} \to x^a(y_{LR})$, and hence $\lim_{\rho\to1} U^*_{LR} = V(y_{LR})$.
In contrast, $x^*_{HP}$ does not converge to $x^a(y_{HP})$ as $\rho \to 1$, because the social comparison cost does not vanish. So
\begin{align*}
    \lim_{\rho \to 1} U^*_{HP} = \underbrace{u\big( \bar{x}_{HP}(s) , y_{HP} \big)}_{< \ V(y_{HP})} - \underbrace{\theta F\big( x^*_{HR} - \bar{x}_{HP}(s) \big)}_{> \ 0} \ < \ V(y_{HP}),
\end{align*}
where $\bar{x}_{HP}(s)$ is defined (in \Cref{sec:extra_proofs}) as the unique root of $u_x(x , y_{HP}) + \theta F'(x^*_{HR} - x)$.
The first term is less than $V(y_{HP})$ because $x^a(y_{HP})$ is the \emph{unique} maximizer of $u(\cdot , y_{HP})$ and $\bar{x}_{HP}(s) \neq x^a(y_{HP})$. The second is strictly positive because $\bar{x}_{HP}(s) < x^*_{HR}$ and $F(\cdot)$ is strictly increasing with $F(0)=0$. 

Finally, note that $V(y_{HP}) - V(y_{LR}) \leq 0$ for all $s \leq s^y$, because $V(\cdot)$ is strictly increasing and $y_{HP} \leq y_{LR}$ exactly when $s \leq s^y$. So
\begin{align*}
    \lim_{\rho \to 1}(U^*_{HP} - U^*_{LR}) \ < \ V(y_{HP}) - V(y_{LR}) \ \leq \ 0 \qquad \text{for all } s \leq s^y.
\end{align*}

\paragraph{Step 2. The marginal value comparison.} The same two limits give
\begin{align*}
    \lim_{\rho\to1} \mu_{HP} = u_y\big( \bar{x}_{HP}(s) , y_{HP} \big), \qquad \lim_{\rho\to1} \mu_{LR} = u_y\big( x^a(y_{LR}) , y_{LR} \big) = V'(y_{LR}),
\end{align*}
the last equality by the Envelope Theorem. Writing $\Delta(s)$ for the difference, and splitting it into an over-consumption term and an earning potential term:
\begin{align}\label{eq:A7_split}
    \Delta(s) = \underbrace{\Big[ u_y\big(\bar{x}_{HP}(s) , y_{HP}\big) - V'(y_{HP}) \Big]}_{> \, 0}
    \ + \ \underbrace{\Big[ V'(y_{HP}) - V'(y_{LR}) \Big]}_{\geq \, 0 \ \text{ for } s \, \leq \, s^y}.
\end{align}
The first bracket is strictly positive: $V'(y_{HP}) = u_y\big(x^a(y_{HP}) , y_{HP}\big)$ by the Envelope Theorem, $\bar{x}_{HP}(s) > x^a(y_{HP})$, and $u_y(\cdot , y_{HP})$ is strictly increasing by \Cref{ass:u_fn}(iii) -- legitimate as $\bar{x}_{HP}(s) > x^a(y_{HP}) > 0$. The second is non-negative for $s \leq s^y$, because $V(\cdot)$ is strictly concave and $y_{HP} \leq y_{LR}$ exactly when $s \leq s^y$. So $\Delta(s) > 0$ for every $s \leq s^y$, again covering that whole range directly.

The two brackets pull the same way only on this range, and a threshold is genuinely needed. For $s > s^y$ the second bracket turns negative, and it eventually dominates: as $s\to1$ we have $y_{HP} \to y_{HR}$ and hence $\bar{x}_{HP}(s) \to x^*_{HR}$, so $\Delta(s) \to V'(y_{HR}) - V'(y_{LP}) < 0$.

\paragraph{Step 3. Extension past $s^y$.} Both $\lim_{\rho \to 1}(U^*_{HP} - U^*_{LR})$ and $\lim_{\rho \to 1}(\mu_{HP} - \mu_{LR})$ are continuous in $s$. 
So there must exist some $\bar{s} > s^y$ such that both $\lim_{\rho \to 1}(U^*_{HP} - U^*_{LR}) < 0$ and $\lim_{\rho \to 1}(\mu_{HP} - \mu_{LR}) < 0$ for all $s < \bar{s}$, which is part~(i).

\paragraph{Part (ii).} Follows the same argument, mirrored: interchange the roles of the $HP$ and $LR$ classes, replace $\rho$ with $1-\rho$, and note that the comparison of earning potentials reverses, as $y_{LR} \leq y_{HP}$ exactly when $s \geq s^y$. 
\end{proof}

\subsection{Proof of \Cref{rem:unique_eqm}}
Uniqueness is \Cref{lem:A0_uniqueness}. Characterization follows from \Cref{lem:A2_unique_eqm}. \hfill \qed  

\subsection{Proof of Remarks \ref{rem:indiv_comp_stat_x_s} and \ref{rem:indiv_comp_stat_x_rho}}
Throughout, we use the equilibrium first order conditions from \Cref{lem:A2_unique_eqm}, and write $u_{xx,q}$ and $u_{xy,q}$ for the second derivatives of $u(\cdot,\cdot)$ evaluated at $(x^*_q , y_q)$.

\textbf{HR class.} $x_{HR}^*$ is the unique solution to $u_x(x_{HR} , y_{HR}) = 0$. None of the terms depend on $s$ or on $\rho$. So $\frac{d x_{HR}^*}{ds} = 0$ and $\frac{d x_{HR}^*}{d\rho} = 0$.
\textbf{HP and LR classes.} $x_{HP}^*$ and $x_{LR}^*$ are, respectively, the unique solutions to:
\begin{align}
    \Phi_{HP} &:= u_x(x_{HP} , y_{HP}) + \theta\rho F'(x_{HR}^* - x_{HP}) = 0, \\
    \Phi_{LR} &:= u_x(x_{LR} , y_{LR}) + \theta(1 -\rho) F'(x_{HR}^* - x_{LR}) = 0.
\end{align}
As $x^*_{HR}$ is invariant in both $s$ and $\rho$, each of these depends on $s$ only through its own earning potential, and on $\rho$ only through its own comparison weight. So, respectively, we have:
\begin{align*}
    \frac{\partial \Phi_{HP}}{\partial x_{HP}} &= u_{xx,HP} - \theta\rho\, F''_{HP,HR} &< 0
    \quad &, \quad
    \frac{\partial \Phi_{LR}}{\partial x_{LR}} = u_{xx,LR} - \theta(1-\rho)\, F''_{LR,HR} &< 0 \\
    \frac{\partial \Phi_{HP}}{\partial s} &= u_{xy,HP} \ \frac{dy_{HP}}{ds} &> 0
    \quad &, \quad
    \frac{\partial \Phi_{LR}}{\partial s} = u_{xy,LR} \ \frac{dy_{LR}}{ds} &< 0 \\
    \frac{\partial \Phi_{HP}}{\partial \rho} &= \theta\, F'(x^*_{HR} - x^*_{HP}) &> 0
    \quad &, \quad
    \frac{\partial \Phi_{LR}}{\partial \rho} = -\theta\, F'(x^*_{HR} - x^*_{LR}) &< 0
\end{align*}
where the signs follow from \Cref{ass:u_fn}(i) and (iii) (evaluated at $x^*_q > 0$, by (P1)), the strict convexity of $F(\cdot)$ together with $x^*_{HR} > x^*_{HP} , x^*_{LR}$ (\Cref{lem:A1_ordering}), and $\frac{dy_{HP}}{ds} = - \frac{dy_{LR}}{ds} > 0$.
Then by the Implicit Function Theorem, we have:
\begin{align*}
    \frac{dx^*_{HP}}{ds} = - \frac{\partial \Phi_{HP} / \partial s}{\partial\Phi_{HP} / \partial x_{HP}} > 0
    \qquad \text{and} \qquad
    \frac{dx^*_{HP}}{d\rho} = - \frac{\partial \Phi_{HP} / \partial \rho}{\partial\Phi_{HP} / \partial x_{HP}} > 0 \\
    \frac{dx^*_{LR}}{ds} = - \frac{\partial \Phi_{LR} / \partial s}{\partial\Phi_{LR} / \partial x_{LR}} < 0
    \qquad \text{and} \qquad
    \frac{dx^*_{LR}}{d\rho} = - \frac{\partial \Phi_{LR} / \partial \rho}{\partial\Phi_{LR} / \partial x_{LR}} < 0
\end{align*}
\textbf{LP class.} Using a similar argument, it is straightforward to show that $\frac{d \Phi_{LP}}{d x_{LP}} < 0$. So, by the Implicit Function Theorem, $\text{sign}\big[\frac{dx^*_{LP}}{ds}\big] = \text{sign}\big[\frac{\partial \Phi_{LP}}{\partial s}\big]$ and $\text{sign}\big[\frac{dx^*_{LP}}{d\rho}\big] = \text{sign}\big[\frac{\partial \Phi_{LP}}{\partial \rho}\big]$. The expressions in part (iv) of each of Remarks \ref{rem:indiv_comp_stat_x_s} and \ref{rem:indiv_comp_stat_x_rho} are essentially those of $\frac{\partial \Phi_{LP}}{\partial s}$ and $\frac{\partial \Phi_{LP}}{\partial \rho}$, respectively. We provide full workings in \Cref{sec:extra_proofs}. \hfill \qed

\subsection{Proof of \Cref{prop:convex_comps_result}}
Given \Cref{lem:A1_ordering}, we can write equilibrium utilities as follows: $U^*_{HR} = u(x^*_{HR} , y_{HR})$ and
\begin{align*}
    U^*_{HP} &= u(x^*_{HP} , y_{HP}) - \rho \theta F(x_{HR}^* - x_{HP}^*) \quad , \quad 
    U^*_{LR} = u(x^*_{LR} , y_{LR}) - (1-\rho)\theta F(x_{HR}^* - x_{LR}^*) \\
    U^*_{LP} &= u(x^*_{LP} , y_{LP}) - (1-\rho)\theta F(x_{HP}^* - x_{LP}^*) - \rho\theta F(x_{LR}^* - x_{LP}^*) 
\end{align*}
Then recall that, by assumption, only $y_{HP}$ and $y_{LR}$ are directly functions of $s$, and additionally $\frac{d y_{HP}}{d s} = - \frac{d y_{LR}}{d s}$. So, applying this equality and the Envelope Theorem, we have $\frac{d U_{HR}^*}{ds} = 0$ and, writing $\mu_q := u_y(x^*_q , y_q)$ for the marginal value of earning potential (as in \Cref{lem:A4_marginal_value}):
\begin{align*}
    \frac{d U_{HP}^*}{ds} &= \mu_{HP} \ \frac{d y_{HP}}{d s}
    \qquad , \qquad
    \frac{d U_{LR}^*}{ds} = - \, \mu_{LR} \ \frac{d y_{HP}}{d s} \\
    \frac{d U_{LP}^*}{ds} &= - (1-\rho)\theta F'(x_{HP}^* - x_{LP}^*) \frac{d x_{HP}^*}{d y_{HP}} \frac{d y_{HP}}{d s} + \rho \theta F'(x_{LR}^* - x_{LP}^*) \frac{d x_{LR}^*}{d y_{LR}} \frac{d y_{HP}}{d s}
\end{align*}

\paragraph{Decomposing welfare changes.} The impact of a change in economic mobility on social welfare decomposes to:
\begin{align}
    \frac{d W}{d s} = \frac{d W}{d U_{HR}} \frac{d U_{HR}}{d s} + \frac{d W}{d U_{HP}} \frac{d U_{HP}}{d s} + \frac{d W}{d U_{LR}} \frac{d U_{LR}}{d s} + \frac{d W}{d U_{LP}} \frac{d U_{LP}}{d s}.
\end{align}
Note that all agents in a class $q$ have the same equilibrium utility, so $\frac{d W}{d U_{q}} := \ \sum_{i \in q} \frac{\partial W}{\partial U_i} \in (0,\infty)$. The Pigou-Dalton Principle carries over to class-level utilities, so that $U^*_q < U^*_{q'} \implies \frac{d W}{d U_q} \geq \frac{d W}{d U_{q'}}$.
Next, we substitute in the expressions for $\frac{d U_q}{d s}$ for $q \in \{HR, HP, LR, LP\}$. Hence we have:
\begin{align}\label{eq:convex_utilitarian_derivative}
    \frac{d W}{d s} &= \overbrace{\left[ \frac{d W}{d U_{HP}} \ \mu_{HP}
    - \frac{d W}{d U_{LR}} \ \mu_{LR}   \right] }^{\textbf{(A)}} \frac{d y_{HP}}{d s} \notag \\
    &\qquad 
    + \frac{d W}{d U_{LP}}  \underbrace{\left[  \rho F'(x_{LR}^* - x_{LP}^*) \frac{d x_{LR}^*}{d y_{LR}} - (1-\rho) F'(x_{HP}^* - x_{LP}^*) \frac{d x_{HP}^*}{d y_{HP}} \right]}_{\textbf{(B)}} \theta \frac{d y_{HP}}{d s}
\end{align}

\paragraph{Characterizing the expression.} We now examine each part of this expression in turn. 

\noindent \textbf{Part (A).} \textbf{For $s$ sufficiently close to $0$}, we have: (i) $\mu_{HP} > \mu_{LR}$, by \Cref{lem:A4_marginal_value}; and (ii) $\frac{d W}{d U_{HP}} \geq \frac{d W}{d U_{LR}}$. All four quantities are strictly positive -- the $\mu_q$ by \Cref{ass:u_fn}(ii), and the welfare weights by assumption -- so
$\frac{d W}{d U_{HP}} \mu_{HP} \geq \frac{d W}{d U_{LR}} \mu_{HP} > \frac{d W}{d U_{LR}} \mu_{LR}$, and hence $\textbf{(A)} > 0$.

Claim (ii) follows from the fact that for $s$ sufficiently close to $0$, $U_{HP}^* < U_{LR}^*$, coupled with the Pigou-Dalton principle. To see that $U_{HP}^* < U_{LR}^*$, recall from the proof of \Cref{lem:A6_endpoint_separation} that $V(y) := \max_{x\geq0}u(x,y)$ is strictly increasing. As $s\to0$ we have $y_{LR} \to y_{HR}$ and $x^*_{LR} \to x^*_{HR} = x^a(y_{HR})$ (\Cref{lem:A3_limiting_cons}), so the consumption of the $LR$ class converges to their ideal level while their social comparison cost vanishes: $U^*_{LR} \to V(y_{HR})$. Meanwhile $y_{HP} \to y_{LP}$, and the limiting value of $x^*_{HP}$ is strictly above $x^a(y_{LP})$ (shown in the proof of \Cref{lem:A4_marginal_value}). As $x^a(y_{LP})$ is the \emph{unique} maximizer of $u(\cdot , y_{LP})$, and social comparison costs are non-negative, $U^*_{HP}$ therefore tends to a limit strictly below $V(y_{LP})$. And $V(y_{LP}) < V(y_{HR})$, so $U^*_{HP} < U^*_{LR}$ for $s$ sufficiently close to $0$.
\textbf{For $s$ sufficiently close to $1$}, the argument proceeds in an identical fashion, 
yielding $\textbf{(A)} < 0$ for $s$ sufficiently close to $1$.

\paragraph{Part (B).} 
As $s\to1$, we have $x_{HP}^* \to x_{HR}^*$ and $x_{LR}^* \to x_{LP}^*$ (by \Cref{lem:A3_limiting_cons}). Hence $F'(x^*_{LR} - x^*_{LP}) \to F'(0) = 0$ and $F'(x^*_{HP} - x^*_{LP}) \to F'(x^*_{HR} - x^*_{LP}) > 0$ as $s\to1$. The bounds on $\frac{d x_{LR}^*}{d y_{LR}}$ and $\frac{d x_{HP}^*}{d y_{HP}}$ from \Cref{lem:A5_bounds} then guarantee that the first term in \textbf{(B)} converges to $0$ as $s \to 1$, and that the second term in \textbf{(B)} remains strictly positive as $s\to1$. Therefore for $s$ sufficiently close to $1$, we have \textbf{(B)}$<0$. 
By an identical argument, we have $\textbf{(B)} > 0$ for $s$ sufficiently close to $0$.

\paragraph{Putting things together.} For $s$ sufficiently close to $1$, we have $\textbf{(A)}<0$ and $\textbf{(B)}<0$. Then notice that all other terms in \Cref{eq:convex_utilitarian_derivative} are positive by assumption: $\frac{d y_{HP}}{d s}> 0, \frac{dW}{d U_{LP}} > 0, \theta>0$. Therefore we have $\frac{dW}{ds}< 0$ for $s$ sufficiently close to $1$.
An identical argument then shows that we have $\frac{dW}{ds} > 0$ for $s$ sufficiently close to $0$. This completes the proof. \hfill \qed

\subsection{Proof of \Cref{rem:income_inequality}}
Recall that by assumption: $y_{HR} > y_{HP}, y_{LR} > y_{LP}$, that $y_{HR}, y_{LP}$ do not change with $s$, and $\frac{d y_{HP}}{d s} = - \frac{d y_{LR}}{d s} > 0$. Additionally, there exists a unique $s^y$ such that $y_{HP} \geq y_{LR}$ iff $s \geq s^y$ (and $y_{HP} = y_{LR}$ iff $s = s^y$). 

It therefore follows that for any $s, s'$ with $s' > s \geq s^y$, $y(s')$ majorizes $y(s)$. And hence for any Schur-convex inequality measure $I(\cdot)$, $I(y(s')) \geq I(y(s))$ \citep{marshall1979inequalities}. The same argument applies for $s, s' \leq s^y$, with the inequality reversed. \hfill \qed

\subsection{Proof of \Cref{prop:welfare_rho}}
\paragraph{Step 1: how individual utilities change with integration.} Taking the characterization of equilibrium utilities from the proof to \Cref{prop:convex_comps_result} and differentiating with respect to $\rho$, we have $\frac{d U_{HR}^*}{d \rho} = 0$ and
\begin{align*}
    \frac{dU^*_{HP}}{d\rho} = -\theta F(x_{HR}^* - x^*_{HP}) 
    \quad , \quad
    \frac{dU^*_{LR}}{d\rho} = \theta F(x_{HR}^* - x^*_{LR}) .
\end{align*}

The $LP$ class are different, because the two classes it looks up at both adjust their own consumption when $\rho$ changes. While the Envelope Theorem again removes the direct term in $dx^*_{LP}/d\rho$, it does not remove these. Writing $G_{HP} := x^*_{HP} - x^*_{LP}$ and $G_{LR} := x^*_{LR} - x^*_{LP}$ for the two gaps the $LP$ class looks up at,
\begin{align*}
    \frac{dU^*_{LP}}{d\rho} = \underbrace{\theta\Big[ F(G_{HP}) - F(G_{LR}) \Big]}_{\text{direct effect}}
    \ - \ \underbrace{\theta\left[ (1-\rho) F'(G_{HP}) \frac{dx^*_{HP}}{d\rho} + \rho F'(G_{LR}) \frac{dx^*_{LR}}{d\rho} \right]}_{\text{indirect effect}} .
\end{align*}
This is the same distinction drawn in the proof to \Cref{prop:convex_comps_result}, where the corresponding indirect terms in $dx^*_{HP}/ds$ and $dx^*_{LR}/ds$ are carried explicitly.

\paragraph{Aside.} Differentiating the first order conditions of the $HP$ and $LR$ classes with respect to $\rho$, exactly as in the proof of Remarks \ref{rem:indiv_comp_stat_x_s} and \ref{rem:indiv_comp_stat_x_rho}, we get:
\begin{align}\label{eq:rho_response}
    \frac{dx^*_{HP}}{d\rho} = \frac{\theta F'(x^*_{HR} - x^*_{HP})}{-u_{xx,HP} + \theta\rho \, F''_{HP,HR}}
    \ , \qquad
    \frac{dx^*_{LR}}{d\rho} = \frac{-\theta F'(x^*_{HR} - x^*_{LR})}{-u_{xx,LR} + \theta(1-\rho) \, F''_{LR,HR}}.
\end{align}
As $y_q$ and $x_q^*$ for all $q$ lie in compact intervals, and $u_{xx}<0$ is continuous, there exists some $\zeta>0$ such that $-u_{xx}(x_q^*, y_q)\geq \zeta$. Moreover, $F'$ is bounded on the relevant compact interval, and $F'' \geq 0$. So the partial derivatives in \Cref{eq:rho_response} are bounded in magnitude by some constant. Then the uniform convergence established in \Cref{lem:A3_limiting_cons}, coupled with $F'(0) = 0$ gives 
\begin{align}\label{eq:marginal_cons_rho_limit}
    \sup_{\rho\in[0,1]}
 \left|\frac{\partial x^*_{LR}}{\partial\rho}\right|\to0
 \quad\text{as }s\to0
 \quad \text{ and } \quad 
 \sup_{\rho\in[0,1]}
 \left|\frac{\partial x^*_{HP}}{\partial\rho}\right|\to0
 \quad\text{as }s\to1.
\end{align}

\paragraph{Step 1a: when mobility is low.} Let $s \to 0$. By \Cref{lem:A3_limiting_cons}, we have $x^*_{LR} \to x^*_{HR}$ and $x^*_{HP} \to x^*_{LP}$ uniformly in $\rho$, with $x^*_{HR} > x^*_{LP}$ for all $(s,\rho) \in [0,1]^2$. This has several implication: \textbf{(A)} $G_{HP} \to 0$, which implies that \textbf{(B)} $F(G_{HP}) \to F(0) = 0$ and \textbf{(C)} $F'(G_{HP}) \to F'(0) = 0$; all uniformly in $\rho$. Additionally \textbf{(D)} $G_{LR} := x_{LR}^* - x_{LP}^* \to \kappa_1$ for some $\kappa_1 > 0$ and \textbf{(E)} $x^*_{HR} - x_{HP}^* \to \kappa_2$, for some $\kappa_2 > 0$, both uniformly in $\rho$.

Observation \textbf{(E)} yields $\lim_{s \to 0} d U^*_{HP} / d \rho = - \kappa_3$ for some $\kappa_3 > 0$, and observation \textbf{(D)} yields $\lim_{s \to 0} d U^*_{LR} / d \rho = 0$, both uniformly in $\rho$. Then combining observation \textbf{(C)} and the first expression in \Cref{eq:marginal_cons_rho_limit} yields that the `indirect effect' in $d U^*_{LP} / d \rho$ converges to $0$ as $s \to 0$, and does so uniformly in $\rho$. Then combining observations \textbf{(A)} and \textbf{(D)} yields that the `direct effect' has $\lim_{s \to 0} d U^*_{LP} / d \rho \to -\kappa_4$ for some $\kappa_4 > 0$, uniformly in $\rho$.



\paragraph{Step 1b: when mobility is high.} Let $s \to 1$. Everything is mirrored, and the argument of step 1a applies after interchanging the $HP$ and $LR$ classes. 

\paragraph{Step 2: putting things together.} Let $w_q := \frac{\partial W}{\partial U_q}$. By assumption, there exist $0<\underline{m}\leq \overline{m} < \infty$ such that $w_q \in [\underline{m}, \overline{m}]$ for all $q$. Then we have:
\begin{align}\label{eq:welfare_deriv_rho}
    \frac{\partial W}{\partial \rho} = \sum_{q \in Q} w_q \frac{\partial U^*_{q}}{\partial \rho}.
\end{align}

\paragraph{Step 2a: when mobility is low.} From step 1a, we have:
\begin{align*}
    \text{as } s \to 0: \frac{\partial U^*_{LR}}{\partial \rho} = 0 \text{ and } \frac{\partial U^*_{HP}}{\partial \rho} < 0 \text{ and } \frac{\partial U^*_{LP}}{\partial \rho} < 0
\end{align*}
uniformly for every $\rho \in [0,1]$ So the overall expression $\lim_{s \to 0} \frac{\partial W}{\partial \rho}$ is negative for every $\rho \in [0,1]$. Hence for $s$ close enough to $0$, $\frac{d W}{d \rho} < 0$.

\paragraph{Step 2b: when mobility is high.} From step 1b, we have:
\begin{align*}
    \text{as } s \to 1: \frac{\partial U^*_{LR}}{\partial \rho} > 0 \text{ and } \frac{\partial U^*_{HP}}{\partial \rho} = 0 \text{ and } \frac{\partial U^*_{LP}}{\partial \rho} > 0
\end{align*}
uniformly for every $\rho \in [0,1]$ So the overall expression $\lim_{s \to 0} \frac{\partial W}{\partial \rho}$ is negative for every $\rho \in [0,1]$. Hence for $s$ close enough to $0$, $\frac{d W}{d \rho} < 0$. 
\hfill \qed

\subsection{Proof of \Cref{prop:welfare_intermediate_mobility_general}}
Throughout, fix $s \in (0,1)$ and consider the behavior as $\rho \to 1$ (for part~(i)) or $\rho \to 0$ (for part~(ii)).

\paragraph{Step 1: how welfare changes with mobility.} 
We use the same decomposition of $dW/ds$ as in the proof of \Cref{prop:convex_comps_result}:
\begin{align*}
    \frac{d W}{d s} &= \overbrace{\left[ \frac{d W}{d U_{HP}} \ \mu_{HP}
    - \frac{d W}{d U_{LR}} \ \mu_{LR} \right] }^{\textbf{(A)}} \frac{d y_{HP}}{d s} \\
    &\qquad
    + \frac{d W}{d U_{LP}}  \underbrace{\left[  \rho F'(x_{LR}^* - x_{LP}^*) \frac{d x_{LR}^*}{d y_{LR}} - (1-\rho) F'(x_{HP}^* - x_{LP}^*) \frac{d x_{HP}^*}{d y_{HP}} \right]}_{\textbf{(B)}} \theta \frac{d y_{HP}}{d s}
\end{align*}
where $\mu_q := u_y(x^*_q , y_q)$ (as in \Cref{lem:A4_marginal_value}).

\paragraph{Step 2: characterizing the expression.} The remainder of the proof amounts to characterizing the sign of this expression. We now examine each part in turn. Note that $\underline{s}$ and $\bar{s}$ are the thresholds from \Cref{lem:A6_endpoint_separation}.

\paragraph{Part (A).} \textbf{Case $\rho \to 1$, with $s < \bar{s}$.}
By \Cref{lem:A6_endpoint_separation}(i), $\lim_{\rho\to1}(\mu_{HP} - \mu_{LR}) > 0$, so $\mu_{HP} > \mu_{LR}$ for $\rho$ sufficiently close to $1$; and, from the same part, $U^*_{HP} < U^*_{LR}$ for $\rho$ sufficiently close to $1$, so the Pigou-Dalton Principle gives $\frac{dW}{dU_{HP}} \geq \frac{dW}{dU_{LR}}$. All four quantities are strictly positive, the $\mu_q$ by \Cref{ass:u_fn}(ii), and the welfare weights by assumption. So:
    $\frac{dW}{dU_{HP}} \ \mu_{HP} \ \geq \ \frac{dW}{dU_{LR}} \ \mu_{HP} \ 
    > \ \frac{dW}{dU_{LR}} \ \mu_{LR}$,
and hence $\textbf{(A)} > 0$. 

\smallskip
\noindent\textbf{Case $\rho \to 0$, with $s > \underline{s}$.} The argument is identical, mirrored. By \Cref{lem:A6_endpoint_separation}(ii) we now have $\mu_{HP} < \mu_{LR}$ and $U^*_{HP} > U^*_{LR}$, so Pigou-Dalton gives $\frac{dW}{dU_{HP}} \leq \frac{dW}{dU_{LR}}$. Combining these gives $\textbf{(A)} < 0$.

\medskip
\paragraph{Part (B).} Fix $s \in (0,1)$. We show that, as $\rho \to 1$, the second term of \textbf{(B)} vanishes while the first stays bounded away from zero; the case $\rho \to 0$ is the mirror image. 

\smallskip
\noindent\textbf{Case $\rho \to 1$: the vanishing term.} Its coefficient is $(1-\rho) \to 0$, and the two factors it multiplies are bounded. First, $d x^*_{HP} / d y_{HP} \leq \overline{K}$ by \Cref{lem:A5_bounds}. Second, $0 < x^*_{HP} - x^*_{LP} < x^*_{HR}$ -- by \Cref{lem:A1_ordering} and $x^*_{LP} > 0$. So $F'(x^*_{HP} - x^*_{LP}) < F'(x^*_{HR})$, as $F'(\cdot)$ is increasing. The second term of \textbf{(B)} therefore tends to zero.

\smallskip
\noindent\textbf{Case $\rho \to 1$: the surviving term.} Here we need the gap $x^*_{LR} - x^*_{LP}$ to stay away from zero, so we take the limit of each consumption level in turn. From the limits recorded in \Cref{sec:extra_proofs}, $x^*_{LR} \to x^a(y_{LR})$. For the $LP$ class, recall their first order condition,
\begin{align*}
    u_x(x , y_{LP}) + (1-\rho)\theta F'(x^*_{HP} - x) + \rho\theta F'(x^*_{LR} - x) = 0 .
\end{align*}
As $\rho \to 1$ the weight on the $HP$ class vanishes -- and the $F'$ term it multiplies is bounded, as we have just seen -- while the weight on the $LR$ class tends to $\theta$. So the left hand side converges, uniformly in $x$, to
\begin{align*}
    \phi_{LP}(x) := u_x(x , y_{LP}) + \theta F'\big( \max\{ x^a(y_{LR}) - x , \, 0 \} \big) ,
\end{align*}
where the $\max$ operator is now written explicitly because we shall evaluate $\phi_{LP}$ at a point where the gap closes. By the same continuity-of-the-root argument used in \Cref{lem:A3_limiting_cons}, $x^*_{LP}$ converges to the root of $\phi_{LP}$, which we write $\hat{x}_{LP}$.

It remains to place $\hat{x}_{LP}$ strictly below $x^a(y_{LR})$. Note first that $\phi_{LP}$ is strictly decreasing: $u_x(\cdot , y_{LP})$ is strictly decreasing by \Cref{ass:u_fn}(i), and the comparison term is non-increasing in $x$. Now evaluate at $x = x^a(y_{LR})$, where the comparison term vanishes as $F'(0)=0$:
\begin{align*}
    \phi_{LP}\big( x^a(y_{LR}) \big) \ = \ u_x\big( x^a(y_{LR}) , y_{LP} \big) \ < \ u_x\big( x^a(y_{LR}) , y_{LR} \big) \ = \ 0 ,
\end{align*}
with the inequality by \Cref{ass:u_fn}(iii) with $y_{LP} < y_{LR}$, and the final equality by the definition of $x^a(y_{LR})$. Since $\phi_{LP}(0) \geq u_x(0 , y_{LP}) > 0$ (that $u_x(0,y)>0$ is implied by \Cref{ass:u_fn}(i)), this one evaluation does two jobs at once: it shows the root exists and is unique, and it locates it, giving $\hat{x}_{LP} \in \big( 0 , x^a(y_{LR}) \big)$.

So $x^*_{LR} - x^*_{LP} \to x^a(y_{LR}) - \hat{x}_{LP} > 0$, and $F'(x^*_{LR} - x^*_{LP})$ converges to a strictly positive limit.\footnote{This is where the order of limits matters. The gap $x^*_{LR} - x^*_{LP}$ is strictly positive for each fixed $s \in (0,1)$, but not uniformly in $s$: it vanishes as $s \to 1$ (\Cref{lem:A3_limiting_cons}). Fixing $s$ first, and only then sending $\rho$ to its limit, is what makes the argument work.} Together with $\rho \to 1$ and $d x^*_{LR} / d y_{LR} \geq \underline{K} > 0$, the first term of \textbf{(B)} is bounded away from zero. Combining the two halves, $\textbf{(B)} > 0$ for $\rho$ sufficiently close to $1$.

\smallskip
\noindent\textbf{Case $\rho \to 0$.} The mirror image: interchange the $HP$ and $LR$ classes and replace $\rho$ by $1-\rho$. It is now $x^*_{HP} \to x^a(y_{HP})$ that pins down the surviving comparison, and the argument above delivers $\lim x^*_{LP} < x^a(y_{HP})$ -- using $y_{LP} < y_{HP}$ in place of $y_{LP} < y_{LR}$, both of which are strict for every $s \in (0,1)$. So the roles of the two terms of \textbf{(B)} are exchanged: the first vanishes and the second is bounded away from zero, giving $\textbf{(B)} < 0$ for $\rho$ sufficiently close to $0$.

\paragraph{Putting things together.} For $\rho \to 1$, with $s < \bar{s}$, we have $\textbf{(A)}>0$ and $\textbf{(B)}>0$. All other terms in \Cref{eq:convex_utilitarian_derivative} are positive by assumption: $\frac{d y_{HP}}{d s}> 0, \frac{dW}{d U_{LP}} > 0, \theta>0$. Therefore we have $\frac{dW}{ds}> 0$. An identical argument then shows that for $\rho \to 0$, with $s > \underline{s}$, we have $\textbf{(A)}<0$ and $\textbf{(B)}<0$ and hence $\frac{dW}{ds} < 0$. As \Cref{lem:A6_endpoint_separation} delivers $\underline{s} < s^y < \bar{s}$, this completes the proof. \hfill \qed

\newpage
\begin{center}
    \section*{Online Appendix}\label{OA}
\end{center}
\setcounter{figure}{0} \renewcommand{\thefigure}{OA\arabic{figure}}
\setcounter{table}{0} \renewcommand{\thetable}{OA\arabic{table}}

\section{Supplementary Proofs}\label{sec:extra_proofs}
\subsection{Proof of \Cref{lem:A5_bounds}}
\begin{lem*}
    There exist $\underline{K}>0$ and $\overline{K}< \infty$ such that, for all $s \in (0,1)$ and all $\rho \in (0,1)$: $x^*_{HP}$ is continuously differentiable in $y_{HP}$ with $\frac{d x_{HP}^*}{d y_{HP}} \in [\underline{K}, \overline{K}]$, and $x^*_{LR}$ is continuously differentiable in $y_{LR}$ with $\frac{d x_{LR}^*}{d y_{LR}} \in [\underline{K}, \overline{K}]$.
\end{lem*}

\paragraph{Proof.} We give the argument for the $HP$ class. That for the $LR$ class is identical, with $\rho$ replaced by $1-\rho$: the constants we construct below hold for every value of $\rho\theta \in (0,\theta)$, and so serve for both classes.

    \textbf{Step 1: the relevant First Order Condition.} By \Cref{lem:A1_ordering} we have $x^*_{LP} < x^*_{HP} < x^*_{HR}$, with both inequalities strict for every $s \in (0,1)$. So the $HP$ class make an unfavorable social comparison with the $HR$ class and with no other class, and this remains true for small perturbations of $y_{HP}$. Their equilibrium consumption is therefore the unique root of
    \begin{align*}
        \Phi(x , y) := u_x(x , y) + \rho\theta F'(x^*_{HR} - x),
    \end{align*}
    where $x^*_{HR} = x^a(y_{HR})$ is a constant (by \Cref{lem:A2_unique_eqm}). 

    \textbf{Step 2: the derivative.} $\Phi(\cdot,\cdot)$ is continuously differentiable, with $\Phi_x = u_{xx}(x,y) - \rho\theta F''(x^*_{HR}-x) < 0$ and $\Phi_y = u_{xy}(x,y)$. As $\Phi_x \neq 0$, the Implicit Function Theorem applies: $x^*_{HP}$ is continuously differentiable in $y_{HP}$, with
    \begin{align*}
        \frac{d x^*_{HP}}{d y_{HP}} = - \frac{\Phi_y}{\Phi_x} = \frac{u_{xy}(x^*_{HP} , y_{HP})}{- u_{xx}(x^*_{HP} , y_{HP}) + \rho\theta F''(x^*_{HR} - x^*_{HP})} \ > \ 0.
    \end{align*}

    \textbf{Step 3: bounding it, uniformly in $s$ and $\rho$.} It remains to bound the numerator and the denominator above, and away from zero. Every quantity in them is evaluated on a compact set that does not move with $s$ or with $\rho$. Indeed, for every class $q$:
    \begin{align*}
        x^a(y_{LP}) \ \leq \ x^a(y_q) \ \leq \ x^*_q \ \leq \ x^*_{HR} = x^a(y_{HR}), \qquad y_{LP} \ \leq \ y_q \ \leq \ y_{HR},
    \end{align*}
    where the first inequality holds because $x^a(\cdot)$ is increasing, the second is (P2), and the third is \Cref{lem:A1_ordering}. And neither $y_{LP}$ nor $y_{HR}$ varies with $s$. So $D := [x^a(y_{LP}) , x^a(y_{HR})] \times [y_{LP} , y_{HR}]$ is compact and contains $(x^*_q , y_q)$ for every class $q$, every $s$, and every $\rho$.

    On $D$, both $u_{xy}$ and $-u_{xx}$ are continuous and strictly positive -- the former by \Cref{ass:u_fn}(iii), using $x \geq x^a(y_{LP}) > 0$ everywhere on $D$, and the latter by \Cref{ass:u_fn}(i). A continuous and strictly positive function on a compact set attains a strictly positive minimum and a finite maximum, so there are constants with $0 < \underline{a} \leq u_{xy} \leq \overline{a} < \infty$ and $0 < \underline{c} \leq -u_{xx} \leq \overline{c} < \infty$ on $D$. Similarly, $F''(\cdot)$ is continuous and so bounded above by some $\overline{f} < \infty$ on the compact interval $[0 , x^a(y_{HR})]$, which contains every gap $x^*_{HR} - x^*_q$. Finally, $\rho\theta \in (0,\theta)$, and $\theta$ is a fixed parameter. So the denominator lies in $[\underline{c} , \overline{c} + \theta \overline{f}]$, and
    \begin{align*}
        \frac{d x^*_{HP}}{d y_{HP}} \ \in \ \left[ \frac{\underline{a}}{\overline{c} + \theta \overline{f}} \ , \ \frac{\overline{a}}{\underline{c}} \right] \ =: \ [\underline{K} , \overline{K}],
    \end{align*}
    where both constants are strictly positive, finite, and independent of $s$ and of $\rho$. \hfill \qed 

\subsection{Additional Claims}
\paragraph{Limiting consumption as integration approaches its bounds.} For $q \in \{HP , LR\}$, let $\bar{x}_q(s)$ denote the unique root of
\begin{align*}
    u_x(x , y_q) + \theta F'(x^*_{HR} - x),
\end{align*}
the consumption class $q$ would choose if it bore the \emph{entire} weight $\theta$ of the comparison with the $HR$ class. This root lies strictly inside $\big( x^a(y_q) , x^*_{HR} \big)$: at $x = x^a(y_q)$ the expression equals $\theta F'\big(x^*_{HR} - x^a(y_q)\big) > 0$, because $x^*_{HR} = x^a(y_{HR}) > x^a(y_q)$; and at $x = x^*_{HR}$ it equals $u_x(x^*_{HR} , y_q) < 0$.\footnote{Note that these bounds come from this First Order Condition itself, and not from \Cref{lem:A1_ordering} -- which does not apply, as $\bar{x}_q(s)$ is not the equilibrium consumption of class $q$ for any $\rho \in (0,1)$.}

The First Order Conditions of the $HP$ and $LR$ classes are $u_x(x , y_{HP}) + \rho\theta F'(x^*_{HR} - x) = 0$ and $u_x(x , y_{LR}) + (1-\rho)\theta F'(x^*_{HR} - x) = 0$ respectively. So, by continuity of the root -- exactly as in the proof of \Cref{lem:A3_limiting_cons} -- we have
\begin{align*}
    \text{as } \rho \to 1: \quad x^*_{HP} \to \bar{x}_{HP}(s) , \quad x^*_{LR} \to x^a(y_{LR}) ;
    \qquad\qquad
    \text{as } \rho \to 0: \quad x^*_{HP} \to x^a(y_{HP}) , \quad x^*_{LR} \to \bar{x}_{LR}(s) .
\end{align*}
In words: as $\rho \to 1$ the $HP$ class come to bear the whole of the comparison with the $HR$ class while the $LR$ class are released from theirs, and the reverse as $\rho \to 0$. Throughout what follows, $s \in (0,1)$ is held fixed and $\rho$ is then sent to its limit, so how close $\rho$ must be may depend on $s$.

\paragraph{Extending a strict inequality past $s^y$.} Here, we establish strict signs on the whole of $(0 , s^y]$ and then extends them to a neighborhood beyond. We record that extension once, as it is used four times.

\medskip
\noindent\textbf{Extension Claim.} \emph{Let $\Phi : (0,1) \to \mathbb{R}$ be continuous. If $\Phi(s) > 0$ for every $s \leq s^y$, then}
\begin{align*}
    \bar{s}(\Phi) := \min\Big\{ \ \inf\big\{ s \in (s^y , 1) : \Phi(s) \leq 0 \big\} \ , \ \tfrac{1+s^y}{2} \ \Big\}, \qquad \text{with } \inf \emptyset := 1 ,
\end{align*}
\emph{satisfies $\bar{s}(\Phi) \in (s^y , 1)$ and $\Phi(s) > 0$ for every $s < \bar{s}(\Phi)$. Symmetrically, if $\Phi(s) > 0$ for every $s \geq s^y$, then}
\begin{align*}
    \underline{s}(\Phi) := \max\Big\{ \ \sup\big\{ s \in (0 , s^y) : \Phi(s) \leq 0 \big\} \ , \ \tfrac{s^y}{2} \ \Big\}, \qquad \text{with } \sup \emptyset := 0 ,
\end{align*}
\emph{satisfies $\underline{s}(\Phi) \in (0 , s^y)$ and $\Phi(s) > 0$ for every $s > \underline{s}(\Phi)$.}

\begin{proof}
    As $\Phi(s^y) > 0$ and $\Phi$ is continuous, $\Phi > 0$ on $[s^y , s^y + \eta)$ for some $\eta > 0$. So the infimum strictly exceeds $s^y$, giving $\bar{s}(\Phi) > s^y$; and $\bar{s}(\Phi) < 1$ by the second entry. Now take any $s < \bar{s}(\Phi)$: if $s \leq s^y$ then $\Phi(s)>0$ is the hypothesis, and if $s \in (s^y , \bar{s}(\Phi))$ then $s$ lies strictly below the infimum, so $\Phi(s) \leq 0$ fails. The second statement is the same argument reflected. Note that no monotonicity of $\Phi$ is claimed, and none is needed; and that the second entry in each expression is what prevents the threshold from landing on an endpoint, which the first entry alone permits when $\Phi > 0$ on all of $(0,1)$.
\end{proof}

\subsection{LP class proof of Remarks \ref{rem:indiv_comp_stat_x_s} and \ref{rem:indiv_comp_stat_x_rho}}
\textbf{LP class.} $x_{LP}^*$ is the unique solution to:
    \begin{align}
    \Phi_{LP} := u_x(x_{LP} , y_{LP}) + \theta[(1-\rho)\, F'(x^*_{HP} - x_{LP}) + \rho\, F'(x^*_{LR} - x_{LP})] = 0
    \end{align}
    Note that $y_{LP}$ is invariant in $s$, but both $x^*_{HP}$ and $x^*_{LR}$ depend on $s$. So we have:
    \begin{align*}
        \frac{\partial \Phi_{LP}}{\partial x_{LP}} &= u_{xx,LP} - \theta(1-\rho)\, F''_{LP,HP} - \theta\rho\, F''_{LP,LR} < 0 \\
        \frac{\partial \Phi_{LP}}{\partial s} &= \theta(1-\rho)\, F''_{LP,HP}\cdot\frac{dx^*_{HP}}{ds} + \theta\rho\, F''_{LP,LR}\cdot\frac{dx^*_{LR}}{ds}.
    \end{align*}
    And notice that $\Phi_{LP}$ depends on $\rho$ directly, and indirectly through $x^*_{HP}$ and $x^*_{LR}$. So we have:
    \begin{align*}
       \frac{\partial \Phi_{LP}}{\partial \rho} &= \underbrace{- \theta F'(x^*_{HP} - x^*_{LP}) + \theta F'(x^*_{LR} - x^*_{LP})}_{\text{direct effect}} \\
        &\qquad + \underbrace{\theta(1-\rho)\, F''_{LP,HP}\cdot\frac{dx^*_{HP}}{d\rho} + \theta\rho\, F''_{LP,LR}\cdot\frac{dx^*_{LR}}{d\rho}}_{\text{indirect effect}}
    \end{align*}
    Then by the Implicit Function Theorem, we have:
\begin{align*}
    \frac{dx^*_{LP}}{ds} = - \frac{\partial \Phi_{LP} / \partial s}{\partial \Phi_{LP} / \partial x_{LP}}
    \qquad \text{and} \qquad
    \frac{dx^*_{LP}}{d\rho} = - \frac{\partial \Phi_{LP} / \partial \rho}{\partial \Phi_{LP} / \partial x_{LP}}
\end{align*}
As $\frac{\partial \Phi_{LP}}{\partial x_{LP}} < 0$, we have $\text{sign}\big[\frac{dx^*_{LP}}{ds}\big] = \text{sign}\big[\frac{\partial \Phi_{LP}}{\partial s}\big]$ and $\text{sign}\big[\frac{dx^*_{LP}}{d\rho}\big] = \text{sign}\big[\frac{\partial \Phi_{LP}}{\partial \rho}\big]$. Dividing each by $\theta>0$ then gives parts (iv) of the two remarks.
\newpage
\section{Microfoundations for the network structure}\label{OA:network_microfoundations}
Here, we provide simple microfoundations that yield the network structure in \Cref{fig:network_structure}. The account below is in terms of \emph{expected} link shares, and delivers the structure of \Cref{sec:model} exactly in the large-$n$ limit. \Cref{sec:model} instead assumes exact class-specific weighted degrees, and an exactly common total degree $d$, for every agent. With finitely many agents a random meeting-and-matching process of the kind described here yields random degrees and random class shares, which agree with those assumptions only in expectation. Readers who prefer an exact construction at finite $n$ may instead take the row sums $\sum_{j \in q'} G_{ij} = d \, g_{qq'}$ as given and construct a network satisfying them directly, in the manner of a configuration model; the point of this appendix is the behavioral account of \emph{why} the shares take the form they do, not the construction itself.

Suppose agents meet uniformly at random. So for any agent $i$ in class $q$, the probability she \emph{meets} an agent in class $q'$ is equal to the fraction of the other agents who are in class $q'$. For $q' \neq q$ this is $\frac{0.25 n}{n-1}$, and for her own class it is $\frac{0.25 n - 1}{n-1}$, as she cannot meet herself. Both tend to one quarter as $n$ grows, and we work with that limit throughout.\footnote{Nothing of substance turns on this. Carrying the exact finite-$n$ probabilities through simply multiplies the own-class connecting probability by $\frac{0.25n-1}{0.25 n}$, which rescales $\omega$ without affecting $\rho$ -- and hence leaves the integration parameter, which is our object of interest, untouched.} Conditional on meeting, a pair of agents \emph{connect} -- form a link -- with the following probability:
\begin{align}
p_{qq'} =
\begin{cases}
    1 &\text{ if } q=q' \\
    a^\dagger &\text{ if } q \neq q', a_q = a_{q'} \\
    b^\dagger &\text{ if } q \neq q', b_q = b_{q'} \\
    0 &\text{ otherwise.}
\end{cases}
\end{align}
This says that, conditional on meeting, agents who have nothing (neither ability nor background) in common will never connect, while those with both ability and background in common always connect. Then agents with only one thing in common have some probability of connecting -- which can depend on which dimension (ability or background) it is that they have in common. 

As all classes are the same size (by assumption), the expected number of links from an agent in class $q$ to those in class $q'$ is simply proportional to the `connecting' probability $p_{qq'}$. Then the total expected number of links is proportional to $\sum_{q'} p_{qq'} = 1 + a^\dagger + b^\dagger$. So the expected \emph{fraction} of links from an agent in class $q$ to those in class $q'$ is:
\begin{align}\label{eq:microfounded_link_weights}
T_{qq'} =
\begin{cases}
    \frac{1}{1 + a^\dagger + b^\dagger} &\text{ if } q=q' \\
    \frac{a^\dagger}{1 + a^\dagger + b^\dagger} &\text{ if } q \neq q', a_q = a_{q'} \\
    \frac{b^\dagger}{1 + a^\dagger + b^\dagger} &\text{ if } q \neq q', b_q = b_{q'} \\
    0 &\text{ otherwise.}
\end{cases}
\end{align}
Finally, letting $\omega := (1 + a^\dagger + b^\dagger)^{-1}$, and $\rho := \frac{a^\dagger}{a^\dagger + b^\dagger}$ then yields the network structure in \Cref{sec:model}.\footnote{Note that this does not microfound the full range $\omega \in (0,1)$ that we consider in \Cref{sec:model}. It imposes a very natural restriction that people are disproportionately connected to their own group (as $\omega > \tfrac{1}{3}$ is a necessary condition for the link weights in \Cref{eq:microfounded_link_weights} to be feasible. Microfoundations for wider range of parameters could be done straightforwardly by also allowing homophily in meetings.}

Note that this very simple microfoundation has entirely focused on homophily driven by \emph{choices} (i.e. preferences for forming connections conditional on meeting). It abstracted away from differences in meeting rates, sometimes called opportunity-homophily or propinquity \citep{currarini2010identifying, zeng2008preference}. This is merely for clarity. It would be straightforward to add in differences in meeting rates. Doing so would add notation, but would not add insight.\label{OA.E.network_microfoundations}
\newpage
\section{Proofs for the enriched model}\label{OA:enriched_proofs}
\subsection{Preliminary Notation}\label{subsec:notation}
A strategy for agent $i$ is a measurable set $X_i \subseteq [0,A]$, identified with any other set that differs from it by a set of measure zero. We write
\begin{align*}
    x_i \ := \ |X_i| ,
    \qquad
    M(X_i) \ := \ \int_{X_i} m^{\alpha} \, dm ,
    \qquad
    \hat{X}_i \ := \ \big[ (1+\alpha) M(X_i) \big]^{\frac{1}{1+\alpha}} ,
\end{align*}
for the \emph{mass} of goods $i$ consumes, the \emph{expenditure} they require, and the expenditure aggregate of \Cref{eq:preferences_microfounded}. Because $X_i$ is contained in the bounded interval $[0,A]$, both $x_i \leq A$ and $M(X_i) \leq \frac{A^{1+\alpha}}{1+\alpha}$ are finite. In this notation agent $i$'s payoff is
\begin{align}\label{eq:enriched_payoff_sets}
    U_i(X) \ = \ v \, x_i \ - \ \kappa\big( \hat{X}_i , y_i \big) \ - \ \sum_{j \neq i} G_{ij} \, F\big( |X_j \setminus X_i| \big) ,
\end{align}
since $\int_0^A \max\{ x_{jm} - x_{im} , 0 \} \, dm = |X_j \setminus X_i|$. 

\subsection{Cost and benefits change is without loss}
Before starting the proofs proper, we first confirm that the decomposition of the intrinsic benefit into a linear term and a strictly increasing cost term is without loss. 
\begin{lem}\label{lem:cost_decomposition}
    For any $u(\cdot,\cdot)$ satisfying \Cref{ass:u_fn} there exist a constant $v>0$ and a function $\kappa(\cdot,\cdot)$, twice continuously differentiable with $\kappa_x > 0$ everywhere and $\kappa_{xy} < 0$ for $x>0$, such that
    \begin{align*}
        u(x,y) \ = \ v \, x - \kappa(x,y) \qquad \text{for all } x \geq 0 \text{ and all } y \in [y_{LP} , y_{HR}] .
    \end{align*}
\end{lem}

\begin{proof}
    Since $u_{xx} < 0$, the function $u_x(\cdot , y)$ is strictly decreasing, so $\sup_{x\geq0} u_x(x,y) = u_x(0,y)$; and $u_x(0,\cdot)$ is continuous on the compact interval $[y_{LP} , y_{HR}]$, so $\bar{v} := \max_{y \in [y_{LP} , y_{HR}]} u_x(0,y)$ is finite, and strictly positive. Take any $v > \bar{v}$ and set $\kappa(x,y) := v x - u(x,y)$. Then $\kappa_x = v - u_x \geq v - \bar{v} > 0$, and $\kappa_{xy} = -u_{xy} < 0$ for $x > 0$ by \Cref{ass:u_fn}(iii). Twice continuous differentiability is inherited from $u$.
\end{proof}

\newpage
\subsection{Proof of \Cref{prop:class_symmetric}}
\paragraph{Payoffs with class-symmetry.} Fix a class-symmetric Nash equilibrium and write $X_q \subseteq [0,A]$ for the set of goods consumed by agents in class $q \in Q$. As usual an agent $i \in q$ takes others' choices as given, so she maximizes
\begin{align}\label{eq:class_objective}
    U_q(Z) \ = \ v \, |Z| \ - \ \kappa\big( \hat{Z} , y_q \big) \ - \ \sum_{p \in Q} d \, g_{q p} \, F\big( |X_{p} \setminus Z| \big)
\end{align}
over measurable $Z \subseteq [0,A]$. In equilibrium $X_q$ attains this maximum.

\paragraph{Additional notation.} For a good $m$, we say its \emph{clientele}, denoted $C(m)$, is the set of classes that buy it. And for a (sub)set of the classes $S \subseteq Q$, we say a \emph{bracket}, denoted $E_S$, is the set of goods whose clientele is exactly $S$. So:
\begin{align*}
    C(m) \ := \ \{ p \in Q \ : \ m \in X_p \} \qquad\text{and}\qquad E_S := \{ m \in [0,A] : C(m) = S \}
\end{align*}
The $2^4$ sets $E_S$ are the level sets of $C$, so they are disjoint and cover $[0,A]$. Note that on $E_S$, the indicator $\mathbf{1}\{ m \in X_p \}$ equals $\mathbf{1}\{ p \in S \}$ for every $p$. 

Call a set of classes $S \subseteq Q$ a \emph{club} if $|E_S| > 0$. That is, if $S$ is the exclusive clientele of some positive mass of goods. Write $\mathcal{K}$ for the collection of clubs. Additionally, call a club $S$ a \emph{separating club} if there are some $p,q \in Q$, with $p \neq q$, such that $p \in S$ and $q \notin S$ (note that all clubs are separating clubs except for $S = Q$ and $S = \emptyset$). Also let $\bar{m} := \max_{q \in Q} \operatorname{ess\,sup} X_q$ denote the most expensive good consumed. Goods $m > \bar{m}$ are bought by nobody, so we discard them and work on $[0 , \bar{m}]$ throughout. 

\paragraph{A First Order Condition.} As in the `simpler' model, agents consume goods whenever the (net) marginal benefit of doing so is positive. In the simpler model, this amounted to a first order condition where \Cref{eq:marginal_payoff} holds with equality. There is a direct analogy here, but it requires a little more notation. For convenience, we write:
\begin{align*}
    \tau_q \ := \ \frac{\kappa_x\big( \hat{X}_q , y_q \big)}{\hat{X}_q^{\,\alpha}},
    \qquad
    \lambda_{qp} \ := \ d \, g_{qp} \, F'\big( |X_p \setminus X_q| \big).
\end{align*}
Using these, the analogue of \Cref{eq:marginal_payoff} is:
\begin{align}\label{eq:psi_def}
    \psi_q(m) \ := \ v \ + \sum_{p \neq q} \lambda_{qp} \, \mathbf{1}\{ m \in X_p \} \ - \ \tau_q \, m^{\alpha}.
\end{align}
So then our First Order Condition is that $X_q$ is the set of goods for which $\psi_q(m) > 0$. Formally:
\begin{lem}\label{lem:marginal}
    $X_q \ = \ \big\{ m \in [0,A] \ : \ \psi_q(m) > 0 \big\}$ up to a set of measure zero. Moreover, $\hat{X}_q > 0$, so $\tau_q$ is well defined and strictly positive.
\end{lem}

\begin{proof}
    \textbf{Step 1: consumption is positive.} Suppose $|X_q| = 0$. Replacing $X_q$ by $[0,\eta]$ weakly shrinks every set $X_{p} \setminus \, \cdot \,$, so 
    \begin{align*}
        U_q\big( [0,\eta] \big) - U_q(X_q) \ \geq \ v \, \eta - \kappa(\eta , y_q) + \kappa(0 , y_q) \ = \ u(\eta , y_q) - u(0 , y_q) \ > \ 0
    \end{align*}
    for sufficiently small $\eta > 0$, because $u_x(0 , y_q) > 0$ by \Cref{ass:u_fn}(i). Contradiction. So $\hat{X}_q > 0$, and hence $\tau_q > 0$. 
  
    \textbf{Step 2: no gain from buying more.} Fix $S \subseteq Q$ with $q \notin S$, and suppose class $q$ deviates by buying a measurable $W$ inside the bracket $E_S$ (i.e., $W \subseteq E_S$). Since $W \cap X_q = \emptyset$ and $W \subseteq X_p$ exactly for $p \in S$, her shortfall against $p$ falls from $|X_p \setminus X_q|$ to $|X_p \setminus X_q| - |W|$ when $p \in S$, and is unchanged when $p \notin S$. Her change in utility is therefore:
    \begin{align}\label{eq:toggle}
    U_q( X_q \cup W ) - U_q( X_q )
    \ = \ & \ v \, |W| \ + \sum_{p \in S} d \, g_{qp} \Big[ F\big( |X_p \setminus X_q| \big) - F\big( |X_p \setminus X_q| - |W| \big) \Big] \nonumber \\
    & - \ \Big[ \kappa\big( \widehat{X_q \cup W} , y_q \big) - \kappa\big( \hat{X}_q , y_q \big) \Big] .
    \end{align}
    In equilibrium this must be weakly negative for every such $W$. Now shrink $W$ to a point. Fix a Lebesgue density point $m$ of $E_S$ -- almost every point of $E_S$ is one -- and put $W_\delta := E_S \cap (m - \delta , \, m + \delta)$, which has positive measure, with $|W_\delta| \downarrow 0$ as $\delta \downarrow 0$. Divide \Cref{eq:toggle} by $|W_\delta|$ and let $\delta \downarrow 0$: each bracket becomes a derivative, and $M(W_\delta) / |W_\delta| \to m^{\alpha}$, since $m \mapsto m^{\alpha}$ is continuous and so its average over $W_\delta \subseteq (m-\delta , m+\delta)$ is squeezed to $m^{\alpha}$. So we have
    \begin{align*}
        v \ + \sum_{p \in S} d \, g_{qp} \, F'\big( |X_p \setminus X_q| \big) \ - \ \kappa_x\big( \hat{X}_q , y_q \big) \, \hat{X}_q^{-\alpha} \, m^{\alpha} \ \leq \ 0 .
    \end{align*}
    The last term is the marginal cost of the extra expenditure. Buying $W_\delta$ raises expenditure by $M(W_\delta)$, and the cost of an expenditure $M$ is $\kappa( \hat{Z} , y_q )$ with $\hat{Z} = [(1+\alpha) M]^{1/(1+\alpha)}$. Since $d \hat{Z} / d M = \hat{Z}^{-\alpha}$, the chain rule gives
    \begin{align*}
        \frac{d}{d M} \, \kappa\big( \hat{Z} , y_q \big) \bigg|_{M \, = \, M(X_q)} \ = \ \kappa_x\big( \hat{X}_q , y_q \big) \, \hat{X}_q^{-\alpha} \ = \ \tau_q ,
    \end{align*}
    which is well defined because Step 1 gives $\hat{X}_q > 0$. Multiplying by $M(W_\delta) / |W_\delta| \to m^{\alpha}$ produces $\tau_q \, m^{\alpha}$. Recalling that $\lambda_{qp} = d \, g_{qp} F'( |X_p \setminus X_q| )$, and that $\mathbf{1}\{ m \in X_p \} = \mathbf{1}\{ p \in S \}$ for $m \in E_S$, the display above is exactly $\psi_q(m) \leq 0$. This holds for almost every $m \in E_S$; and the sets $E_S$ with $q \notin S$ are finite in number and cover $[0,A] \setminus X_q$, so $\psi_q(m) \leq 0$ for almost every $m \notin X_q$.

    \textbf{Step 3: no gain from buying less.} Suppose instead that class $q$ deviates by giving up a measurable $W \subseteq E_S$ for some $S \ni q$, so that $W \subseteq X_q$. The reasoning of Step 2 runs again, with three quantities moving the other way: the mass falls by $|W|$, expenditure falls by $M(W)$, and the shortfall against each $p \in S$ rises from $|X_p \setminus X_q|$ to $|X_p \setminus X_q| + |W|$. So \Cref{eq:toggle} holds with $v|W|$, $M(W)$ and $|W|$ each carrying the opposite sign, and the same limit gives $-\psi_q(m) \leq 0$, that is $\psi_q(m) \geq 0$, for almost every $m \in X_q$. The own-class term $p = q$ is active here -- an agent who drops $W$ falls behind her own class by $|W|$ -- but it contributes $d \, g_{qq} F'(0) = 0$ in the limit, which is why $\psi_q$ sums only over $p \neq q$.
 
    \textbf{Step 4: ties are null.} Combining steps 2 and 3, we have $\psi_q(m) > 0$ for $m \in X_q$ and $\psi_q(m) < 0$ for $m \in [0,A] \setminus X_q$, up to a set of measure zero. All that remains is to rule out ties. Since $g_{qp} = 0$ unless $p$ is one of the two classes $q$ is linked to, the sum in $\psi_q$ takes at most four values; on the set where it equals a given $\ell$ the equation $\psi_q(m) = 0$ reads $\tau_q m^{\alpha} = v + \ell$, which has a single root because $\tau_q > 0$. So $\{ \psi_q = 0 \}$ is finite, and the two inclusions collapse to $X_q = \{ \psi_q > 0 \}$ up to a set of measure zero.
\end{proof}

\paragraph{Social comparison pressures across classes.} We now consider when a class can pass over some cheaper goods and buy dearer ones in their place. Such a `swap' can only be worthwhile if it reduces the social comparison costs with some `neighboring' class. So no class will do this unilaterally: there must be some other neighboring class that does so too (both buying the dearer goods and dropping the cheaper ones). Before proving this formally, we need to introduce more notation.

\emph{Adjacent classes.} Let $\mathcal{N}(q)$ denote the \emph{neigbhours} of class $q$ (i.e., the other classes that $q$ is linked to). Given the network structure we have assumed (shown in \Cref{fig:network_structure}) the classes form a cycle: $HR \ - \ HP \ - \ LP \ - \ LR \ - \ HR $, with adjacent classes sharing an ability or a background, and the two diagonal pairs sharing neither. Write $q_1 , q_2 , q_3 , q_4$ for $HR , HP , LP , LR$ with indices modulo $4$, so $\mathcal{N}(q_k) = \{ q_{k-1} , q_{k+1} \}$. Therefore we have $g_{qp} > 0$ exactly for $p \in \mathcal{N}(q) \cup \{q\}$.   

\emph{Placing pressure.} We say that class $p$ \emph{exerts pressure on} $q$, written $p \succ q$, if: $|X_p \setminus X_q| > 0$ and $p \in \mathcal{N}(q)$.\footnote{It is important to note that $\succ$ here is \emph{not} an ordering. It is neither transitive nor antisymmetric: $p \succ q$ and $q \succ p$ can hold together, since two classes can each buy goods the other does not.} 
Equivalently, $p \succ q$ if $p$ is a neighboring class (to $q$) and there exists some club that contains $p$ but not $q$. Any such club is a separating club, since it contains $p$ but not $q$. So only separating clubs can give rise to pressures. Let $\mathcal{N}^+(q) := \{ p : p \succ q \}$ denote the set of classes that exert pressure on $q$.

As $F'$ is strictly increasing with $F'(0) = 0$, we have $\lambda_{qp} > 0$ if and only if $p \in \mathcal{N}^+(q)$. Write
\begin{align}\label{eq:pressure}
    \ell_q(S) \ := \ \sum_{p \, \in \, S} \lambda_{qp} \ = \sum_{p \, \in \, \mathcal{N}^+(q) \, \cap \, S} \lambda_{qp}
\end{align}
for the pressure that the classes in $S$ put on $q$: a sum of strictly positive terms, depending on $S$ only through $\mathcal{N}^+(q) \cap S$. The sum in \Cref{eq:psi_def} is constant on each bracket, $E_S$, so
\begin{align*}
    \psi_q(m) \ = \ v \ + \ \ell_q(S) \ - \ \tau_q \, m^{\alpha}
    \qquad \text{for } m \in E_S ,
\end{align*}
which is strictly decreasing in $m$.

\emph{Comparing clubs by their cheapest good.} For $S \in \mathcal{K}$ let $\underline{m}_S := \operatorname{ess\,inf} E_S$ and $\overline{m}_S := \operatorname{ess\,sup} E_S$. We say that $S$ \emph{dips below} $S'$, written $S \lhd S'$, if $\underline{m}_S < \overline{m}_{S'}$: that is, if some of the goods in $E_S$ are cheaper than some of the goods in $E_{S'}$. Any two distinct clubs must be comparable in this way: one dips below the other (or both). Otherwise $\overline{m}_S \leq \underline{m}_{S'} \leq \overline{m}_{S'} \leq \underline{m}_S \leq \overline{m}_S$, forcing $\underline{m}_S = \overline{m}_S$ and hence $|E_S| = 0$ (in which case $S$ is not a club, yielding a contradiction).

\begin{lem}\label{lem:jump}
    Let $S , S' \in \mathcal{K}$ with $S \lhd S'$.
    \emph{(i)} If $q \in S' \setminus S$, then $\mathcal{N}^+(q) \cap (S' \setminus S) \neq \emptyset$.
    \emph{(ii)} Hence $|S' \setminus S| \neq 1$; and no two clubs have $|S \setminus S'| = |S' \setminus S| = 1$.
\end{lem}

\begin{proof}
    \emph{(i)} As $S \lhd S'$, we can pick $m \in E_S$ and $m' \in E_{S'}$ with $m < m'$, both outside the null set on which \Cref{lem:marginal} fails. Since $q \notin S$ we have $\psi_q(m) \leq 0$, and since $q \in S'$ we have $\psi_q(m') > 0$, so with $\tau_q > 0$ and $m < m'$,
    \begin{align*}
        v + \ell_q(S) \ \leq \ \tau_q \, m^{\alpha} \ < \ \tau_q \, (m')^{\alpha} \ < \ v + \ell_q(S') ,
    \end{align*}
    and $\ell_q(S) < \ell_q(S')$. By \Cref{eq:pressure} the two sides sum strictly positive terms over $\mathcal{N}^+(q) \cap S$ and over $\mathcal{N}^+(q) \cap S'$, so the second index set is not contained in the first.

    \emph{(ii)} For the first claim, suppose $|S' \setminus S| = 1$, say $S' \setminus S = \{p\}$. Then $p \in S' \setminus S$, so part \emph{(i)} applies -- yielding $\mathcal{N}^+(p) \cap (S' \setminus S) = \mathcal{N}^+(p) \cap \{p\} \neq \emptyset$. But by definition $p$ cannot be its own neighbor, nor can it exert pressure on itself. So $\mathcal{N}^+(p) \cap \{p\} = \emptyset$. Contradiction. 
    For the second claim, let $S \neq S'$ be clubs with $S \setminus S' = \{q\}$ and $S' \setminus S = \{p\}$. If $S \lhd S'$, the first claim gives $|S' \setminus S| \neq 1$, a contradiction. If instead $S' \lhd S$, the first claim with the roles of $S$ and $S'$ exchanged gives $|S \setminus S'| \neq 1$, again a contradiction. So neither club dips below the other -- and that is impossible, since any two distinct clubs are comparable.

\end{proof}

\paragraph{The top club.} Next we consider the club that consumes the most expensive good (of all the goods that are consumed) -- which we call the \emph{top club}. Let $T$ be a non-empty club with $\overline{m}_T = \bar{m}$. Every other club dips below it, since a bracket of positive measure inside $[0 , \bar{m}]$ has $\underline{m}_S < \bar{m}$.

\begin{lem}[Peeling]\label{lem:peel}
    Either \emph{(a)} some class $q$ has $X_q = [0 , \bar{m}]$ up to a null set; or \emph{(b)} every $q \in T$ has both of its neighbors in $T$, and both exert pressure on $q$.
\end{lem}

\begin{proof}
    Suppose $\mathcal{N}^+(q) \cap T = \emptyset$ for some $q \in T$, and let $S$ be a club with $q \notin S$. Then $S \neq T$, so $S \lhd T$, and \Cref{lem:jump}\emph{(i)} applied to $q \in T \setminus S$ produces a member of $\mathcal{N}^+(q) \cap T$. Contradiction. So $q$ belongs to every club. Then, because the brackets $E_S$ partition $[0 , \bar{m}]$, we get $X_q = [0 , \bar{m}]$.

    Otherwise $\mathcal{N}^+(q) \cap T \neq \emptyset$ for every $q \in T$, and it is enough to rule out $\mathcal{N}^+(q) \cap T = \{p\}$. Suppose it holds. For any club $S$ with $q \notin S$ we have $S \lhd T$, and \Cref{lem:jump}\emph{(i)} places a member of $\mathcal{N}^+(q)$ in $T \setminus S$; it can only be $p$, so $p \notin S$. Every club containing $p$ therefore contains $q$, so $X_p \subseteq X_q$ up to a null set and $p \not\succ q$ -- contradicting $p \in \mathcal{N}^+(q)$.
\end{proof}

\begin{lem}[No full top]\label{lem:fulltop}
    Alternative \emph{(b)} of \Cref{lem:peel} cannot occur. So some class $q$ has $X_q = [0 , \bar{m}]$.
\end{lem}

\begin{proof}
    \textbf{Step 1: under \emph{(b)}, the top club contains every class.}
    Suppose, seeking a contradiction, that alternative \emph{(b)} holds. Every class in $T$ then has both of its neighbors inside $T$, so the network induced on $T$ has minimum degree two. As the network is a cycle, the subgraph induced on any proper subset of it is a forest -- which has a class with degree at most one. So it must be that $T = Q$; and then every class must have pressure exerted on it by both of its neighbors (i.e., $\mathcal{N}^+(q) = \mathcal{N}(q)$ for all $q$). Since $T = Q$ is the top club, every other club dips below it, so we may apply \Cref{lem:jump} to the pair $S \lhd Q$ for every club $S \neq Q$.

    \textbf{Step 2: which separating clubs are possible.} We now narrow down which sets of classes can be separating clubs. Let $S$ be a separating club and write $Z := Q \setminus S$ for the classes it omits. Since $S$ is neither $Q$ nor $\emptyset$, the set $Z$ is neither empty nor all of $Q$. Since $Q$ is the top club we have $S \lhd Q$. So \Cref{lem:jump}\emph{(i)} may be applied to this pair, which says that for any $q \in Z$ there exists a class in $\mathcal{N}^+(q) \cap Z$.\footnote{Here, we are setting the $S'$ of \Cref{lem:jump} to be $Q$, and hence $Z = S' \setminus S$.} 
    And recall that under alternative \emph{(b)}, every class is pressured by both of its neighbors, so $\mathcal{N}^+(q) = \mathcal{N}(q)$. In words, this says that every class $q$ \emph{not} in the club $S$ (i.e., $q\in Z$) must have a neighbor that is \emph{also} not in the club $S$. This implies that the network induced on $Z$ has no isolated class (i.e., one with no neighbors). 

    This imposes restrictions on the possible makeup of $Z$, and hence on the possible makeup of $S$. \textbf{(1)} First, $Z$ cannot be a singleton, because the single class would have no neighbors in $Z$ -- which we have shown is not possible. As a result, $S$ cannot be a three-element set (recall that $|S| + |Z| = |Q| = 4$ by definition). \textbf{(2)} Second, $Z$ cannot be a ``diagonal pair'' -- two classes that are not adjacent in the network $HR \ - \ HP \ - \ LP \ - \ LR \ - \ HR $. This is because each class in a diagonal pair would have no neighbors in $Z$, which is again not possible. As a result, $S$ cannot be a diagonal pair either. \textbf{(3)} Therefore $Z$ \emph{may} be an adjacent pair or a set of three classes; in both, each class in $Z$ has at least one neighbor in $Z$. Translating back to the club itself: $S$ is not a three-element set by \textbf{(1)}, not a diagonal pair by \textbf{(2)}, and neither $Q$ nor $\emptyset$ because it is a separating club. What remains is that $S$ is a singleton or an adjacent pair.


    \textbf{Step 3: how many separating clubs can coexist.} We now bound how many of these can be clubs simultaneously. By the previous paragraph, every separating club is either a singleton or an adjacent pair. 

    \emph{Claim (1).} At most one singleton is a club. Suppose not: so $\{ q_j \}$ and $\{ q_k \}$ are both clubs, with $j \neq k$. From above, all pairs of clubs are comparable by the `dips below' definition, so we have $\{q_j\} \lhd \{q_k\}$ (with a relabeling of $j,k$ if necessary). Then notice that $\{q_j\} \setminus \{q_k\} = \{q_j\}$ and $\{q_k\} \setminus \{q_j\} = \{q_k\}$. This implies that $|\{q_j\} \setminus \{q_k\}| = |\{q_k\} \setminus \{q_j\}| = 1$. Then letting $S = \{q_j\}$, $S' = \{q_k\}$ (with $S \lhd S'$) \Cref{lem:jump}\emph{(ii)} shows this is not possible. Contradiction.

    \emph{Claim (2).} Two adjacent pairs that share a class cannot both be clubs. Suppose not: so $\{ q_k , q_{k+1} \}$ and $\{ q_{k+1} , q_{k+2} \}$ are both clubs. By the same argument as in Claim (1), they are comparable by the `dips below' definition, so we have $\{ q_k , q_{k+1} \} \lhd \{ q_{k+1} , q_{k+2} \}$ (with a relabeling of $k$ if necessary). Then notice that $\{ q_k , q_{k+1} \} \setminus \{ q_{k+1} , q_{k+2} \} = \{ q_k \}$ and $\{ q_{k+1} , q_{k+2} \} \setminus \{ q_k , q_{k+1} \} = \{ q_{k+2} \}$. This implies that $|\{ q_k , q_{k+1} \} \setminus \{ q_{k+1} , q_{k+2} \}| = |\{ q_{k+1} , q_{k+2} \} \setminus \{ q_k , q_{k+1} \}| = 1$. Then letting $S = \{q_k , q_{k+1}\}$, $S' = \{q_{k+1},q_{k+2}\}$ (with $S \lhd S'$) \Cref{lem:jump}\emph{(ii)} shows this is not possible. Contradiction. 
    
    \emph{Claim (3).} At most two adjacent pairs are clubs. By Claim 2, overlapping adjacent pairs cannot both be clubs. So clubs made up of adjacent pairs must use disjoint sets of classes. Each club that is an adjacent pair therefore `uses up' two of the four classes, leaving room for at most two such clubs. Putting \emph{(1)} and \emph{(3)} together, there are at most three separating clubs: at most one singleton, and at most two adjacent pairs.

    \textbf{Step 4: alternative \emph{(b)} imposes four requirements.} We now show that alternative \emph{(b)} demands more separating clubs than exist. Under \emph{(b)} every class is pressured by both of its neighbors. Applying this to the class $q_{k+1}$, whose neighbors are $q_k$ and $q_{k+2}$, gives $q_k \succ q_{k+1}$ for each $k = 1,2,3,4$. And $q_k \succ q_{k+1}$ means, by definition, that some club contains $q_k$ but not $q_{k+1}$. Write $R_k$ for this requirement. All four of $R_1 , R_2 , R_3 , R_4$ must hold. Moreover a club that meets $R_k$ contains $q_k$ and excludes $q_{k+1}$, so it is a separating club: neither $Q$ nor $\emptyset$ can meet any requirement.

    \textbf{Step 5: the four requirements cannot all be met.} For all of $R_1 , R_2 , R_3 , R_4$ to hold, we would need four distinct separating clubs. To see this, note that if a single separating club $S$ meets two requirements $R_j$ and $R_k$, then it contains both $q_j$ and $q_k$, and excludes both $q_{j+1}$ and $q_{k+1}$. Suppose $j \neq k$. If $k = j+1$, then $R_k$ places $q_{j+1}$ inside $S$ while $R_j$ keeps it out; and if $k = j-1$, then $R_k$ keeps $q_j$ outside $S$ while $R_j$ places it in. Both are impossible. That leaves only $k = j+2$, which forces $S = \{ q_j , q_{j+2} \}$ -- a diagonal pair, which Step 2 has already ruled out. Hence each requirement must be met by a different separating club. Then the previous paragraph shows that there are at most three separating clubs. So alternative \emph{(b)} cannot hold.
\end{proof}

\paragraph{Putting it all together: proof of \Cref{prop:class_symmetric}.} Fix a class-symmetric equilibrium. We recursively peel classes. At any stage, let $R\subseteq Q$ be the \emph{remaining} classes, let $\mathcal P:=Q\setminus R$ be the \emph{peeled} classes, and define
\[    \bar m_R:=\max_{p\in R}\operatorname{ess\,sup}X_p.   \]
Our induction hypothesis is that every $q\in\mathcal P$ has a cut-off bundle $X_q=[0,x_q]$ with $x_q\geq \bar m_R$. Thus every peeled class buys every good $m \in [0,\bar m_R]$. For each $S\subseteq R$, define the residual bracket:\footnote{These brackets retain all goods in $[0,\bar m_R]$. In particular, goods bought only by peeled classes belong to $E_\varnothing^R$ and remain available as possible deviations.}
\[    E_S^R:=\{m\in[0,\bar m_R]:C(m)\cap R=S\}. \]
For any remaining class $p\in R$, its marginal payoff on this interval can be written as
\[  \psi_p(m)  =  \underbrace{v+\sum_{q\in\mathcal P}\lambda_{pq}}_{=:v_p^{\mathcal P}}
 +\sum_{r\in R\setminus\{p\}}    \lambda_{pr}\mathbf 1\{m\in X_r\}  -\tau_p m^\alpha.  \]
Each peeled class contributes a class-specific constant because it buys every good in the interval. Thus peeling changes only the intercept of each remaining class's marginal payoff; it removes neither goods nor feasible deviations. This simply rewrites the marginal payoff, so \Cref{lem:marginal} still applies. Moreover, \Cref{lem:jump} applies to residual brackets, wit pressure restricted to classes in $R$ (as the $v_p^{\mathcal P}$ terms cancel when comparing two residual brackets). \Cref{lem:peel} then yields that: either some $p\in R$ buys every good in $[0,\bar m_R]$, or every class in the residual top club has at least two pressure-exerting neighbors inside that club.

Initially, $R=Q$ and $\mathcal P=\varnothing$. \Cref{lem:fulltop} rules out the second alternative and gives some $q$ with $X_q=[0,\bar m_Q]$. Peel this class. At every subsequent stage, the graph induced on $R$ is a proper subgraph of the original network, and hence a forest. So for every possible remaining club, there is always a class with at most one neighbor inside that club. Therefore, the second alternative is again impossible, and we have some $p\in R$ that buys every good in $[0,\bar m_R]$. Since no remaining class buys above $\bar m_R$, this gives $X_p=[0,\bar m_R]$ up to a null set. Peel $p$. Its cut-off endpoint is at least the next value of $\bar m_R$, so the induction hypothesis is preserved. Iterating shows that every class has a cut-off bundle.

    So a class-symmetric equilibrium is a cut-off profile. Finally, \Cref{lem:cutoff_br} shows that cut-off profiles are the unique best responses to cut-off profiles, and that the game `collapses' to the `main' model of \Cref{sec:model}. By \Cref{lem:A0_uniqueness} (also by the contraction mapping argument of the proof to \Cref{rem:cutoff_convergence}) the levels $x_q$ are uniquely determined, and the profile of cut-off sets is therefore unique and the same as in the `main' model. \hfill \qed

\subsection{Myopic Best Responses from Cut-offs Yield Equilibrium}
On the way to proving \Cref{rem:cutoff_convergence}, it is helpful to present a pair of building blocks. First, we show that cut-off profiles are `cheap', in the sense that they cost less than any other profile of the same mass, and are `effective', in the sense that they minimize the harmful social comparisons an agent whose friends are playing cut-off strategies faces. Second, we show that best responses to cut-off strategies are themselves cut-off strategies. 

\begin{lem}
    \label{lem:rearrangement}
    Fix $t \in [0,A]$ and let $X \subseteq [0,A]$ be measurable with $|X| = t$. Then
    \begin{enumerate}
        \item[\emph{(i)}] $M(X) \ \geq \ M\big([0,t]\big) = \tfrac{t^{1+\alpha}}{1+\alpha}$, with equality if and only if $X = [0,t]$ up to a set of measure zero. Equivalently, $\hat{X} \geq t$, with equality under the same condition.
        \item[\emph{(ii)}] $\big| X \cap [0,s] \big| \ \leq \ \min\{ t , s \} \ = \ \big| [0,t] \cap [0,s] \big|$ for every $s \in [0,A]$.
    \end{enumerate}
\end{lem}

\paragraph{Proof.} \emph{(i)} Since $|X| = \big|[0,t]\big| = t$, the two halves of the symmetric difference have the same measure: writing $\eta := \big| X \setminus [0,t] \big|$,
    \begin{align*}
        \big| [0,t] \setminus X \big| \ = \ t - \big| X \cap [0,t] \big| \ = \ |X| - \big| X \cap [0,t] \big| \ = \ \eta .
    \end{align*}
    Canceling the goods the two bundles share,
    \begin{align*}
        M(X) - M\big([0,t]\big)
        \ = \ \int_{X \setminus [0,t]} m^{\alpha} \, dm \ - \int_{[0,t] \setminus X} m^{\alpha} \, dm .
    \end{align*}
    Every good in the first set has $m > t$ and so costs $m^\alpha > t^\alpha$; every good in the second has $m \leq t$ and so costs $m^\alpha \leq t^\alpha$. The two sets have the same measure $\eta$, so the difference is at least $\eta \, t^{\alpha} - \eta \, t^{\alpha} = 0$, and strictly positive whenever $\eta > 0$. Finally $\eta = 0$ says exactly that $X$ and $[0,t]$ agree up to a set of measure zero. The statement for $\hat{X}$ follows because $\hat{X}$ is a strictly increasing function of $M(X)$ which equals $t$ at $M(X) = \tfrac{t^{1+\alpha}}{1+\alpha}$.

    \emph{(ii)} $\big| X \cap [0,s] \big| \leq \min\big\{ |X| , \big|[0,s]\big| \big\} = \min\{t,s\}$, and $[0,t] \cap [0,s] = \big[ 0 , \min\{t,s\} \big]$. \hfill \qed

\begin{lem}
    \label{lem:cutoff_br}
    Suppose every agent $j \neq i$ plays a cut-off strategy, $X_j = [0 , x_j]$. Then the function
    \begin{align}\label{eq:cutoff_objective}
        \Pi_i(t) \ := \ u(t , y_i) \ - \ \sum_{j \neq i} G_{ij} \, F\big( \max\{ x_j - t , \, 0 \} \big)
    \end{align}
    has a unique maximizer $BR_i(x_{-i})$ on $[0,A]$, and agent $i$ has a \emph{unique} best response, namely the \emph{cut-off} strategy $\big[ 0 , BR_i(x_{-i}) \big]$.
\end{lem}

\paragraph{Proof.}
    Fix $X_{-i}$ with $X_j = [0 , x_j]$ for every $j \neq i$, let $X \subseteq [0,A]$ be any strategy for $i$, and write $t := |X|$. The argument has two halves: Step 1 shows that replacing $X$ by the cut-off bundle of the same mass -- its \emph{cut-off rearrangement} $[0,t]$ -- raises $i$'s payoff, strictly unless $X$ was already of that form; Step 2 then optimizes over the mass.

    \textbf{Step 1: $U_i(X , X_{-i}) \leq \Pi_i(t)$, with equality if and only if $X = [0,t]$ up to a set of measure zero.}
    The benefit term of \Cref{eq:enriched_payoff_sets} is $v \, t$, which depends on $X$ only through $t$. For the cost term, \Cref{lem:rearrangement}\emph{(i)} gives $\hat{X} \geq t$, strictly unless $X$ is the cut-off set, so $\kappa_x > 0$ gives
    \begin{align*}
        - \, \kappa\big( \hat{X} , y_i \big) \ \leq \ - \, \kappa(t , y_i) ,
    \end{align*}
    again strictly unless $X$ is the cut-off set. For the comparison terms, \Cref{lem:rearrangement}\emph{(ii)} applied at $s = x_j$ gives
    \begin{align*}
        \big| X_j \setminus X \big| \ = \ x_j - \big| X \cap [0,x_j] \big| \ \geq \ x_j - \min\{ t , x_j \} \ = \ \max\{ x_j - t , \, 0 \} ,
    \end{align*}
    so each comparison cost is at least $G_{ij} F\big( \max\{x_j - t , 0\} \big)$, as $F$ is increasing and $G_{ij} \geq 0$. Summing the three terms and using $u(t,y_i) = v \, t - \kappa(t , y_i)$ delivers $U_i(X , X_{-i}) \leq \Pi_i(t)$. The cost term alone is strict whenever $X$ is not the cut-off set, which gives the equality claim.

    \textbf{Step 2: $\Pi_i$ has a unique maximizer on $[0,A]$.}
    The term $u(\cdot , y_i)$ is strictly concave by \Cref{ass:u_fn}\emph{(i)}. For each $j$, the map $t \mapsto \max\{x_j - t , 0\}$ is convex, and $F$ is increasing and convex, so the composition $t \mapsto F\big( \max\{x_j - t , 0\} \big)$ is convex; as $G_{ij} \geq 0$, the comparison term of $\Pi_i$ is concave. Hence $\Pi_i$ is strictly concave. It is also continuous on the compact interval $[0,A]$, so it attains a maximum, at a unique point $BR_i(x_{-i})$.

    \textbf{Step 3: conclusion.} Write $BR_i := BR_i(x_{-i})$. For an arbitrary strategy $X$, Steps 1 and 2 give
    \begin{align*}
        U_i(X , X_{-i}) \ \leq \ \Pi_i\big( |X| \big) \ \leq \ \Pi_i(BR_i) \ = \ U_i\big( [0 , BR_i] , X_{-i} \big) ,
    \end{align*}
    the final equality because Step 1 holds with equality at a cut-off set. So $[0,BR_i]$ is a best response. It is the only one: equality throughout requires $X = \big[0 , |X|\big]$ up to a set of measure zero, by Step 1, and then $|X| = BR_i$, by the uniqueness in Step 2. \hfill \qed

\paragraph{Some more notation.} Recall from \Cref{sec:proofs} the marginal payoff of the `main' model,
\begin{align*}
    \Psi_i(t , x_{-i}) \ = \ u_x(t , y_i) \ + \sum_{j \neq i} G_{ij} \, F'\big( \max\{ x_j - t , \, 0 \} \big) ,
\end{align*}
which is continuous. Write
\begin{align*}
    \underline{\sigma}_u \ := \min_{\substack{t \in [0,A] \\ y \in [y_{LP} , y_{HR}]}} \big( - u_{xx}(t , y) \big)
    \qquad \text{and} \qquad
    \bar{\sigma}_F \ := \max_{z \in [0,A]} F''(z)
\end{align*}
for the smallest curvature of the intrinsic benefit and the largest curvature of comparison costs, over the ranges that can arise. Both extrema are attained, the sets being compact and the functions continuous, and $\underline{\sigma}_u > 0$ because $-u_{xx} > 0$ everywhere by \Cref{ass:u_fn}\emph{(i)}. 

\paragraph{Proof of \Cref{rem:cutoff_convergence}.}
    \textbf{Step 1: the dynamic is the iteration of a map on $[0,A]^n$.}
    By hypothesis, $X^0$ is a profile of cut-off strategies. If $X^{t-1}$ is a profile of cut-off strategies, then \Cref{lem:cutoff_br} gives each agent a unique best response, itself a cut-off strategy, with level $BR_i(x^{t-1}_{-i})$. So by induction every $X^t$ is a profile of cut-off strategies, and $x^t = BR(x^{t-1})$, where $BR := (BR_i)_{i \in N}$ maps $[0,A]^n$ into itself.

    \textbf{Step 2: $BR_i(x_{-i})$ is interior, and is the unique root of $\Psi_i(\cdot , x_{-i})$.}
    By \Cref{lem:cutoff_br}, $\Pi_i$ is the `main' model payoff, so $\Pi_i' = \Psi_i$ and $\Pi_i$ is strictly concave; hence $\Psi_i(\cdot , x_{-i})$ is strictly decreasing and has at most one root. At the two endpoints,
    \begin{align*}
        \Psi_i(0 , x_{-i}) \ = \ u_x(0,y_i) + \sum_{j \neq i} G_{ij} F'(x_j) \ > \ 0 ,
        \qquad
        \Psi_i(A , x_{-i}) \ = \ u_x(A , y_i) \ < \ 0 ,
    \end{align*}
    the first by \Cref{ass:u_fn}(i) and $F' \geq 0$; the second because $x_j \leq A$ makes every comparison term vanish, while $A > x^a(y_{HR}) \geq x^a(y_i)$ and $u_x(\cdot , y_i)$ is strictly decreasing through zero at $x^a(y_i)$. So the root exists, lies in $(0,A)$, and equals $BR_i(x_{-i})$.

    \textbf{Step 3: $BR$ is a $\zeta$-contraction in the supremum norm.}
    Fix $x , x' \in [0,A]^n$, write $\delta := \| x - x' \|_\infty$, and fix an agent $i$. Relabeling if necessary, suppose $BR_i(x'_{-i}) \geq BR_i(x_{-i})$, and set $t := BR_i(x_{-i})$ and $h := BR_i(x'_{-i}) - t \geq 0$. By Step 2 both are roots, so
    \begin{align}\label{eq:contraction_difference}
        0 \ = \ \Psi_i(t + h , x'_{-i}) \ - \ \Psi_i(t , x_{-i})
        \ = \ \big[ u_x(t+h , y_i) - u_x(t , y_i) \big]
        \ + \sum_{j \neq i} G_{ij} \, \Delta_j ,
    \end{align}
    where $\Delta_j := F'\big( (x'_j - t - h)^+ \big) - F'\big( (x_j - t)^+ \big)$ and $(\cdot)^+ := \max\{ \cdot \, , 0 \}$. The first bracket is at most $- \underline{\sigma}_u \, h$. For the second, $x'_j \leq x_j + \delta$ and $F'$ is increasing, so
    \begin{align*}
        \Delta_j \ \leq \ F'\big( (x_j - t + \delta - h)^+ \big) - F'\big( (x_j - t)^+ \big) \ \leq \ \bar{\sigma}_F \, \max\{ \delta - h , \, 0 \} ,
    \end{align*}
    the last step because $r \mapsto r^+$ is $1$-Lipschitz and $F'$ is $\bar{\sigma}_F$-Lipschitz on $[0,A]$, and because the left-hand difference is non-positive when $\delta \leq h$. Substituting both bounds into \Cref{eq:contraction_difference}, and using $\sum_{j \neq i} G_{ij} = d$,
    \begin{align*}
        0 \ \leq \ - \, \underline{\sigma}_u \, h \ + \ d \, \bar{\sigma}_F \, \max\{ \delta - h , \, 0 \} .
    \end{align*}
    Were $h > \delta$ the right-hand side would equal $- \underline{\sigma}_u h < 0$, so $h \leq \delta$; and then $0 \leq - \underline{\sigma}_u h + d \bar{\sigma}_F (\delta - h)$, which rearranges to $(\underline{\sigma}_u + d \bar{\sigma}_F) h \leq d \bar{\sigma}_F \delta$. Then defining $\zeta:=\frac{d\bar{\sigma}_F}{\sigma_u+d\bar{\sigma}_F}$ yields $h\leq\zeta\delta$. And note that since $\sigma_u>0$, $d>0$, $\bar{\sigma}_F>0$ by assumption, we have $\zeta\in(0,1)$. As $i$ was arbitrary, $\| BR(x) - BR(x') \|_\infty \leq \zeta \, \| x - x' \|_\infty$.

    \textbf{Step 4: convergence, and the limit.}
    The set $[0,A]^n$ with the supremum norm is a complete metric space, so by the Banach fixed point theorem $BR$ has a unique fixed point $\bar{x}$, and $\| x^t - \bar{x} \|_\infty \leq \zeta^{\,t} \| x^0 - \bar{x} \|_\infty$. Now $\bar{x}$ is a fixed point of $BR$ exactly when every agent is best responding, and \Cref{lem:cutoff_br} gives this two readings at once. Read against \emph{levels}, $\Pi_i$ is the `main' model payoff, so $\bar{x}$ is a Nash equilibrium of that model and $\bar{x} = x^*$ by \Cref{lem:A0_uniqueness}. Read against \emph{bundles}, $\big[ 0 , \bar{x}_i \big]$ is $i$'s unique best response among \emph{all} measurable subsets of $[0,A]$, so $X^*_i = [0 , x^*_i]$ is a Nash equilibrium of the enriched model.

    \textbf{Step 5: the convergence is uniform.}
    Both $x^0$ and $x^*$ lie in $[0,A]^n$, so $\| x^0 - x^* \|_\infty \leq A$ and hence $\| x^t - x^* \|_\infty \leq \zeta^{\,t} A$ for \emph{every} cut-off starting profile. Given $\varepsilon > 0$, any date $t \geq \log(\varepsilon / A) / \log \zeta$ therefore places every agent within $\varepsilon$ of $x^*_i$, whatever the starting profile: the bound depends on neither $x^0$ nor $i$. \hfill \qed

\newpage
\section{The empirical relationship: robustness checks}\label{OA:empirical_relationship}

\subsection{Splitting the data}
In \Cref{sec:data}, we ran a quadratic regression on the sample of all US counties. In \Cref{tab:OA1}, we consider the top tercile (columns I-III) and the bottom tercile (columns IV-VI) separately. \Cref{tab:OA2} then reports top and bottom quartiles (columns I-II), and top and bottom halves (columns III-IV). 
All columns in \Cref{tab:OA2} exclude the mobility-integration interaction term.

\begin{table}[htbp]\centering
\def\sym#1{\ifmmode^{#1}\else\(^{#1}\)\fi}
\caption{WELL-BEING AND ECONOMIC MOBILITY -- REGRESSION RESULTS FOR TOP AND BOTTOM THIRD OF US COUNTIES}\label{tab:OA1}
\begin{tabular}{l*{6}{c}}
\toprule
Dep. var.: well-being
& \multicolumn{3}{c}{Top $1/3^{rd}$} & \multicolumn{3}{c}{Bottom $1/3^{rd}$} \\
          & (I) & (II) & (III) & (IV) & (V) & (VI) \\
\hline 
Mobility  &   -2.504\sym{**} &   -5.320         &   -1.461         &    2.802\sym{**} &    3.391         &    3.864\sym{***}\\
          &  (1.028)         &  (4.334)         &  (1.199)         &  (1.171)         &  (3.415)         &  (1.280)         \\
Integration&   -0.561\sym{**} &   -1.982         &                  &   -0.561\sym{***}&   -0.253         &                  \\
          &  (0.255)         &  (2.139)         &                  &  (0.215)         &  (1.690)         &                  \\
Mobility $\times$ Integration&                  &    2.921         &   -1.096\sym{**} &                  &   -0.839         &   -1.516\sym{***}\\
          &                  &  (4.368)         &  (0.522)         &                  &  (4.569)         &  (0.581)         \\
Median Income&    0.185\sym{***}&    0.186\sym{***}&    0.182\sym{***}&    0.004         &    0.005         &    0.006         \\
          &  (0.049)         &  (0.049)         &  (0.049)         &  (0.054)         &  (0.054)         &  (0.054)         \\
\hline
Other controls & YES & YES & YES & YES & YES & YES \\
\hline
R-squared & 0.491 & 0.491 & 0.490 & 0.532 & 0.532 & 0.532 \\
Observations&  829 &  829 &  829 &  927 &  927 &  927 \\
\bottomrule
\end{tabular}
\begin{center}{\parbox[b]{15cm}{\footnotesize \emph{Notes:} Standard errors in parentheses. \sym{*} \(p<0.10\), \sym{**} \(p<0.05\), \sym{***} \(p<0.01\). Other controls = State FE, Population, Religion, and Social capital controls (as in \Cref{tab:1}).}} 
\end{center}
\end{table}

\begin{table}[htbp]\centering
\def\sym#1{\ifmmode^{#1}\else\(^{#1}\)\fi}
\caption{WELL-BEING AND ECONOMIC MOBILITY -- ALTERNATIVE WAYS OF SPLITTING THE SAMPLE}\label{tab:OA2}
\begin{tabular}{l*{4}{c}}
\toprule
Dep. var.: well-being
& $Q1$ & $Q4$ & $H1$ & $H2$   \\
          & (I) & (II) & (III) & (IV) \\
\hline
Mobility  &    4.174\sym{***}&   -2.770\sym{**} &    3.254\sym{***}&   -1.472\sym{*}          \\
          &  (1.460)         &  (1.298)         &  (0.858)         &  (0.789)           \\
Integration &   -0.704\sym{***}&   -0.234         &   -0.543\sym{***}&   -0.750\sym{***}  \\
          &  (0.231)         &  (0.319)         &  (0.176)         &  (0.198)            \\
Median Income&    0.029         &    0.128\sym{**} &    0.044         &    0.157\sym{***} \\
          &  (0.061)         &  (0.060)         &  (0.042)         &  (0.040)                  \\
\hline
Other controls & YES & YES & YES & YES \\
\hline
R-squared & 0.544         &    0.445         &    0.514         &    0.483           \\
Observations&  680 &  591 & 1418 & 1310          \\
\bottomrule
\end{tabular}
\begin{center}{\parbox[b]{12cm}{\footnotesize \emph{Notes:} Standard errors in parentheses. \sym{*} \(p<0.10\), \sym{**} \(p<0.05\), \sym{***} \(p<0.01\). $Q=$ quartile, $H=$ half. $1$ denotes lowest quantile. Other controls = State FE, Population, Religion, and Social capital controls (as in \Cref{tab:1}).}} 
\end{center}
\end{table}

\subsection{Trimming the data}
While visual inspection of the data (\Cref{fig:histograms}) does not suggest that either the well-being or the social mobility data contains significant outliers, it is still helpful to confirm that the `inverse-U' shaped relationship is not being driven by the extremes. \Cref{tab:OA3} removes the top and bottom 5\% (column I) and top and bottom 10\% (column II) of observations by value of the well-being variable. It also removes the top and bottom 5\% (column III) and top and bottom 10\% (column IV) of observations by value of the social mobility variable. 

\begin{figure}[htbp]
    \centering
    \begin{subfigure}[b]{0.95\textwidth}
        \centering
        \includegraphics[width=\linewidth]{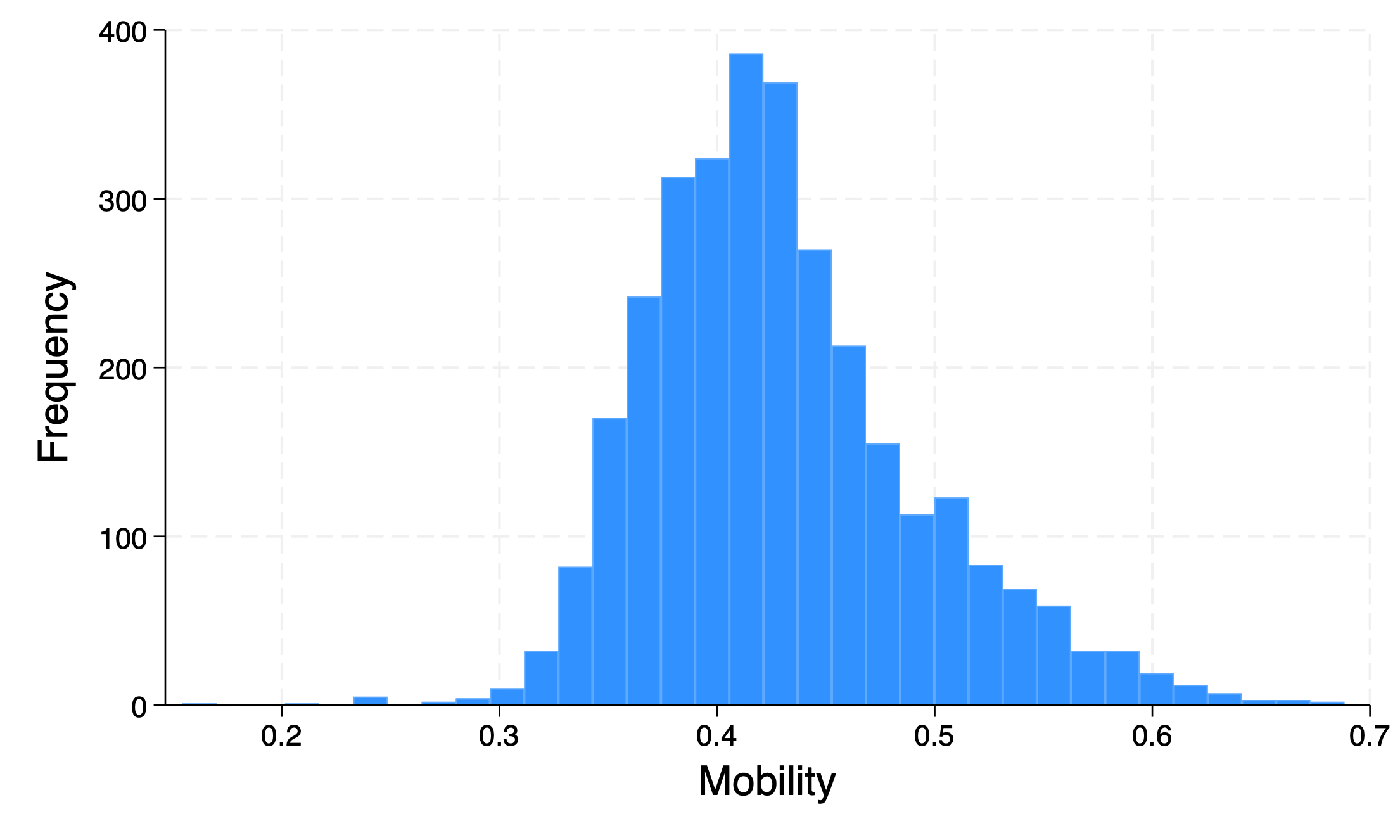}
    \end{subfigure}
    \hfill
    \begin{subfigure}[b]{0.95\textwidth}
        \centering
        \includegraphics[width=\linewidth]{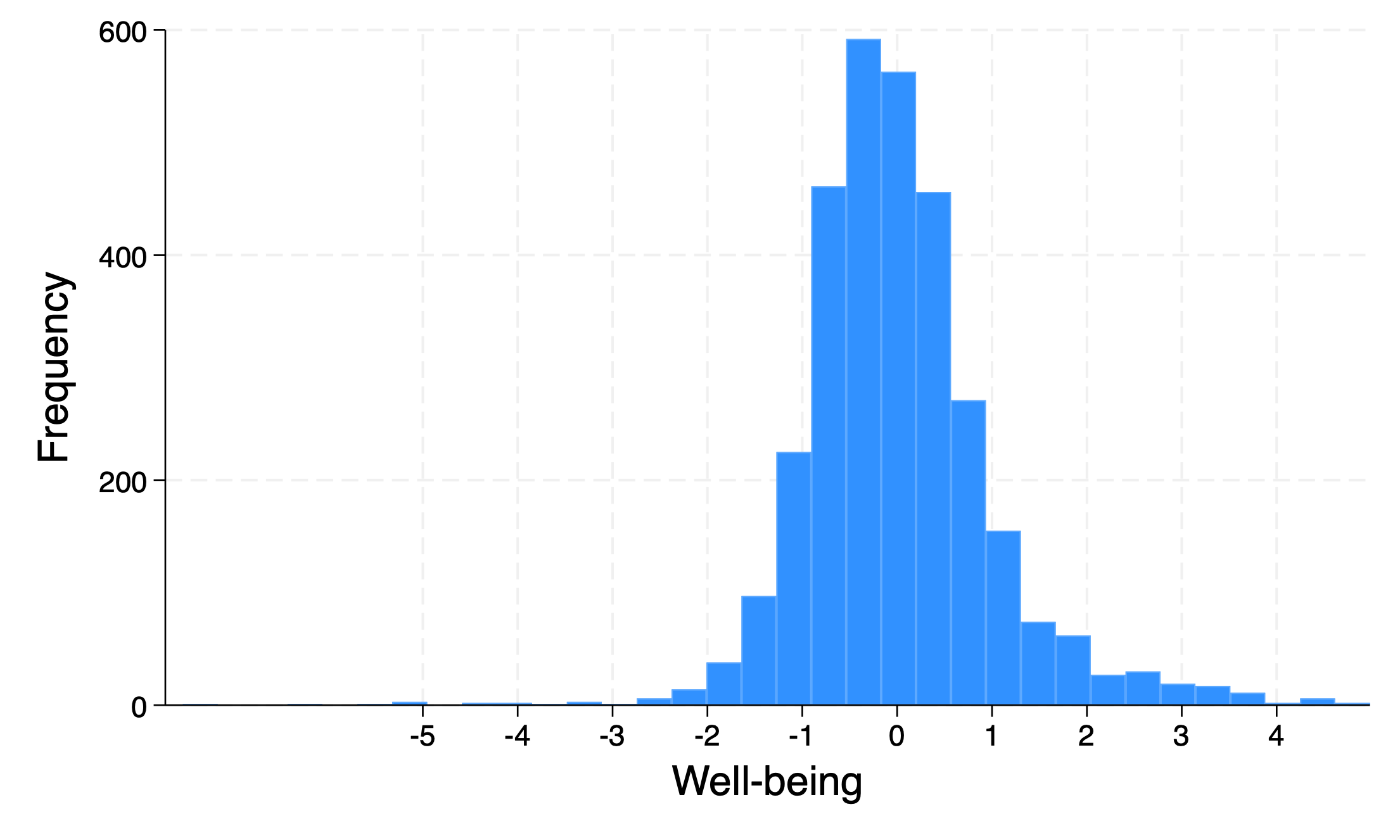}
    \end{subfigure}
    \caption{Histograms of values of mobility (top panel) and well-being (happiness, bottom panel) at US county level.}
    \label{fig:histograms}
\end{figure}

\begin{table}[htbp]\centering
\def\sym#1{\ifmmode^{#1}\else\(^{#1}\)\fi}
\caption{WELL-BEING AND ECONOMIC MOBILITY -- REGRESSION RESULTS AFTER TRIMMING UPPER AND LOWER TAILS}\label{tab:OA3}
\begin{tabular}{l*{4}{c}}
\toprule
Dep. var.: well-being
& $5\%$ happiness & $10\%$ happiness & $5\%$ mobility & $10\%$ mobility  \\
          & (I) & (II) & (III) & (IV) \\
\hline 
Mobility  &   16.137\sym{***}&   15.773\sym{***}&  25.253\sym{***}                &   25.854\sym{***}               \\
          &  (2.615)         &  (2.438)         &  (5.441)                &  (8.037)                \\
Mobility$^{2}$ &  -17.287\sym{***}&  -17.344\sym{***}&  -27.339\sym{***}&  -27.908\sym{***}\\
          &  (2.915)         &  (2.714)         &  (6.169)         &  (9.268)         \\
Integration&   -0.483\sym{***}&   -0.339\sym{***}&   -0.739\sym{***}&   -0.698\sym{***}\\
          &  (0.103)         &  (0.094)         &  (0.136)         &  (0.143)         \\
Median Income&    0.078\sym{***}&    0.072\sym{***}&    0.150\sym{***}&    0.142\sym{***}\\
          &  (0.022)         &  (0.020)         &  (0.029)         &  (0.030)         \\
\hline
Other controls & YES & YES & YES & YES \\
\hline
R-squared &    0.464         &    0.407         &    0.491         &    0.495         \\
Observations& 2526         & 2268  & 2539   & 2294      \\
\bottomrule
\end{tabular}
\begin{center}{\parbox[b]{15cm}{\footnotesize \emph{Notes:} Standard errors in parentheses. \sym{*} \(p<0.10\), \sym{**} \(p<0.05\), \sym{***} \(p<0.01\). Other controls = State FE, Population, Religion, and Social capital controls (as in \Cref{tab:1}). }} 
\end{center}
\end{table}

\newpage
\subsection{Alternative Measures of Well-being}
In \Cref{sec:data}, we used a measure of happiness from the \emph{Human Flourishing Geographical Index} \citep{iacus2025human} as our proxy for well-being. Happiness is perhaps the most direct proxy for the utility outcomes in our model. It is also the measure in the \emph{Human Flourishing Geographical Index} that has the most data.\footnote{\cite{iacus2025human} note that over $80\%$ of tweets in their dataset are judged to be `salient' for happiness (i.e. contain sentiments -- whether positive or negative -- that are related to happiness). The next most salient measure is `optimism'.}
Here, we use measures of life satisfaction and of optimism as two alternatives. After happiness, life satisfaction is perhaps conceptually closest to our notion of utility. Optimism is the variable with the second most `salient' data in the \emph{Human Flourishing Geographical Index} (after happiness). Further, we separately add income inequality as a control variable, to confirm that the relationship between well-being and mobility is not just a reflection of the relationship between well-being and inequality. \Cref{tab:OA4} reports the results.

\begin{table}[htbp]\centering
\def\sym#1{\ifmmode^{#1}\else\(^{#1}\)\fi}
\caption{WELL-BEING AND ECONOMIC MOBILITY -- INCLUDING INEQUALITY AND ALTERNATIVE MEASURES OF WELL-BEING}\label{tab:OA4}
\begin{tabular}{l*{6}{c}}
\toprule
Dep. var.:
& Happiness & Happiness & Life Sat. & Life Sat. & Optimism & Optimism  \\
          & (I) & (II) & (III) & (IV) & (V) & (VI) \\
\hline 
Mobility  &   21.668\sym{***}&   22.026\sym{***}&   11.601\sym{***}&   11.928\sym{***}&   20.029\sym{***}&   20.703\sym{***}\\
          &  (3.217)         &  (3.219)         &  (3.035)         &  (3.040)         &  (3.353)         &  (3.354)         \\
Mobility$^{2}$&  -23.547\sym{***}&  -26.985\sym{***}&  -10.939\sym{***}&  -13.468\sym{***}&  -22.769\sym{***}&  -27.981\sym{***}\\
          &  (3.598)         &  (3.904)         &  (3.400)         &  (3.709)         &  (3.757)         &  (4.092)         \\
Integration&   -0.766\sym{***}&   -2.110\sym{***}&   -0.669\sym{***}&   -1.632\sym{***}&   -0.698\sym{***}&   -2.683\sym{***}\\
          &  (0.131)         &  (0.609)         &  (0.123)         &  (0.579)         &  (0.136)         &  (0.638)         \\
Mobility $\times$ Integration &                  &    3.179\sym{**} &                  &    2.273\sym{*}  &                  &    4.684\sym{***}\\
          &                  &  (1.407)         &                  &  (1.334)         &                  &  (1.472)         \\
Inequality (Gini)      &    2.363\sym{***}&    2.252\sym{***}&                  &                  &                  &                  \\
          &  (0.495)         &  (0.497)         &                  &                  &                  &                  \\
Median Income&    0.178\sym{***}&    0.169\sym{***}&    0.106\sym{***}&    0.101\sym{***}&    0.135\sym{***}&    0.123\sym{***}\\
          &  (0.029)         &  (0.029)         &  (0.026)         &  (0.026)         &  (0.029)         &  (0.029)         \\
\hline
Other controls & YES & YES & YES & YES  & YES & YES \\
\hline
R-squared &    0.496         &    0.497         &    0.541         &    0.542         &    0.417         &    0.419         \\
Observations& 2728 & 2728    & 2728  & 2728  & 2728   & 2728       \\
\bottomrule
\end{tabular}
\begin{center}{\parbox[b]{16cm}{\footnotesize \emph{Notes:} Standard errors in parentheses. \sym{*} \(p<0.10\), \sym{**} \(p<0.05\), \sym{***} \(p<0.01\). Other controls = State FE, Population, Religion, and Social capital controls (as in \Cref{tab:1}). }} 
\end{center}
\end{table}

\subsection{Adding Higher-order Terms}
As an additional check, we also add higher-order terms for integration, income inequality and social mobility. This leaves some individual coefficients not significantly different from zero -- especially in columns (IV) and (V) -- due to high correlation between the linear and respective higher-order terms. But joint significance tests confirm that there is a significant effect of mobility, of integration, and of inequality (Wald test, $p<0.001$). Further, the non-linear relationship between happiness and mobility is also significant (Wald test, $p<0.001$).\footnote{More precisely, the joint significance test on the quadratic and cubic mobility terms rejects the null hypothesis that both are equal to zero.}

\begin{table}[htbp]\centering
\def\sym#1{\ifmmode^{#1}\else\(^{#1}\)\fi}
\caption{WELL-BEING AND ECONOMIC MOBILITY -- ADDING HIGHER-ORDER TERMS}\label{tab:OA5}
\begin{tabular}{l*{5}{c}}
\toprule
Dep. var.: well-being
          & (I) & (II) & (III) & (IV) & (V) \\
\hline 
Mobility  &   19.563\sym{***}&   19.606\sym{***}&   21.662\sym{***}&    2.641         &    4.815         \\
          &  (3.200)         &  (3.199)         &  (3.219)         & (19.777)         & (19.720)         \\
Mobility$^{2}$&  -21.432\sym{***}&  -21.496\sym{***}&  -23.540\sym{***}&   17.108         &   14.780         \\
          &  (3.585)         &  (3.584)         &  (3.600)         & (44.594)         & (44.455)         \\
Mobility$^{3}$ &                  &                  &                  &  -28.657         &  -28.481         \\
          &                  &                  &                  & (33.050)         & (32.944)         \\
Integration&   -0.683\sym{***}&   -1.300\sym{***}&   -0.766\sym{***}&   -0.686\sym{***}&   -1.164\sym{***}\\
          &  (0.130)         &  (0.405)         &  (0.131)         &  (0.130)         &  (0.405)         \\
Integration$^{2}$&                  &    0.366         &                  &                  &    0.236         \\
          &                  &  (0.228)         &                  &                  &  (0.229)         \\
Inequality (Gini)     &                  &                  &    2.855         &                  &    2.979         \\
          &                  &                  &  (6.811)         &                  &  (6.815)         \\
Inequality$^{2}$ (Gini$^{2}$)     &                  &                  &   -0.542         &                  &   -0.746         \\
          &                  &                  &  (7.491)         &                  &  (7.496)         \\
Median Income&    0.140\sym{***}&    0.131\sym{***}&    0.179\sym{***}&    0.138\sym{***}&    0.169\sym{***}\\
          &  (0.028)         &  (0.028)         &  (0.029)         &  (0.028)         &  (0.030)         \\
\hline
Other controls & YES & YES & YES & YES & YES  \\
\hline
R-squared &    0.492         &    0.492         &    0.492         &    0.492         &    0.496         \\
Observations& 2728 & 2728 & 2728  & 2728 & 2728 \\
\bottomrule
\end{tabular}
\begin{center}{\parbox[b]{16cm}{\footnotesize \emph{Notes:} Standard errors in parentheses. \sym{*} \(p<0.10\), \sym{**} \(p<0.05\), \sym{***} \(p<0.01\). Other controls = State FE, Population, Religion, and Social capital controls (as in \Cref{tab:1}).}} 
\end{center}
\end{table}

\subsection{Leavers in the data}
In the social mobility variable, `Child' incomes are measured in 2014--15 for birth cohorts 1978--83. Income measured in 2014--15 is allocated to the counties that the relevant child grew up in, weighted by the time they spent in each county. This means the economic mobility of a county is based on people who were born around 1980 and were living in that county as a child. These people may not still live in the county at the point their income is measured (although many will). In contrast, all other variables are based on the people actually living in the county at the point the relevant information is collected.\footnote{Note that for urbanization and religion controls, our data is from 2020 -- the closest available census. But these are slower moving variables, and nothing substantive would change were we to use data from an earlier year for these variables.}

This means our economic mobility variable will pick up some aspect of \emph{geographic} mobility. Consider a case where people who grew up in county A `leave' to county B. If county B is richer [resp. poorer] than county A, then this could mechanically increase [resp. decrease] the measure of economic mobility for county A. Suppose these `leavers' retain the same within-county income rank they would have done had they stayed in county A. Loosely speaking, we can view this as a benchmark assumption that people do not become more or less socially mobile as a result of moving counties. Even under this assumption, these `leavers' will move up [resp. down] the \emph{national} income ranking -- simply by virtue of having moved to a richer [resp. poorer] county. 
If the counties that `leavers' move to, on average, `look the same' as the county that they left, then this might all wash out. 

But `leavers' disproportionately move to \emph{richer} counties (see for example, \cite{borjas1992self, kennan2011effect, kaplan2017understanding, purcell2020geographic}, and citations within). This would mean that counties with more `leavers' end up with a higher measure of economic mobility. In principle, it is not immediately clear whether this is a bias in the measure or the correct way of measuring economic mobility. But for the purposes of our model -- where we have abstracted from geography and how people might move between different locations -- it is more likely to be a bias.

Additionally, it is possible that `leavers' differ on unobservables from `stayers' (those who as adults live in the same county they grew up in) in a way that relates to their well-being. We do not have a strong view as to whether `leavers' would have systematically higher or lower well-being than `stayers'. But it would be more problematic for our argument if `leavers' had systematically higher well-being. In that case, counties with more `leavers' would have lower reported well-being -- because people more prone to high well-being have left.

Taking these two conjectures together, if both true, would lead to a downward bias in the relationship between well-being and mobility. This need not be a problem for our argument -- an overall downward bias in the relationship would not in itself erroneously generate the `inverse-U' shaped relationship we document.

The problem for our argument would arise if the counties with higher (true) mobility also have more `leavers'. Then in our data, the downward bias in the relationship between mobility and well-being would be greater for higher mobility counties. If this were sufficiently extreme, this could artificially create the downwards-sloping part of our `inverse-U' shaped relationship -- in what would otherwise be a monotonically increasing relationship between well-being and mobility.

However, high mobility counties are typically richer (see \Cref{fig:F1}). And movement is disproportionately from lower income counties to higher income ones (see citations above). We argue that, taken together, it is unlikely that a downwards bias in the relationship between mobility and well-being (to the extent that such a bias exists) is disproportionately affecting highly mobile counties. In fact, it seems more plausible that it should affect less mobile counties -- which are poorer on average and so more likely to have more `leavers' moving to richer counties.

\begin{figure}[h!]
    \centering
    \includegraphics[width=0.95\linewidth]{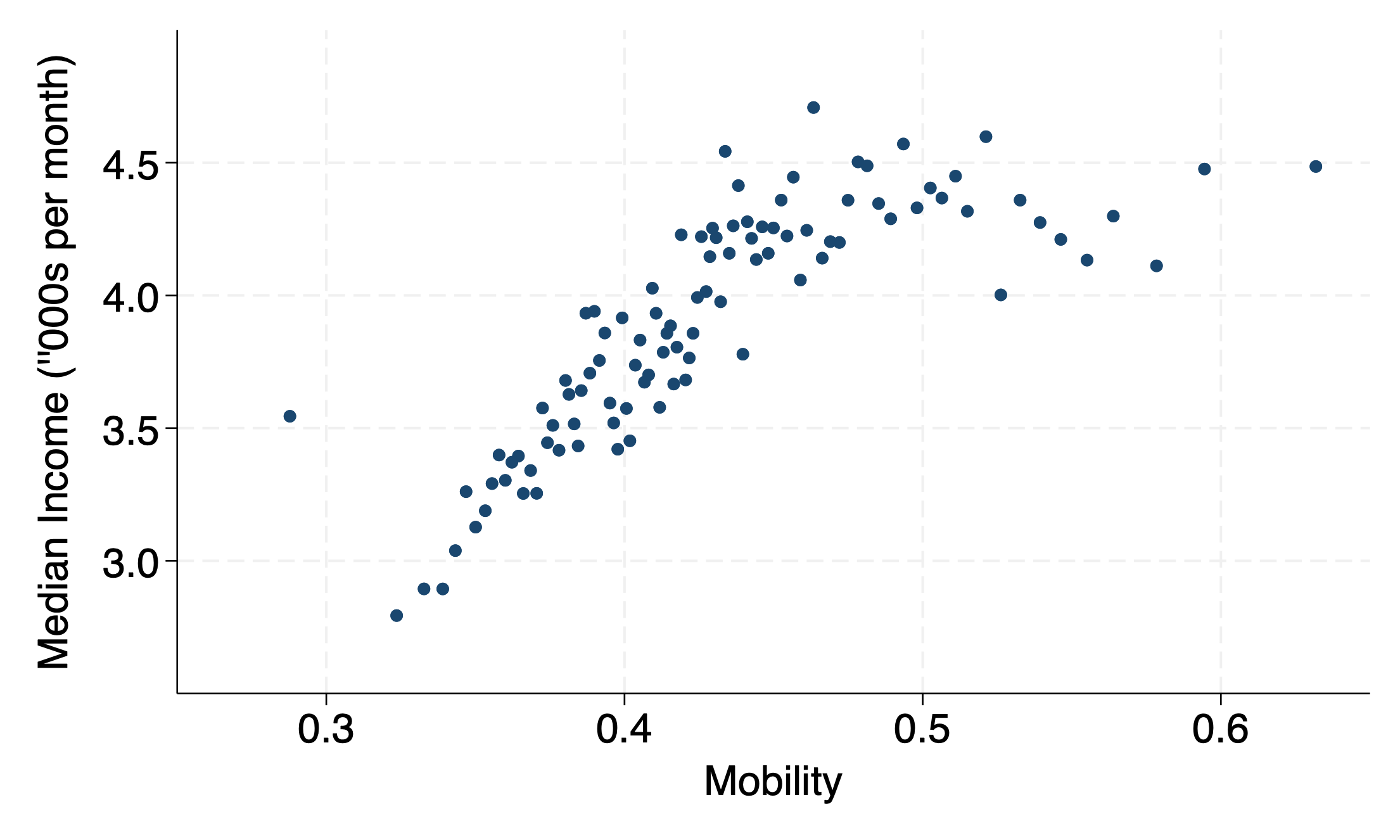}
    \caption{Median income and social mobility, binned scatter plot.}
    \label{fig:F1}
\end{figure}

\newpage
\section{The broader role of natural resources}
\label{OA:resources}

In Section~\ref{subsec:bakken_case} we argued that some counties attain high absolute upward mobility not through greater intergenerational mobility as such, but through local growth shocks that move a county upward relative to the rest of the nation, and that such shocks can carry real welfare consequences through our social-comparison mechanism. The Bakken oil boom -- a natural-resource boom that raised measured mobility sharply while creating salient local winners and losers -- is our leading illustration. This appendix places that case study in a broader context. We document county-level measures of the intensity of the two main natural resource sectors: agriculture and mining (which includes oil and gas extraction). And we describe how their geography relates to that of social mobility. The exercise is descriptive and aims to show that the kind of resource-based local economy emphasized in the Bakken case overlaps with the high-mobility regions that motivate \Cref{subsec:bakken_case}.

\subsection{Intensity of the natural resource sector}

We measure natural resource intensity using earnings by industry from the Bureau of Economic Analysis (BEA) Regional Economic Accounts (the personal income accounts). For county $c$, natural resource intensity is the agriculture and mining sectors' share of total local earnings,
\[
  \text{Resource Intensity}_{c} \;=\; \frac{\text{Earnings}_{\text{agriculture},c} + \text{Earnings}_{\text{mining},c}}{\text{Total Earnings}_{c}},
\]
where $\text{Total Earnings}_{c}$ is total earnings (across all sectors). 
We average each share over 2001 to 2015.

Our earnings measure contains the labor and proprietors' income that accrues to those working in the sector within the county. This is a natural measure of local sectoral importance: it is the income that enters residents' livelihoods, and hence the comparisons at the center of our mechanism. It is also likely more informative than a simple employment share, because it reflects how well remunerated the activity is rather than just a headcount -- a distinction that matters for a high-wage sector such as oil and gas, whose local weight an employment count would understate.
\footnote{Because the BEA withholds earnings for cells that would disclose individual establishments, the mining-earnings share is missing for some counties; these appear in light gray.}

\subsection{Geographic distribution}
 
\Cref{fig:primary_map} maps this measure across US counties, and \Cref{fig:mobility_map} maps our measure of social mobility \citep{chetty2018impacts, chetty2014land} on a common projection.
 
Primary-sector intensity is greatest in two broad areas. The first is the \emph{Great Plains}, where a broad and largely contiguous band of counties running from the Dakotas south through Nebraska and Kansas derives a substantial share of local earnings from agriculture. Corn is the dominant crop in these areas -- which therefore benefited substantially from the 2005 ethanol boom triggered by the Energy Policy Act of 2005. The act established the Renewable Fuel Standard (RFS), requiring the blending of 7.5 billion gallons of renewable fuel into the nation's gasoline supply by 2012, sparking a massive increase in the demand for corn. The second is the country's main oil- and gas-producing basins, most prominently the Permian Basin of western Texas and southeastern New Mexico and the Bakken formation of western North Dakota and eastern Montana. 
 
Comparing \Cref{fig:primary_map} with the mobility map in \Cref{fig:mobility_map} suggests spatial overlap between primary-sector intensity and social mobility. Much of the high-mobility region lies across the northern Great Plains, where both agricultural and -- in the Bakken sub-region -- mining intensity are high. This is consistent with the idea that resource-based local economies are disproportionately represented among high-mobility counties.
 
\begin{figure}[t]
\centering
\includegraphics[width=0.82\textwidth]{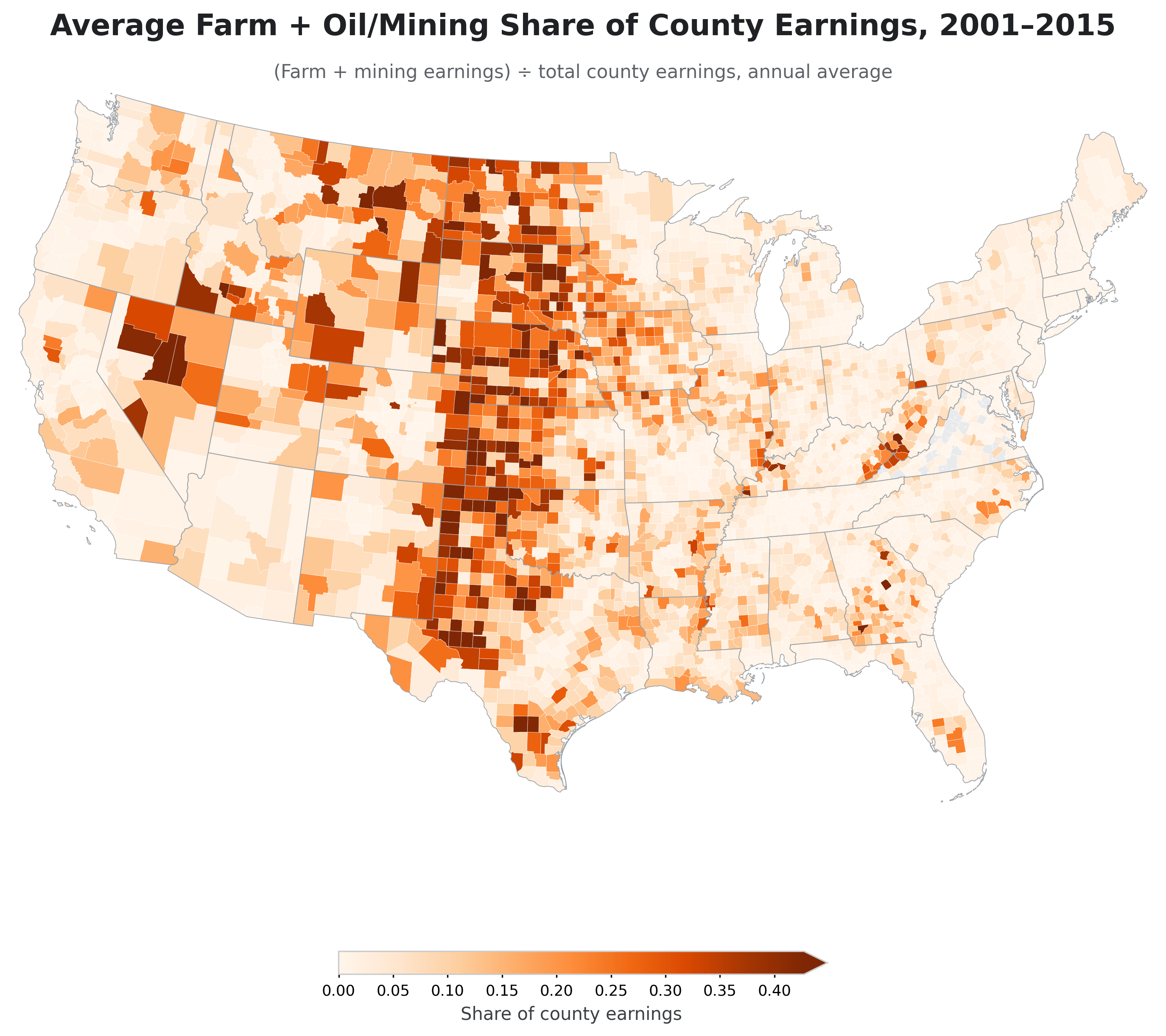}
\caption{County-level total primary-sector intensity. Counties are shaded by the combined earnings of agriculture (the BEA farm sector) and mining, quarrying, and oil and gas extraction, expressed as a share of total county earnings and averaged over 2001--2015. Darker shades denote higher shares. Shading is winsorized at the 99th percentile. Counties in light gray have missing data. The map covers the contiguous United States in an Albers equal-area projection. Source: BEA Regional Economic Accounts, earnings by
industry.}\label{fig:primary_map}
\end{figure}
 
\begin{figure}[t]
\centering
\includegraphics[width=0.82\textwidth]{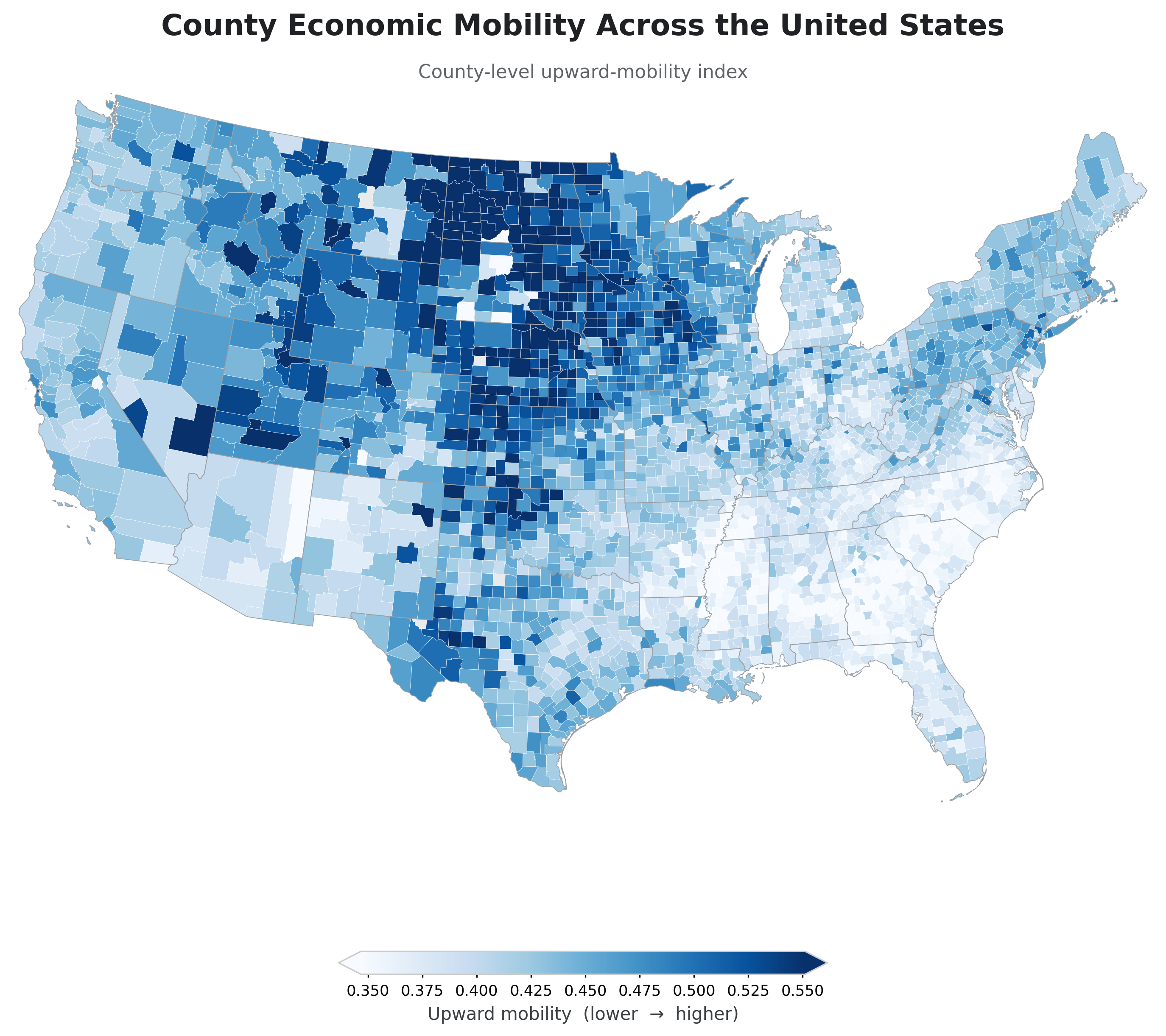}
\caption{County-level social mobility. Counties are shaded by absolute upward mobility from the Opportunity Atlas \citep{chetty2014land,chetty2018impacts} -- the predicted adult income rank of children born to parents at the 25th percentile of the national income distribution -- with darker shades denoting higher mobility. Counties in light gray have missing data. The map covers the contiguous United States in the same projection as \Cref{fig:primary_map}.}\label{fig:mobility_map}
\end{figure}

\end{document}